\documentclass[10pt, conference]{IEEEtran}
\IEEEoverridecommandlockouts
\usepackage{cite}
\usepackage{algorithmic}
\usepackage{textcomp}
\usepackage{xcolor}
\def\BibTeX{{\rm B\kern-.05em{\sc i\kern-.025em b}\kern-.08em
    T\kern-.1667em\lower.7ex\hbox{E}\kern-.125emX}}
    
\usepackage{graphicx}
\usepackage[utf8]{inputenc}
\usepackage{amsmath}
\usepackage{amssymb}
\usepackage{stmaryrd}
\usepackage{tikz}
\usepackage{multirow}
\usetikzlibrary{arrows,positioning,shapes.misc}
\usepackage{comment}
\usepackage{versions}
\usepackage{algorithm2e}
\usepackage{ upgreek }
\usepackage{wasysym}
\usepackage[normalem]{ulem}
\usepackage{hyperref}
\usepackage[font=small,labelfont=bf,tableposition=top]{caption}
\DeclareCaptionLabelFormat{andtable}{#1~#2  \&  \tablename~\thetable}
\usepackage{enumitem}

\newtheorem{proposition}{Proposition}
\newtheorem{proof}{Proof}
\newtheorem{theorem}{Theorem}
\newtheorem{example}{Example}

\newcommand\scalemath[2]{\scalebox{#1}{\mbox{\ensuremath{\displaystyle #2}}}}

 \excludeversion{ICALP} 
 \includeversion{ARXIV}

\newcommand{\at}{\mathit{at}}

\newcommand{\DNK}{\textnormal{DyNetKAT}}
\newcommand{\tool}{DyNetiKAT}
\newcommand{\DNA}{\DNK}

\newcommand{\ENK}{E_{\textit{NK}}}
\newcommand{\EDNK}{E_{\textit{DNK}}}
\newcommand{\EDNKtr}{E_{\textit{DNK}}^{\textit{tr}}}

\newcommand{\EDNKnoAIP}{E_{\DNK \setminus AIP}}

\newcommand{\NetKATnoDup}{\textnormal{NetKAT}^{-{\bf{dup}}}}

\newcommand{\drop}{\mathbf{0}}

\newcommand{\hnf}[1]{{\textit{hnf}(#1)}}

\newcommand{\cpolX}{cpol_{X}}

\newcommand{\cpolseqsucc}{cpol_{\_\Seq}^{\checkmark}}

\newcommand{\cpoloplusl}{cpol_{\_\oplus}}
\newcommand{\cpoloplusr}{cpol_{\oplus\_}}

\newcommand{\rec}[1]{\mathbf{rcfg(#1)}}
\newcommand{\recp}[1]{\mathbf{rcfg}_{#1}}

\newcommand{\cpolmsgsend}{cpol_!}
\newcommand{\cpolmsgrec}{cpol_?}
\newcommand{\cpolmsg}{cpol_{\bullet}}

\newcommand{\reconfig}{cpol_{\clubsuit \spadesuit}}
\newcommand{\reconfigsr}{cpol_{!?}}
\newcommand{\reconfigrs}{cpol_{?!}}

\newcommand{\intl}{cpol_{\_||}}
\newcommand{\intr}{cpol_{||\_}}

\newcommand{\packcopy}{\mathbf{1}}
\newcommand{\zeroq}{\bot}

\newcommand{\Par}{\mathop{||}}
\newcommand{\Seq}{\mathop{;}}

\newcommand{\xgtrans}[2]{\,\,{{\xrightarrow{#2}}_{#1}}\,\,}

\newcommand{\trans}[1]{\xgtrans{}{#1}}

\newcommand{\sosrule}[2]{\frac{\raisebox{.7ex}{\satsize{$#1$}}}
                        {\raisebox{-1.0ex}{\satsize{$#2$}}}}
\newcommand{\satsize}{\normalsize}

\newcommand{\dedr}[1]{\ensuremath{\mathbf{(#1)}}}

\newenvironment{todo}{\bigskip\hrule\medskip\noindent}{\medskip\hrule\bigskip}

\newcommand{\hctobs}[1]{\begin{todo}{{\textcolor{red}{Hunkar:}} #1}\end{todo}}

\newenvironment{caseof}{}{\vskip.5\baselineskip}
\newcommand{\case}[3]{\vskip.5\baselineskip\par\noindent {\bfseries Case (#1):} #2\\#3}

\newcommand{\myquad}[1][1]{\hspace*{#1em}\ignorespaces}
\newcommand{\pushright}[1]{\ifmeasuring@#1\else\omit\hfill$\displaystyle#1$\fi\ignorespaces}
\newcommand{\pushleft}[1]{\ifmeasuring@#1\else\omit$\displaystyle#1$\hfill\fi\ignorespaces}
\usepackage{enumitem}

\newtheorem{defn}{Definition}

\newtheorem{lem}[defn]{Lemma}

\begin{document}
\begin{ARXIV}
\onecolumn
\end{ARXIV}

\title{DyNetKAT: An Algebra of Dynamic Networks
\thanks{The work of Georgiana Caltais and H\"unkar Can Tun\c c was supported by the DFG project “CRENKAT”, proj. no. 398056821.
The work of Mohammad Reza Mousavi was supported by the UKRI Trustworthy Autonomous Systems Node in Verifiability, Grant Award Reference EP/V026801/1.}
}

\author{\IEEEauthorblockN{Georgiana Caltais, 
H\"unkar Can Tun\c c}
\IEEEauthorblockA{\textit{University of Konstanz},
Germany \\
$\{$georgiana.caltais, huenkar.tunc$\}$@uni-konstanz.de
}
\and
\IEEEauthorblockN{Hossein Hojjat}
\IEEEauthorblockA{\textit{Tehran Inst. for Advanced Studies}, Iran \\
h.hojjat@teias.institute}
\and
\IEEEauthorblockN{Mohammad Reza Mousavi}
\IEEEauthorblockA{\textit{King's College London}, UK \\
mohammad.mousavi@kcl.ac.uk}
}

\maketitle

\begin{abstract}
We introduce a formal language for specifying dynamic updates for Software Defined Networks. Our language builds upon Network Kleene Algebra with Tests (NetKAT) and adds constructs for synchronisations and multi-packet behaviour to capture the interaction between the control- and data-plane in dynamic updates. We provide a sound and ground-complete axiomatization of our language. We exploit the equational theory to provide an efficient reasoning method about safety properties for dynamic networks. We implement our equational theory in \tool\ -- a tool prototype, based on the Maude Rewriting Logic and the NetKAT tool, and apply it to a case study. We show that we can analyse the case study for networks with hundreds of switches using our initial tool prototype.  
\end{abstract}


\section{Introduction}\label{sec:introduction}


\begin{ARXIV}
Software defined networking (SDN) has gained immense popularity due to simplicity in network management and offering network programmability. 
Many programming languages have been designed for programming SDNs~\cite{TroisFBM16,LatifSLKBW20}. They range from industrial-scale, hardware-oriented and low-level programming languages such as OpenFlow~\cite{openflow} to domain-specific, high-level and programmer-centric  languages such as Frenetic~\cite{frenetic}.
In recent years, there has been a growing interest in analysable languages based on  mathematical foundations which provide a solid reasoning framework to prove correctness properties in SDNs (e.g., safety). 
\end{ARXIV}

There is a spectrum of mathematically inspired network programming languages that varies between those with a small number of language constructs and those with  expressive language design which allow them to  support more networking features.
\begin{ICALP}
Flowlog~\cite{flowlog} and Kinetic~\cite{kinetic}  are points on the more expressive side of the spectrum, which provide support for formal reasoning based on SAT-solving and model checking, respectively. 
\end{ICALP}
\begin{ARXIV}
On the more expressive side of the spectrum, Flowlog~\cite{flowlog} is an example of a language that uses a powerful formalism (first-order Horn clause logic) to program a Software Defined Network (SDN). 
In order to keep the language decidable, Flowlog disallows recursion in the clauses. For the purpose of formal analysis of a Flowlog program, the authors of ~\cite{flowlog} provide a translator to the Alloy tool.  
As another example of an expressive language, Kinetic~\cite{kinetic} is a language based on finite state machines that is mostly geared towards dynamic feature of SDNs. Model checking is used to formally analyse the Kinetic programs.
\end{ARXIV}
NetKAT~\cite{netkat,netkat-decision} is an example of a minimalist language based on Kleene algebra with tests that has a sound and complete equational theory.
While the core of the language is very simple with a few number of operators, the language has been extended in various ways to support different aspects of networking such as congestion control~\cite{prob-netkat}, history-based routing~\cite{temporal-netkat} and higher-order functions~\cite{fun-temporal-netkat}.

Our starting point is NetKAT, because it provides a clean and analysable framework for specifying SDNs. 
The minimalist design of NetKAT does not cater for  some common (failure) patterns in SDNs, particularly those arising from dynamic reconfiguration and the interaction between the data- and control-plane flows. In~\cite{McClurgHFC16}, the authors have proposed an extension to NetKAT to support stateful network updates. The extension embraces the notion of mutable state in the language which is in contrast to its pure functional nature.
The purpose of this paper is to propose an extension to NetKAT to support dynamic and stateful behaviours.
{On the one hand, we preserve the big-step denotational semantics of NetKAT-specific constructs enabling, for instance, handling flow table updates atomically, in the spirit of~\cite{DBLP:conf/sigcomm/ReitblattFRSW12}. On the other hand, we extend NetKAT in a modular fashion, to integrate concurrent SDN behaviours such as dynamic updates, defined via a small-step operational semantics.}
To this end, we pledge to keep the minimalist design of NetKAT by adding only a few new operators. Furthermore, our extension does not contradict the nature of the language.  

A number of concurrent extensions of NetKAT have been introduced to date \cite{Silva20,Wagemaker20,Kappe20}. 
These extensions followed different design decisions than the present paper and a  comparison of their approaches with ours is provided in Section \ref{sec:lang-design}; 
however, the most important difference lies in the fact that inspired by earlier abstractions in this domain \cite{DBLP:conf/sigcomm/ReitblattFRSW12}, we were committed to 
create different layers for data-plane flows and 
dynamic updates such that every data-plane packet observes a single set of flow tables through its flight through the network.  
This allowed us, unlike the earlier approaches,  to build a layer on top of NetKAT without modifying its semantics. 
Although our presentation in this paper is based on NetKAT, we envisage that our concurrency layer can be modularly (in the sense of Modular SOS \cite{Mosses04a}) used for other network programming languages in the above-mentioned spectrum. We leave a more careful investigation of the modularity on other network  languages  for future work.  


\begin{ARXIV}
\subsection{Running Examples}\label{sec:examples}
Throughout the paper, we focus on modelling with {\DNK} two examples that involve dynamically updating the network configuration. In the first example, stateful firewall, the data-plane initiates the update by allowing a disallowed path in the network as a result of requests received from the trusted intranet. 
In the second, distributed controller, the control-plane initiates the update by modifying the forwarding route of a packet in a multi-controller setting.
\end{ARXIV}

\begin{ICALP} 
\noindent {\textbf{Running Example.}}\label{sec:examples}
To illustrate our language concepts, we focus on modelling with {\DNK} an example of a stateful firewall that involves dynamically updating the flow table.
{The example is overly simplified for the purpose of presentation.} Towards the end of this paper and also in the extended version \cite{arxiv}, we treat more complex and larger-scale case studies to evaluate the applicability and analysability of our language.  
\end{ICALP}

\begin{example}\label{eg:firewall}
A firewall is supposed to protect the intranet of an organization from unauthorised access from the Internet. However, due to certain requests from the intranet, it should be able to open up connections from the Internet to intranet.  An example is when a user within the intranet requests a secure connection to a node on the Internet; in that case, the response from the node should be allowed to enter the intranet. 
The behaviour of updating the flow tables with respects to some events in the network such as receiving a specific packet is a challenging phenomenon for languages such as NetKAT. 

Figure~\ref{fig:stateful1} shows a simplified version of the stateful firewall network.
{
Note that we are not interested in the flow of packets but interested in the flow update.}
In this version, the $\it{Switch}$ does not allow any packet from the port $\it{ext}$ to $\it{int}$ at the beginning.
When the $\it{Host}$ sends a request to the $\it{Switch}$ it opens up the connection. 

\begin{figure}[ht]
\begin{center}
\begin{tikzpicture}[auto, node distance=2cm,>=latex']

    \node [draw,circle] (sw) {$\mathit{Switch}$};
    \node [draw,rounded rectangle,left=2.5cm of sw] (h) {$\mathit{Host}$};
    \node [draw,left=0cm of sw] (int) {$\mathit{int}$};
    \node [draw,right=0cm of sw] (ext) {$\mathit{ext}$};

    \draw [->] (h) -- (int);

\end{tikzpicture}

\end{center}
\caption{Stateful Firewall\label{fig:stateful1}}
\end{figure}
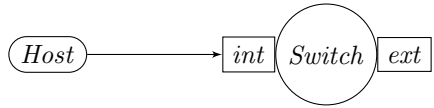
\end{example}

\begin{ARXIV}

\begin{example}\label{eg:controller}
Another running example concerns a well-known challenge in SDNs, namely, race conditions resulting from dynamic updates of flow-tables and in-flight packets \cite{McClurgHC17,Koponen14}. Below we specify a typical scenario for such race conditions; similar scenarios concerning actual bugs are abundant in the literature \cite{Koponen14,El-Hassany16,Kuzniar14}.    

Consider the network topology depicted in Figure \ref{fig:distributedController1}.  
The controller $C1$ controls the top part of the network (switches $S1$, $S3$ and $S5$) and the controller $C2$ is responsible for the bottom part.
Initially, the packets from $H1$, which enter the network through switch $S1$ (port $2$), should be routed through switches $S5$ (through ports $6$ and $7$), $S6$ (through ports $8$ and $10$) and finally port $12$ of switch $S2$, to reach $H2$. 
Due to an event, the controllers have to take down the previous route, and to install a new route in the network that routes packets from $H3$ through $S3$ (ports $1$ to $3$), $S5$ (through ports $5$ to $7$), $S6$ (ports $8$ to $9$) to switch $S4$ (port $11$) to finally reach $H4$. 
It is an important security property that the traffic in these two routes should not mix. In particular, it will be serious breach if packets from $H1$ arrive at $H4$ or vice versa, packets from $H3$ arrive at $H2$. 


\begin{figure}[h]
\centering
\begin{tikzpicture}
\tikzset{switch/.style = {shape=circle,draw,minimum size=1.5em}}
\tikzset{host/.style = {rounded rectangle=0pt,draw,minimum size=1.5em]}}
\tikzset{edge/.style = {->,> = latex}}

\draw[dashed] (-4, -2) to (5, -2);

\node[switch, 
label={[red,anchor=east, inner sep=9pt]above:1}, 
label={[red,anchor=north, inner sep=2pt]below:3}] (S3) at  (0,0) {$S3$};

\node[switch, 
label={[red,anchor=west, inner sep=9pt]above:2}, 
label={[red,anchor=north, inner sep=2pt]below:4}] (S1) at  (2,0) {$S1$};

\node[switch,
label={[red,anchor=east, inner sep=9pt]below:13}, 
label={[red,anchor=south, inner sep=2pt]above:11}] (S4) at  (0,-4) {$S4$};

\node[switch,
label={[red,anchor=west, inner sep=9pt]below:14}, 
label={[red,anchor=south, inner sep=2pt]above:12}] (S2) at  (2,-4) {$S2$};

\node[switch, fill=white, 
label={[red,anchor=west, inner sep=13pt]above:6}, 
label={[red,anchor=east, inner sep=9pt]below:7},
label={[red,anchor=east, inner sep=13pt]above:5},
] (S5) at  (1,-1.3) {$S5$};

\node[switch, fill=white, 
label={[red,anchor=west, inner sep=9pt]above:8}, 
label={[red,anchor=east, inner sep=11pt]below:9},
label={[red,anchor=west, inner sep=11pt]below:10},
] (S6) at  (1,-2.7) {$S6$};

\node[host, label={[red,anchor=west, inner sep=9pt]above:15}] (H4) at  (-2, -4) {$H4$};
\node[host, label={[red,anchor=east, inner sep=9pt]above:16}] (H2) at  (4, -4) {$H2$};
\node[host] (H1) at  (4,0) {$H1$};
\node[host] (H3) at  (-2,0) {$H3$};
\node[draw,rectangle] (C1) at  (-3,-1.2) {$C1$};
\node[draw,rectangle] (C2) at  (-3,-2.7) {$C2$};


\draw (S3) to (S5);
\draw (S6) to (S4);
\draw (S6) to (S2);
\draw (S5) to (S6);
\draw (S1) to (S5);
\draw (S2) to (H2);
\draw (H3) to (S3);
\draw (H1) to (S1);
\draw (H4) to (S4);

\end{tikzpicture}    
    
    \caption{Race Condition in a Distributed Controller}
    \label{fig:distributedController1}
\end{figure}
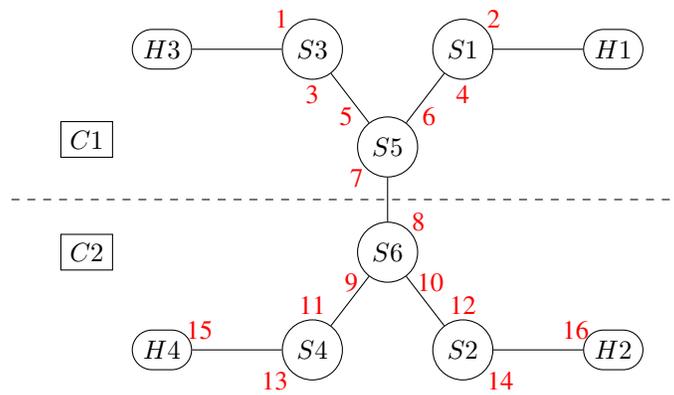
\end{example}

\end{ARXIV}

\begin{ARXIV}
\subsection{Our Contributions}
\label{sec:contributions}
\end{ARXIV}
\begin{ICALP} 
\noindent {\textbf{Our contributions.}}
\end{ICALP}
The contributions of this paper are summarized as follows: 
\begin{itemize}
    \item we define the syntax and operational semantics of a dynamic extension of NetKAT that allows for modelling and reasoning about control-plane updates and their interaction with data-plane flows ({Sections~\ref{sec:syntax},~\ref{sec:semantics}});
    \item we give a sound and ground-complete axiomatization of our language ({Section~\ref{sec:sem-results}}); and 
    \item we devise analysis methods for reasoning about flow properties using our axiomatization, apply them on examples from the domain and gather and analyze evidence of applicability and efficiency for our approach ({ Sections~\ref{sec:safety}, \ref{sec:implementation}}).
\end{itemize}

\begin{ARXIV}
\subsection{Structure of Paper}\label{sec:structure}
\end{ARXIV}

\begin{ICALP} 
\noindent {\textbf{Structure of Paper.}}
\end{ICALP}
In Section \ref{sec:lang-design}, we provide a brief overview of NetKAT, review our design decision and introduce the syntax and operational semantics of {\DNK}. In Section \ref{sec:sem-results}, we investigate some semantic properties of  {\DNK} by defining a notion of behavioural equivalence and providing a sound and ground-complete axiomatization. We exploit this axiomatization in Section \ref{sec:safety} in an analysis method. We implement and apply our analysis method in Section \ref{sec:implementation} on a case study and report about its scalability on large examples with hundreds of switches. We conclude the paper and present some avenues for future work in Section~\ref{sec:conclusions}.

\section{Language Design}\label{sec:lang-design}

In what follows, we provide a brief overview of the NetKAT syntax and semantics~\cite{netkat}. Then, we motivate our language design decisions, we introduce the syntax of {\DNK} and its underlying semantics, and provide the corresponding encoding of our running examples presented in Section~\ref{sec:examples}.

\subsection{Brief Overview of NetKAT}\label{sec:netKAT}
We proceed by first introducing some basic notions that are used throughout the paper.

\begin{defn}[Network Packets.]\label{def::config}
Let $F = \{f_1, \ldots, f_n\}$ be a set of field names $f_i$ with $i \in \{1, \ldots n\}$.
We call \emph{network packet} a function in $F \rightarrow \mathbb{N}$ that maps field names in $F$ to values in $\mathbb{N}$. We use $\sigma, \sigma'$ to range over network packets. We write, for instance, $\sigma(f_i) = v_i$ to denote a test checking whether the value of $f_i$ in $\sigma$ is $v_i$. Furthermore, we write $\sigma[f_i := n_i]$ to denote the assignment of $f_i$ to $v_i$ in $\sigma$.

A (possibly empty) \emph{list of packets} is formally defined as a function from natural numbers to packets, where the natural number in the domain denotes the position of the packet in the list such that the domain of the function forms an interval starting from $0$.

The \emph{empty list} is denoted by $\langle \rangle$ and is formally defined as the empty function (the function with the empty set as its domain). 
Let $\sigma$ be a packet and $l$ be a list, then  $\sigma::l$ is the list $l'$ in which $\sigma$ is at position $0$ in $l'$, i.e., $l'(0) = \sigma$, and $l'(i+1) = l(i)$, for all $i$ in the domain of $l$.
\end{defn}

%
In Figure~\ref{fig:NetKAT-syn-sem}, we recall the {NetKAT} syntax and semantics~\cite{netkat}.

\begin{figure}[ht]
\[
\small
\begin{array}{rcl}
 \multicolumn{3}{l}{\textnormal{\bf{NetKAT Syntax:}}} \\
 \mathit{Pr} & ::= & \drop \mid \packcopy{} \mid \mathit{Pr} + \mathit{Pr} \mid \mathit{Pr} \cdot \mathit{Pr} \mid \neg \mathit{Pr} \\
     N & ::= & \mathit{Pr} \mid f \leftarrow n \mid N + N \mid N \cdot N \mid N^* \mid \textbf{dup}
\end{array}
\]
\[
\small
\begin{array}{cc}
    \begin{array}{rcl}
   \multicolumn{3}{l}{\textnormal{\bf{NetKAT Semantics:}}} \\
   \llbracket \packcopy \rrbracket (h) & \triangleq & \{h\}\\
   \llbracket \drop \rrbracket (h) & \triangleq & \{\}\\
   \llbracket f=n \rrbracket \; (\sigma{:}{:}h) & \triangleq &\; \left\{
                \begin{array}{ll}
                  \{\sigma{:}{:}h\} & \textnormal{if}\; \sigma(f) = n \\
                  \{ \} & \textnormal{otherwise}
                \end{array}
              \right.\\
    \llbracket \neg a \rrbracket \; (h) & \triangleq &\; \{h\} \setminus\,\, \llbracket a \rrbracket \; (h) \\
    \llbracket f \leftarrow n \rrbracket \; (\sigma{:}{:}h) &\triangleq &\; \{ \sigma[f := n]{:}{:}h\} \notag \\
    \llbracket p + q \rrbracket \; (h) &\triangleq &\; \llbracket p \rrbracket \; (h) \cup \llbracket q \rrbracket \; (h)\\
    \llbracket p \cdot q \rrbracket \; (h)  &\triangleq &\; ( \llbracket p \rrbracket \bullet \llbracket q \rrbracket) \; (h) \\
    \llbracket p^* \rrbracket \; (h) &\triangleq &\; \bigcup_{i \in N} F^i\;(h)  \\
F^0\; (h) &\triangleq &\; \{ h \}\\
F^{i+1}\;(h) &\triangleq &\; (\llbracket p \rrbracket \bullet F^i)\;(h)  \\
(f \bullet g)(x) &\triangleq &\; \bigcup \{g(y) \mid y \in f(x)\}\\
\llbracket \textbf{dup} \rrbracket \; (\sigma{:}{:}h) &\triangleq &\; \{\sigma{:}{:}(\sigma{:}{:}h)\} 

 \end{array}
\end{array}
\]
\caption{NetKAT: Syntax and Semantics~\cite{netkat}}
\label{fig:NetKAT-syn-sem}
\end{figure}

The predicate for dropping a packet is denoted by $\drop{}$, while passing on a packet (without any modification) is denoted by $\packcopy{}$. The predicate checking whether the field $f$ of a packet has value $n$ is denoted by $(f = n)$; if the predicate fails on the current packet it results on dropping the packet, otherwise it will pass the packet on.  Disjunction and conjunction between predicates are denoted by  $\mathit{Pr} + \mathit{Pr}$ and $\mathit{Pr} \cdot \mathit{Pr}$, respectively. Negation is denoted by  $\neg \mathit{Pr}$. 
Predicates are the basic building blocks of NetKAT policies and hence, a predicate is a policy by definition. The policy that modifies the field $f$ of the current packet to take value $n$ is denoted by $(f \leftarrow n)$. 
A multicast behaviour of policies is denoted by $N + N$, while sequencing policies (to be applied on the same packet) are denoted by $N \cdot N$. The repeated application of a policy is encoded as $N^*$. The construct {\textbf{dup}} simply makes a copy of the current network packet.

In~\cite{netkat}, lists of packets are referred to as \emph{histories}.
Let $H$ stand for the set of packet histories, and ${\cal P}(H)$ denote the powerset of $H$. More formally, the denotational semantics of NetKAT policies is inductively defined via the semantic map $\llbracket - \rrbracket : N \rightarrow (H \rightarrow {\cal P}(H))$ in Figure~\ref{fig:NetKAT-syn-sem}, where $N$ stands for the set of NetKAT policies,
$h \in H$ is a packet history, $a \in {\textit Pr}$ denotes a NetKAT predicate and $\sigma \in F \rightarrow \mathbb{N}$ is a network packet.

\begin{ICALP}
For a reminder, the equational axioms of NetKAT include the Kleene Algebra axioms, Boolean Algebra axioms and the so-called Packet Algebra axioms that handle NetKAT networking specific constructs such as field assignments and {\textbf{dup}}. In this paper, we write $\ENK$ to denote the NetKAT axiomatization~\cite{netkat}.
\end{ICALP}

\begin{ARXIV}
For a reminder, the equational axioms of NetKAT, denoted by $\ENK$, are provided in Figure~\ref{fig::axioms-NetKAT}.
$\ENK$ includes the Kleene Algebra axioms (KA-$\ldots$), Boolean Algebra axioms (BA-$\ldots$) and Packet Algebra axioms (PA-$\ldots$). The novelty is the set of PA-axioms. In short, PA-MOD-MOD-COMM states that the order in which two different packet fields are assigned does not matter. PA-MOD-FILTER-COMM encodes a similar property, for the case of a field assignment followed by a test of a different field's value. PA-MOD-FILTER ignores the test of a field preceded by an assignment of the same value to the field. Orthogonaly, PA-FILTER-MOD ignores a field assignment preceded by a test against the assigned value. PA-MOD-MOD states that a sequence of assignments to the same field only takes into consideration the last assignment. PA-CONTRA encodes the fact that a field cannot have two different values at the same point. PA-MATCH-ALL identifies the policy accepting all the packets with the sum of all possible tests of a field's value. Intuitively, PA-DUP-FILTER-COMM states that adding the current packet to the history is independent of tests.

\begin{figure}[ht]
\center
\[
\small{\scalemath{.8}{
\begin{array}{r@{}c@{}ll|r@{}c@{}ll}
p + (q+r) & \,\equiv\, & (p+q) + r & \textnormal{\footnotesize KA-PLUS-ASSOC}~ &
~a + (b\cdot c) & \,\equiv\, &\, (a+b)\cdot(a+c) & \textnormal{\footnotesize BA-PLUS-DIST}\\
p + q& \,\equiv\, & q+p & \textnormal{\footnotesize KA-PLUS-COMM} &
a +1 & \,\equiv\, &\, 1 & \textnormal{\footnotesize BA-PLUS-ONE}\\
p + 0& \,\equiv\, &p & \textnormal{\footnotesize KA-PLUS-ZERO} &
a +\neg a & \,\equiv\, &\, 1 & \textnormal{\footnotesize BA-EXCL-MID}\\
p + p & \,\equiv\, & p & \textnormal{\footnotesize KA-PLUS-IDEM} &
a \cdot b & \,\equiv\, &\, b\cdot a & \textnormal{\footnotesize BA-SEQ-COMM}\\
p \cdot (q\cdot r) & \,\equiv\, & (p\cdot q) \cdot r & \textnormal{\footnotesize KA-SEQ-ASSOC}~ &
a \cdot \neg a & \,\equiv\, &\, 0 & \textnormal{\footnotesize BA-CONTRA}\\
1\cdot p & \,\equiv\, & p & \textnormal{\footnotesize KA-ONE-SEQ}~ &
a \cdot a & \,\equiv\, &\,a & \textnormal{\footnotesize BA-SEQ-IDEM}\\
p\cdot 1 & \,\equiv\, & p & \textnormal{\footnotesize KA-SEQ-ONE}~ & & &  & \\
p\cdot (q+r) & \,\equiv\, & p\cdot q + p\cdot r & \textnormal{\footnotesize KA-SEQ-DIST-L}~ & f \leftarrow n \cdot f' \leftarrow n' & \,\equiv\, & f' \leftarrow n' \cdot f \leftarrow n, \textnormal{if } f \not = f' ~~& \textnormal{\footnotesize PA-MOD-MOD-COMM}\\
(p+q)\cdot r & \,\equiv\, & p\cdot r + q\cdot r & \textnormal{\footnotesize KA-SEQ-DIST-R}~ & f \leftarrow n \cdot f' = n' & \,\equiv\, & f' = n' \cdot f \leftarrow n, \textnormal{if } f \not = f' & \textnormal{\footnotesize PA-MOD-FILTER-COMM} \\
0\cdot p & \,\equiv\, & 0 & \textnormal{\footnotesize KA-ZERO-SEQ}~ & {\bf dup} \cdot f = n & \,\equiv\, & f = n \cdot {\bf dup} & \textnormal{\footnotesize PA-DUP-FILTER-COMM} \\
p\cdot 0 & \,\equiv\, & 0 & \textnormal{\footnotesize KA-ZERO-SEQ}~ & f \leftarrow n \cdot f = n & \,\equiv\, & f \leftarrow n & \textnormal{\footnotesize PA-MOD-FILTER} \\
1 + p\cdot p^* & \,\equiv\, & p^* & \textnormal{\footnotesize KA-UNROLL-L}~ & f = n \cdot f \leftarrow n & \,\equiv\, & f = n & \textnormal{\footnotesize PA-FILTER-MOD} \\
1 + p^*\cdot p & \,\equiv\, & p^* & \textnormal{\footnotesize KA-UNROLL-R}~ & f \leftarrow n \cdot f \leftarrow n' & \,\equiv\, & f \leftarrow n' & \textnormal{\footnotesize PA-MOD-MOD} \\
q+p\cdot r \leq r  & \,\Rightarrow\, & p^* \cdot q\leq r & \textnormal{\footnotesize KA-LFP-L}~ & f = n \cdot f = n' & \,\equiv\, &\,0, \textnormal{if } n \not= n'& \textnormal{\footnotesize PA-CONTRA} \\
p+q\cdot r \leq q  & \,\Rightarrow\, & p\cdot r^*\leq q & \textnormal{\footnotesize KA-LFP-R}~ & \Sigma_i f = i & \,\equiv\, & \,1 & \textnormal{\footnotesize PA-MATCH-ALL} 
\end{array}
}}
\]
\caption{$\ENK$: NetKAT Equational Axioms~\cite{netkat}} \label{fig::axioms-NetKAT} 
\end{figure}
\end{ARXIV}

\subsection{Design Decisions}\label{sec:design-decisions}
Our main motivation behind {\DNA} is to have a \emph{minimalistic} language that can model \emph{control-plane} and \emph{data-plane} network traffic and their interaction. 
Our choice for a minimal language is motivated by our desire to use our language as a basis for scalable analysis. 
We would like to be able to compile major practical languages into ours. 
Our minimal design helps us reuse much of the well-known scalable analysis techniques.   Regarding its modelling capabilities, we are interested in modelling the stateful and dynamic behaviour of networks emerging from these interactions.  
We would like to be able to model control messages, connections between controllers and switches, data packets, links among switches, and model and analyse their interaction in a seamless manner.

Based on these motivations, we start off with NetKAT as a fundamental and minimal network programming language, which allows us to model the basic policies governing the network traffic. 
The choice of NetKAT, in addition to its minimalist nature, is motivated by its rigorous semantics and equational theory, and the existing techniques and tools for its analysis. 
This motivates our next design constraint, namely, to build upon NetKAT in a hierarchical manner and without redefining its semantics. 
This constraint should not be taken lightly as the challenges in the recent concurrent extensions of NetKAT demonstrated \cite{Silva20,Wagemaker20,Kappe20}. 
We will elaborate on this point, in the presentation of our syntax and semantics. We can achieve this thanks to the abstractions introduced in the domain  \cite{DBLP:conf/sigcomm/ReitblattFRSW12} that allows for a neat layering of data-plane and control-plan flows such that every data-plane flow sees one set of flow-tables in its flight through the network.

We then introduce a few extensions and modifications to cater for the phenomena we desire to model in our extension regarding control-plane and dynamic and stateful behaviour: 

\begin{itemize}
    \item Synchronisation: we introduce a basic mechanism of handshake synchronisation with the possibility of communicating a network program (a flow table). 
    This construct allows for capturing the dynamicity and interaction between the control and data planes.
    
    \item Guarded recursion: we introduce the concept of recursion to model the (persistent) dynamic changes that result from control messages and stateful behaviour.
    In other words, recursion is used to model the new state of the flow tables. 
    An alternative modelling construct could have been using ``global'' variables and guards, but we prefer recursion due to its neat algebraic representation. 
    We restrict the use of recursion to guarded recursion, that is a policy should be applied before changing state to a new recursive definition, in order to remain within a decidable and analyse-able realm.
    A natural extension of our framework could introduce formal parameters and parameterised recursive variables; this future extension is orthogonal to our existing extensions and in this paper, we go for a minimal extension in which the parameters are coded in variable names. 
    
    \item Multi-packet semantics: we introduce the semantics of treating a list of packets, which is essential for studying the interaction between control- and data plane packets. This is in contrast with NetKAT where a single-packet semantics is introduced. The introduction of multi-packet semantics also called for a new operator to denote the end of applying a flow-table to the current packet and proceeding with the next packet (possibly with the modified flow-table in place). This is our new sequential composition operator, denoted by ``$;$''. 
\end{itemize}

\subsection{DyNetKAT Syntax}\label{sec:syntax} 

As already mentioned, NetKAT provides the possibility of recording the individual ``hops'' that packets take as they go through the network by using the so-called \textbf{dup} construct. The latter keeps track of the state of the packet at each intermediate hop. As a brief reminder of the approach in~\cite{netkat}: assume a NetKAT switch policy $p$ and a topology $t$, together with an ingress ${\it in}$ and an egress ${\it out}$. Checking whether ${\it out}$ is reachable from ${\it in}$ reduces to checking:
$
{\it in} \cdot \textbf{dup} \cdot (p \cdot t \cdot \textbf{dup})^* \cdot {\it out} \not \equiv \drop
$ (see Definition $2$ and Theorem $4$ in~\cite{netkat}). Furthermore, as shown in~\cite{netkat-decision}, \textbf{dup} plays a crucial role in devising the NetKAT language semantics in a coalgebraic fashion, via Brzozowski-like derivatives on top of NetKAT coalgebras (or NetKAT automata) corresponding to NetKAT expressions.

We decided to depart from NetKAT in this respect, due to our important constraint not to redefine the NetKAT semantics: the \textbf{dup} expression allows for observable intermediate steps that result from incomplete application of flow-tables and in concurrency scenarios, the same data packet may become subject to more than one flow table due to the concurrent interactions  with the control plane. For this semantics to be compositional, one needs to define a small step operational semantics in such a way that the small steps in predicate evaluation also become visible (see our past work on compositionality of SOS with data on such constraints \cite{MOUSAVI2005107}). This will first break our constraint in building upon NetKAT semantics and secondly, due to the huge number of possible interleavings, make the resulting state-space intractable for analysis. 


In addition to the argumentation above, note that similarly to the approach in~\cite{netkat}, we work with packet fields ranging over finite domains. Consequently, our analyses can be formulated in terms of reachability properties, further verifiable by means of \textbf{dup}-free expressions of shape:
$
{\it in} \cdot (p \cdot t )^* \cdot {\it out} \not \equiv \drop
$.
Hence, we chose to define {\DNK} synchronization, guarded recursion and multi-packet semantics on top of the \textbf{dup}-free fragment of NetKAT, denoted by {$\NetKATnoDup$}.

The syntax of {\DNA} is defined on top of the \textbf{dup}-free fragment of NetKAT as:
\begin{equation}\label{def::DNA-syntax}
\begin{array}{ccl}
     \mathit{N} & ::= & {\NetKATnoDup}\\[1ex]
     \mathit{D} & ::= &
     { \zeroq } \mid
     \mathit{N} \Seq \mathit{D}
     \mid  x?\mathit{N} \Seq D \mid 
     x!\mathit{N} \Seq D \mid \mathit{D} \Par \mathit{D} \mid
     D \oplus D \mid X  \\
     && X \triangleq D
     
\end{array}
\end{equation}

We sometimes write $p \in \textnormal{NetKAT}$, $p \in {\NetKATnoDup}$ or, respectively, $p \in \textnormal{\DNK}$ in order to refer to a NetKAT, ${\NetKATnoDup}$ or, respectively, {\DNK} policy $p$.

The {\DNK}-specific constructs are as follows.
By $\bot$ we denote a dummy policy without behaviour.
Our new sequential composition operator, denoted by $\mathit{N} \Seq \mathit{D}$, specifies when the $\NetKATnoDup$ policy $N$ applicable to the current packet has come to a successful end and, thus, the packet can be transmitted further and the next packet can be fetched for processing according to the rest of the policy $D$.

Communication in {\DNK}, encoded via $x!\mathit{N} ; D$ and $x?\mathit{N} ; D$, consists of two steps. In the first place, sending and receiving $\NetKATnoDup$ policies through channel $x$ are denoted by $x!\mathit{N}$, and $x?\mathit{N}$. Intuitively, these correspond to updating the current network configuration according to $N$. Secondly, as soon as the sending or receiving messages are successfully communicated, a new packet is fetched and processed according to $D$.
The  parallel composition of two {\DNK} policies (to enable synchronization) is denoted by $ \mathit{D} \Par \mathit{D}$.

As it will become clearer in Section~\ref{sec:semantics} (semantics), communication in {\DNK} guarantees preservation of well-defined behaviours when transitioning between network configurations. This corresponds to the so-called per-packet consistency in~\cite{DBLP:conf/sigcomm/ReitblattFRSW12}, and it guarantees that every packet traversing the network is processed according to exactly one $\NetKATnoDup$ policy.

Non-deterministic choice of {\DNK} policies is denoted by $D \oplus D$.
For a non-determinstic choice over a finite domain $P$, we use the syntactic sugar $\oplus_{p \in P} P'$, where $p$ appears as ``bound variable'' in $P'$; this is interpreted as a sum of finite summand by replacing the variable $p$ with  all its possible values in $P$.

Finally, one can use recursive variables $X$ in the specification of {\DNK} policies, where each recursive variable should have a unique defining equation $X \triangleq D$. 

For the simplicity of notation, we do not explicitly specify the trailing ``$; \bot$'' in our policy specifications, whenever clear from the context.

In Figure~\ref{fig:stateful2} we provide the {\DNK} formalization of the firewall in Example~\ref{eg:firewall}.
In the {\DNK} encoding, we use the message channel $\it{secConReq}$ to open up the connection and $\it{secConEnd}$ to close it.
We model the behavior of the switch using the two programs $\mathit{Switch}$ and $\mathit{Switch'}$.

\begin{ICALP}
\begin{figure}[h]
		\[
		\def\arraystretch{1.2}
		\footnotesize{
			\begin{array}{r@{}c@{}ll}
     \mathit{Switch} & \triangleq &   \big((port = int) \cdot (port \leftarrow  ext) \big) \Seq \mathit{Switch}  \oplus \\
                                 && \big( (port = ext) \cdot \drop \big)   \Seq \mathit{Switch} \oplus \\ 
                                 && secConReq? \packcopy  \Seq Switch' \\ \\
         
     \mathit{Switch}' & \triangleq &   \big( (port = int) \cdot (port \leftarrow  ext) \big) \Seq \mathit{Switch}'  \oplus \\ 
                                 && \big( (port = ext) \cdot (port \leftarrow  int) \big) \Seq \mathit{Switch}'  \oplus \\ 
                                 && secConEnd? \packcopy \Seq Switch\\ \\
    \mathit{Host} & \triangleq  & secConReq!\packcopy \Seq \mathit{Host} \oplus
      secConEnd!\packcopy \Seq \mathit{Host} \\ \\
      \mathit{Init} & \triangleq &   \mathit{Host} \Par \mathit{Switch}  \\
\end{array}
}
\]
\caption{Stateful Firewall in {\DNK}\label{fig:stateful2}}
\end{figure}
\end{ICALP}

\begin{ARXIV}
\begin{figure}[h]
\begin{center}
{
\[
\begin{array}{ccl}
     \mathit{Host} & \triangleq  & secConReq!\packcopy \Seq \mathit{Host} \oplus\\
     &   & secConEnd!\packcopy \Seq \mathit{Host} \\ \\
     \mathit{Switch} & \triangleq &   \big((port = int) \cdot (port \leftarrow  ext) \big) \Seq \mathit{Switch}  \oplus \\
                                 && \big( (port = ext) \cdot \drop \big)   \Seq \mathit{Switch} \oplus \\ 
                                 && secConReq? \packcopy  \Seq Switch' \\ \\
         
     \mathit{Switch}' & \triangleq &   \big( (port = int) \cdot (port \leftarrow  ext) \big) \Seq \mathit{Switch}'  \oplus \\ 
                                 && \big( (port = ext) \cdot (port \leftarrow  int) \big) \Seq \mathit{Switch}'  \oplus \\ 
                                 && secConEnd? \packcopy \Seq Switch \\    \\        
      \mathit{Init} & \triangleq &   \mathit{Host} \Par \mathit{Switch}  \\
\end{array}
\]
}
\end{center}
\caption{Stateful Firewall in {\DNK}\label{fig:stateful2}}
\end{figure}
\end{ARXIV}

\begin{ARXIV}
In Figure~\ref{fig:independent2} we provide the {\DNK} formalization of the distributed controllers in Example~\ref{eg:controller}. 
In the code in Figure~\ref{fig:independent2} the controllers work independently to update the network (which can lead to security breach).
The specification $\mathit{SwitchX}_{ft}$ is a generic specification for the behaviour of all switches in this example; the domain of $P$ in this example is the set of all $5$ policies that are being communicated, such as $\drop$, 
$((port = 11) \cdot  (port \leftarrow 13))$, and $((port = 5) \cdot  (port \leftarrow 7))$.

However, in the code in Figure~\ref{fig:synchronise2} the controllers synchronise before updating the rest of the switches.

\begin{figure}
\[
\begin{array}{lcl}
     \mathit{L} & \triangleq &   (( (port = 3) \cdot  (port \leftarrow 5) )   +  \\
     && ( (port = 4) \cdot  (port \leftarrow 6) ) 
     +  \\
     && ( (port = 7) \cdot  (port \leftarrow 8) )  
     +  \\
     && ( (port = 9) \cdot  (port \leftarrow 11) ) 
     +  \\
     && ( (port = 10) \cdot  (port \leftarrow 12) ) 
     +  \\
     && ( (port = 13) \cdot  (port \leftarrow 15) )
    +  \\
     && ( (port = 14) \cdot  (port \leftarrow 16) ))\\ \\
     S_1        & \triangleq &  (port = 2) \cdot  (port \leftarrow 4) \\
     S_2        & \triangleq &  (port = 12) \cdot  (port \leftarrow 14) \\
     S_3        & \triangleq & \drop \\
     S_4        & \triangleq & \drop \\
     S_5        & \triangleq & (port = 6) \cdot  (port \leftarrow 7) \\
     S_6        & \triangleq & (port = 8) \cdot  (port \leftarrow 10) \\

                    
     \mathit{SDN}_{X_1, \ldots, X_6} & \triangleq & (  (X_1 + \ldots + X_6) 
     \cdot L) ^* 
     \Seq  \mathit{SDN}_{X_1, \ldots, X_6}  \oplus \\
                       && \sum_{X'_i \in \mathit{FT}} \mathit{upSi}?X'_i \Seq \mathit{SDN}_{X_1, \ldots, X'_i, \ldots, X_6}  \\ \\
                       
        \mathit{ft3} & \triangleq & (port = 1) \cdot  (port \leftarrow 3)\\
    \mathit{ft4} & \triangleq & (port = 11) \cdot  (port \leftarrow 13)\\
    \mathit{ft5} & \triangleq & (port = 5) \cdot  (port \leftarrow 7)\\
    \mathit{ft6} & \triangleq & (port = 8) \cdot  (port \leftarrow 9)\\
    \mathit{FT} & = & \{ \drop{}, \mathit{ft3}, \mathit{ft4}, \mathit{ft5}, \mathit{ft6}\} \\ \\
     
      \mathit{SDN} & \triangleq & \mathit{SDN}_{S_1, \ldots, S_6} \Par C_1 \Par C_2  \\ \\\
    
     C_1 & \triangleq &   upS1!\drop{}  \Par 
                    upS3!\mathit{ft3}   \Par
                    upS5!\mathit{ft5}   \\ \\
     C_2 & \triangleq &   upS2!\drop{}  \Par 
                    upS4!\mathit{ft4}  \Par
                    upS6!\mathit{ft6} 
\end{array}
\]
    \caption{Distributed Controller in {\DNK}: Independent Controllers}
    \label{fig:independent2}
\end{figure}

\begin{figure}
\noindent
\begin{minipage}[t]{.5\textwidth}
\[
\begin{array}{ccl}
     C_1 & \triangleq &  upS1!\drop{}  ; \\ 
                    && syn!\packcopy{} ;    \\
                    && upS3!( (port = 1) \cdot  (port \leftarrow 3))     ; \\
                    && upS5!( (port = 5) \cdot  (port \leftarrow 7)) \\
                    \\
\end{array}
\]

\end{minipage}
\begin{minipage}[t]{.5\textwidth}
\[
\begin{array}{ccl}
     C_2 & \triangleq &   upS2!\drop{}  ; \\ 
                    && syn?\packcopy{} ;    \\
                    && upS4!( (port = 11) \cdot  (port \leftarrow 13))  ;   \\
                  && upS6!( (port = 8) \cdot  (port \leftarrow 9)) 
\end{array}
\]

\end{minipage}
    \caption{Distributed Controller in {\DNK}: Synchronizing Controllers}
    \label{fig:synchronise2}
\end{figure}
\end{ARXIV}

\subsection{DyNetKAT Semantics}\label{sec:semantics}

\begin{ICALP}
The operational semantics of {\DNK} in Figure~\ref{fig::sem-queue-sem} is provided over configurations of shape $(d, H, H')$, where $d$ stands for the current {\DNK} policy, $H$ is the list of packets to be processed by the network according to $d$ and $H'$ is the list of packets handled successfully by the network. The rule labels $\gamma$ range over pairs of packets $(\sigma, \sigma')$ or communication/reconfiguration-like actions of shape $x!q,\,x?q$ or $\rec{x,q}$, depending on the context.

Note that the {\DNK} semantics is devised in a ``layered'' fashion. Rule $\dedr{\cpolseqsucc}$ in Figure~\ref{fig::sem-queue-sem} is the base rule that makes the transition between the NetKAT denotations and {\DNK} operations. 
More precisely, whenever $\sigma'$ is a packet resulted from the successful evaluation of a NetKAT policy $p$ on $\sigma$, a $(\sigma, \sigma')$-labelled step is observed at the level of {\DNK}. This transition applies whenever the current configuration encapsulates a {\DNK} policy of shape $p;q$ and a list of packets to be processed starting with $\sigma$. The resulting configuration continues with evaluating $q$ on the next packet in the list, while $\sigma'$ is marked as successfully handled by the network.

The remaining rules in Figure~\ref{fig::sem-queue-sem} define non-deterministic choice $\oplus$, synchronization $\Par$ and recursion $X$ in the standard fashion. Note that synchronizations leave the packet lists unchanged. Due to space limitation, we omitted the explicit definitions of the symmetric cases for $\oplus$ and $\Par$. The full semantics is provided in~\cite{arxiv}. 

\begin{figure*}[t]
\footnotesize{
\[
\begin{array}{|c|c|}
\hline &
\\
\dedr{\cpolseqsucc} \sosrule{{\sigma' \in \llbracket p \rrbracket(\sigma \!\!::\!\! \langle \rangle)} }{ ( p ; q, \sigma :: H, H')  \trans{(\sigma, \sigma')} ( q, H , \sigma' :: H')}
& 
\dedr{\cpolX}\sosrule{(p, {H_0,  H_1}) \trans{\gamma} (p', {H'_0,  H'_1}) }{ ( X, H_0, H_1)  \trans{\gamma} ( p' , H'_0, H'_1)} X \triangleq p
\\ & \\
\hline  &
\\
\dedr{\cpoloplusl}\sosrule{(p, H_0, H'_0)  \trans{\gamma} (p', H_1, H'_1)  }{ (p \oplus q, H_0, H'_0)  \trans{\gamma} ( p', H_1, H'_1)}
& 
\dedr{\intl}\sosrule{(p, H_0, H'_0)  \trans{\gamma} (p', H_1, H'_1)  }{ (p || q, H_0, H'_0)  \trans{\gamma} ( p' || q, H_1, H'_1)}
\\ & \\
\hline  
\multicolumn{2}{|c|}{} \\
\multicolumn{2}{|c|}{
\dedr{\cpolmsg}\sosrule{}{ (x\, \bullet \, p ; q, H, H')  \trans{x \bullet p} (  q, H, H')} ~~~\bullet \in \{?, !\}
}\\
\multicolumn{2}{|c|}{} \\
\hline 
\multicolumn{2}{|c|}{} \\
\multicolumn{2}{|c|}{
\dedr{\reconfig}\sosrule{(q, {H,  H'})  \trans{x \,\clubsuit\, p} (q', H,  H')  \quad (s, {H,  H'})  \trans{x \,\spadesuit\, p} (s', H,  H')  }{ ( q || s, H, H')  \trans{\rec{x,p}} ( q' || s' , H, H')}~~~
\begin{array}{c}
     \clubsuit = ? ~~~ \spadesuit = !\\
    \textnormal{or}\\ 
     \clubsuit = ! ~~~ \spadesuit = ?
\end{array}
}
\\
\multicolumn{2}{|c|}{} \\
\hline
\end{array} 
\]
\begin{center}
$\gamma ::= (\sigma, \sigma') \mid x!q \mid x?q \mid {\rec{x,q}}$
\end{center}
}
\caption{{\DNK}: Operational Semantics (relevant excerpt)}
\label{fig::sem-queue-sem} 
\end{figure*}

\end{ICALP}

\begin{ARXIV}
The operational semantics of {\DNK} in Figure~\ref{fig::sem-queue-sem} is provided over configurations of shape $(d, H, H')$, where $d$ stands for the current {\DNK} policy, $H$ is the list of packets to be processed by the network according to $d$ and $H'$ is the list of packets handled successfully by the network. The rule labels $\gamma$ range over pairs of packets $(\sigma, \sigma')$ or communication/reconfiguration-like actions of shape $x!q,\,x?q$ or $\rec{x,q}$, depending on the context.

Note that the {\DNK} semantics is devised in a ``layered'' fashion. Rule $\dedr{\cpolseqsucc}$ in Figure~\ref{fig::sem-queue-sem} is the base rule that makes the transition between the NetKAT denotations and {\DNK} operations. More precisely, whenever $\sigma'$ is a packet resulted from the successful evaluation of a NetKAT policy $p$ on $\sigma$, a $(\sigma, \sigma')$-labelled step is observed at the level of {\DNK}. This transition applies whenever the current configuration encapsulates a {\DNK} policy of shape $p;q$ and a list of packets to be processed starting with $\sigma$. The resulting configuration continues with evaluating $q$ on the next packet in the list, while $\sigma'$ is marked as successfully handled by the network.

The remaining rules in Figure~\ref{fig::sem-queue-sem} define non-deterministic choice, synchronization and recursion in the standard fashion.

Rules $\dedr{\cpoloplusl}$ and $\dedr{\cpoloplusl}$ define non-deterministic behaviours. Assume $H_0$ is the list of packets to be processed by the network according to $p$ (respectively, $q$) and $H'_0$ is the list of packets handled successfully by the network.
Whenever $p$ (respectively, $q$) determines a $\gamma$-labelled transition into $(p', H_1, H'_1)$ (respectively, $(q', H_1, H'_1)$), the policy $p \oplus q$ is able to mimic the same behaviour.
Rules $\dedr{\intl}$ and $\dedr{\intr}$ follow a similar pattern; the only difference is that the ``inactive'' operand is preserved by the target of the semantic rule.

Mere sending $\dedr{\cpolmsgsend}$ and receiving $\dedr{\cpolmsgrec}$ entail transitions labelled accordingly, and continue with the {\DNK} policy following the $;$ operator. Note that the list of packets to be processed by the network and the list of packets handled successfully by the network remain unchanged.

{\DNK} synchronization is defined by $\dedr{\reconfigsr}$ and $\dedr{\reconfigrs}$. Intuitively, when both operands $q$ and, respectively, $s$ ``agree'' on sending/receiving a policy $p$ on channel $x$ in the context of the same packet lists $H$ and $H'$, and behave like $q'$, respectively, $s'$ afterwards, then a $\rec{x,p}$ step can be observed. The system proceeds with the continuation behaviour $q' \Par s'$.

As denoted by $\dedr{\cpolX}$, a recursive variable defined as $X \triangleq p$ behaves according to $p$.

\begin{figure}[t]
\[
\begin{array}{|c|c|}
\hline &
\\
\dedr{\cpolseqsucc} \sosrule{{\sigma' \in \llbracket p \rrbracket(\sigma \!\!::\!\! \langle \rangle)} }{ ( p ; q, \sigma :: H, H')  \trans{(\sigma, \sigma')} ( q, H , \sigma' :: H')}
& 
\dedr{\cpolX}\sosrule{(p, {H_0,  H_1}) \trans{\gamma} (p', {H'_0,  H'_1}) }{ ( X, H_0, H_1)  \trans{\gamma} ( p' , H'_0, H'_1)} X \triangleq p
\\ & \\
\hline  &
\\
\dedr{\cpoloplusl}\sosrule{(p, H_0, H'_0)  \trans{\gamma} (p', H_1, H'_1)  }{ (p \oplus q, H_0, H'_0)  \trans{\gamma} ( p', H_1, H'_1)}
& 
\dedr{\cpoloplusr}
\sosrule{(q, H_0, H'_0)  \trans{\gamma} (q', H_1, H'_1)  }{ (p \oplus q, H_0, H'_0)  \trans{\gamma} (  q', H_1, H'_1)}
\\ & \\
\hline  &
\\
\dedr{\intl}\sosrule{(p, H_0, H'_0)  \trans{\gamma} (p', H_1, H'_1)  }{ (p || q, H_0, H'_0)  \trans{\gamma} ( p' || q, H_1, H'_1)}
& 
\dedr{\intr}\sosrule{(q, H_0, H'_0)  \trans{\gamma} (q', H_1, H'_1)  }{ (p || q, H_0, H'_0)  \trans{\gamma} ( p || q', H_1, H'_1)}
\\ & \\
\hline  &
\\
\dedr{\cpolmsgrec}\sosrule{}{ (x?p ; q, H, H')  \trans{x?p} (  q, H, H')}
& 
\dedr{\cpolmsgsend}\sosrule{}{ (x!p ; q, H, H')  \trans{x!p} ( q, H, H')}
\\ & \\
\hline  
\multicolumn{2}{|c|}{} \\
\multicolumn{2}{|c|}{
\dedr{\reconfigsr}\sosrule{(q, H,  H') \trans{x!p} (q', H , H') \quad (s, H,  H')  \trans{x?p} (s', H,  H')  }{ ( q || s, H, H')  \trans{\rec{x,p}} ( q' || s' , H, H')}}\\
\multicolumn{2}{|c|}{} \\
\hline 
\multicolumn{2}{|c|}{} \\
\multicolumn{2}{|c|}{
\dedr{\reconfigrs}\sosrule{(q, {H,  H'})  \trans{x?p} (q', H,  H')  \quad (s, {H,  H'})  \trans{x!p} (s', H,  H')  }{ ( q || s, H, H')  \trans{\rec{x,p}} ( q' || s' , H, H')}}\\
\multicolumn{2}{|c|}{} \\
\hline
\end{array} 
\]
$\gamma ::= (\sigma, \sigma') \mid x!q \mid x?q \mid {\rec{x,q}}$
\caption{{\DNK}: Operational Semantics}
\label{fig::sem-queue-sem} 
\end{figure}
\end{ARXIV}

In Figure~\ref{fig:stateful-firewall-LTS} we depict a labelled transition system (LTS) encoding a possible behaviour of the stateful firewall in Example~\ref{eg:firewall}. We assume the list of network packets to be processed consists of a ``safe'' packet $\sigma_i$ travelling from {\it{int}} to {\it{ext}}   (i.e., $\sigma_i(\it{port}) = \it{int}$) followed by a potentially ``dangerous'' packet $\sigma_e$ travelling from {\it{ext}} to {\it{int}} (i.e., $\sigma_e(\it{port}) = \it{ext}$). For the simplicity of notation, in Figure~\ref{fig:stateful-firewall-LTS} we write \textit{H} for \textit{Host}, \textit{S} for \textit{Switch}, $\mathit{S}'$ for $\mathit{Switch}'$, $\mathit{SCR}$ for $\mathit{secConReq}$ and $\mathit{SCE}$ for $\mathit{secConEnd}$. Note that $\sigma_e$ can enter the network only if a secure connection request was received. More precisely, the transition labelled $(\sigma_e, \sigma_i)$ is preceded by a transition labelled $\mathit{SCR}?1$ or $\rec{\mathit{SCR}, 1}$:
$n_2 \xrightarrow{\mathit{SCR}?1,\,\, \rec{\mathit{SCR}, 1}} n_3 \xrightarrow{(\sigma_e, \sigma_i)} n_4$.

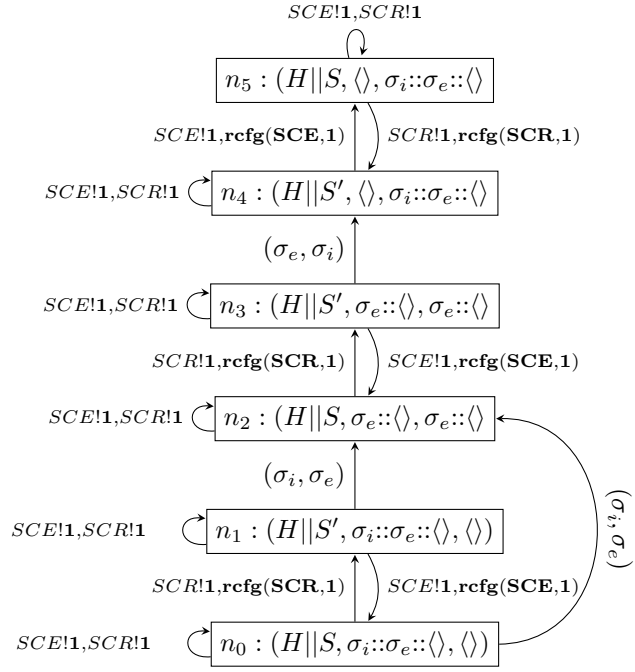
\begin{figure}[ht]
\begin{center}
\begin{tikzpicture}[>=stealth, scale=0.16,,
   shorten >=.3pt,auto,node distance=1.5cm, node/.style={rectangle,draw}]
            \node[node] (n1) {$n_1: (H||S', \sigma_i{::}\sigma_e{::}\langle \rangle, \langle \rangle)$};
            
            \node[node] (n0) [below of=n1] {$n_0: (H||S, \sigma_i{::}\sigma_e{::}\langle \rangle, \langle \rangle)$};
        
            \node[node] (n2) [above of=n1] {$n_2: (H||S, \sigma_e{::}\langle \rangle, \sigma_e{::}\langle \rangle$};
            
            \node[node] (m3) [above of=n2] {$n_3: (H||S', \sigma_e{::}\langle \rangle, \sigma_e{::}\langle \rangle$};
            
            \node[node] (m4) [above of=m3] {$n_4: (H||S', \langle \rangle, \sigma_i{::}\sigma_e{::}\langle \rangle$};
            
            \node[node] (m5) [above of=m4] {$n_5: (H||S, \langle \rangle, \sigma_i{::}\sigma_e{::}\langle \rangle$};
                         
            \path[->]
            (n1) edge [out=185,in=175,looseness=3] node {$\scriptstyle{SCE!\packcopy{}, SCR!\packcopy{}}~~~$} (n1)
            (n1) edge [bend left] node {$\scriptstyle{SCE!\packcopy{}, \rec{SCE, \packcopy{}}}$} (n0)
            (n1) edge [left] node {$(\sigma_i, \sigma_e)$} (n2)            
            (n0) edge [out=185,in=175,looseness=3] node {$\scriptstyle{SCE!\packcopy{}, SCR!\packcopy{}}~~~$} (n0)
            (n0) edge node [left] {$\scriptstyle{SCR!\packcopy{}, \rec{SCR, \packcopy{}}}$} (n1)
            (n0) edge [bend right=90, ,looseness=1.5] node [near start,midway,sloped,above] {$(\sigma_i, \sigma_e)$} (n2)
            (n2) edge [out=185,in=175,looseness=3] node [left] {$\scriptstyle{SCE!\packcopy{}, SCR!\packcopy{}}$} (n2)
            (n2) edge node [left] {$\scriptstyle{SCR!\packcopy{}, \rec{SCR, \packcopy{}}}$} (m3)
            (m3) edge [bend left] node {$\scriptstyle{SCE!\packcopy{}, \rec{SCE, \packcopy{}}}$} (n2)
            (m3) edge node {$(\sigma_e, \sigma_i)$} (m4)
            (m3) edge [out=185,in=175,looseness=3] node [left] {$\scriptstyle{SCE!\packcopy{}, SCR!\packcopy{}}$} (m3)
            (m4) edge [out=185,in=175,looseness=3] node [left] {$\scriptstyle{SCE!\packcopy{}, SCR!\packcopy{}}$} (m4)
            (m4) edge node [left] {$\scriptstyle{SCE!\packcopy{}, \rec{SCE, \packcopy{}}}$} (m5)
            (m5) edge [bend left] node [right] {$\scriptstyle{SCR!\packcopy{}, \rec{SCR, \packcopy{}}}$} (m4)
            (m5) edge [loop above] node {$\scriptstyle{SCE!\packcopy{}, SCR!\packcopy{}}$} (m5);
 \end{tikzpicture}       
\end{center}
    \caption{Stateful Firewall LTS}
    \label{fig:stateful-firewall-LTS}
\end{figure}

\begin{ARXIV}
In Figure~\ref{fig:indep-controllers-LTS} we depict an excerpt of the LTS corresponding to the distributed independent controllers in Example~\ref{eg:controller}, given a network packet denoted by $\sigma_2$. In Figure~\ref{fig:indep-controllers-LTS} we write $\sigma_i$ to denote a network packet such that $\sigma_i(port)=i$. For instance, transitions of shape $n_0 \xrightarrow{(\sigma_2, \sigma_i)} \overline{n}_i$ encode forwarding of the current packet from port $2$ to port $i$ based on the subsequent unfoldings of the Kleene-star expression in the definition of $\mathit{SDN}_{X_1, \ldots, X_6}$. The transition $n_2 \xrightarrow{(\sigma_2, \sigma_{15})} \overline{\overline{n}}_{15}$ reveals a breach in the network corresponding to the possibility of forwarding the current packet from $H_1$ to $H_4$. This is possible due to two consecutive reconfigurations of the flow tables of switches $S_6$ and $S_4$, respectively, enabling traffic from port $8$ to $9$, and from port $11$ to $13$.

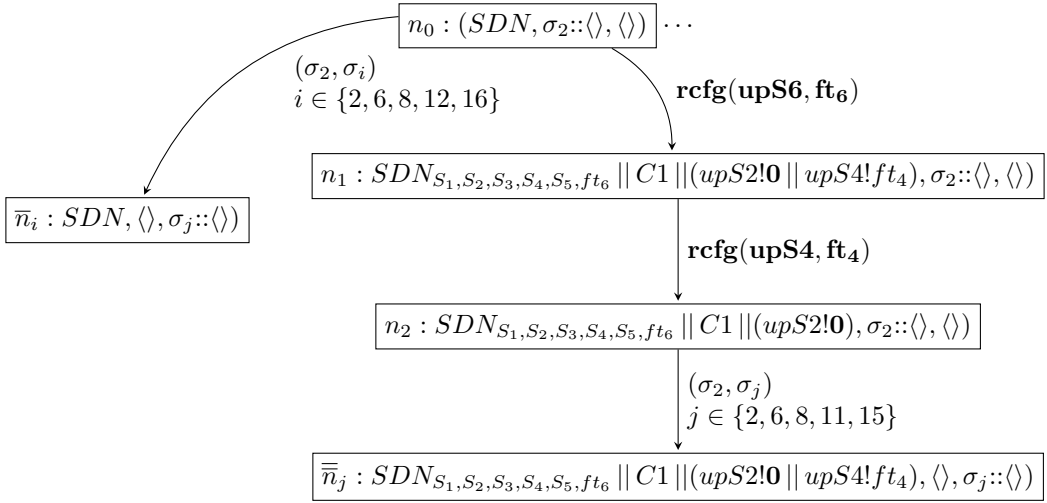
\begin{figure}[h]
	\centering
	\begin{tikzpicture}[>=stealth, scale=0.4,,
	shorten >=1pt,auto,node distance=2cm, dot/.style={}, node/.style={rectangle,draw}]
	\node[node] (m0) {$n_0: (SDN, \sigma_2{::}\langle \rangle, \langle \rangle)$};
	
	\node[dot] (d1) [right of=m0] {$\ldots$};
	
	\node[node] (m1) [below of=d1] {$n_1: SDN_{S_1,S_2,S_3,S_4,S_5,ft_6} \Par C1 \Par (upS2!\drop{} \Par upS4!ft_4), \sigma_2{::}\langle \rangle, \langle \rangle)$};
	
	\node[node] (m2) [below of=m1] {$n_2: SDN_{S_1,S_2,S_3,S_4,S_5,ft_6} \Par C1 \Par (upS2!\drop{}), \sigma_2{::}\langle \rangle, \langle \rangle)$};
	
	\node[node] (mj) [below of=m2] {$\overline{\overline{n}}_j: SDN_{S_1,S_2,S_3,S_4,S_5,ft_6} \Par C1 \Par (upS2!\drop{} \Par upS4!ft_4), \langle \rangle, \sigma_j{::}\langle \rangle)$};
	
	\node[node] (mi) [below left=2.8cm of m0] {$\overline{n}_i: SDN, \langle \rangle, \sigma_j{::}\langle \rangle)$};
		
	\path[->]
	(m0) edge [bend left] node [pos=0.7] {$\rec{upS6,ft_6}$} (m1)
	(m1) edge [] node {$\rec{upS4,ft_4}$} (m2)
	(m2) edge [] node [align=left] {$(\sigma_2,\sigma_j)$\\$j\in \{ 2, 6, 8, 11, 15\}$} (mj)
	(m0) edge [bend right] node [pos=0.35,align=left] {$(\sigma_2,\sigma_i)$\\$i\in \{ 2, 6, 8, 12, 16\}$} (mi);
	\end{tikzpicture}    
	
	\caption{Independent Controllers LTS (excerpt)}
	\label{fig:indep-controllers-LTS}
\end{figure}
\end{ARXIV}

\begin{figure*}[t]
	\begin{minipage}[t]{0.5\linewidth}
		\[
		\def\arraystretch{1.2}
		\footnotesize{
			\begin{array}{r@{}c@{}ll}
			\multicolumn{4}{l}{\textnormal{for $p, q, r \in \DNK$ and $z, y \in {\NetKATnoDup}$}}\\
			\multicolumn{4}{l}{\textnormal{for } a ::= z \mid x?z \mid x!z \mid \recp{x,z}}\\
			
			\drop{} \Seq p &\,\equiv\,\,& \bot & (A0) \\
			(z + y) \Seq p &\,\equiv\,\,& z \Seq p \oplus y \Seq p & (A1) \\
			p \oplus q &\,\equiv\,\,& q \oplus p & (A2) \\
			(p \oplus q) \oplus r &\,\equiv\,\,& p \oplus (q \oplus r) & (A3) \\
			p \oplus p &\,\equiv\,\,& p & (A4) \\
			p \oplus \bot &\,\equiv\,\,& p & (A5) \\
			p \Par q &\,\equiv\,\,& q \Par p & (A6) \\
			
			p \Par \bot &\,\equiv\,\,& p & (A7) \\
			p \Par q &\,\equiv\,\,& p \llfloor q \oplus q \llfloor p \oplus p \mid q & (A8) \\
			\bot \llfloor p &\,\equiv\,\,& \bot & (A9) \\
			(a \Seq p) \llfloor q &\,\equiv\,\,& a \Seq (p \Par q) & (A10) \\
			(p \oplus q) \llfloor r &\,\equiv\,\,& (p \llfloor r) \oplus (q \llfloor r) & (A11) \\
			(x?\mathit{z} \Seq p) \mid (x!\mathit{z} \Seq q) &\,\equiv\,\,& \recp{x,z} \Seq (p \Par q) & (A12)\\
			(p \oplus q) \mid r &\,\equiv\,\,& (p \mid r) \oplus (q \mid r) & (A13)\\
			p \mid q &\,\equiv\,\,& q \mid p & (A14)\\
			p \mid q &\,\equiv\,\,& \bot~[owise] & (A15)\\			
			\end{array}
		}
		\]
	\end{minipage}
	\hfill
	\begin{minipage}[t]{0.5\linewidth}
		\[
		\footnotesize{
			\def\arraystretch{1.2}
			\begin{array}{r@{}c@{}ll}
			& & & \\
			\multicolumn{4}{l}{\textnormal{for $\at  ::=  \alpha \cdot \pi \mid x?\mathit{z} \mid x!\mathit{z} \mid \recp{x,z}$:}}\\
			
			\delta_{{\mathcal{L}}}(\bot) &\,\equiv\,\,& \bot & (\delta_\bot)\\
			\delta_{{\mathcal{L}}}(\at \Seq p) &\,\equiv\,\,& \at \Seq \delta_{{\mathcal{L}}}(p)~{\textnormal{if}}~ \at \not \in {\mathcal{L}} & (\delta_{\Seq})\\
			\delta_{{\mathcal{L}}}(\at \Seq p) &\,\equiv\,\,& \bot ~{\textnormal{if}}~ \at \in {\mathcal{L}} & (\delta_{\Seq}^{\bot})\\
			\delta_{{\mathcal{L}}}(p\oplus q) &\,\equiv\,\,& \delta_{{\mathcal{L}}}(p) \oplus \delta_{{\mathcal{L}}}(q) & (\delta_{\oplus})\\
			
			& & & \\
			\multicolumn{4}{l}{\textnormal{for $n \in \mathbb{N}:$}}\\
			\pi_0(p) &\,\equiv\,\,& \bot & (\Pi_0)\\
			\pi_{n}(\bot) &\,\equiv\,\,& \bot & (\Pi_\bot)\\
			\pi_{n+1}(\at \Seq p) &\,\equiv\,\,& \at \Seq \pi_{n}(p)  & (\Pi_{\Seq})\\
			\pi_{n}(p \oplus q) &\,\equiv\,\,& \pi_{n}(p) \oplus \pi_{n}(q)  & (\Pi_{\oplus})\\
			
			& & & \\
			p &\,\equiv\,\,& q \hfill~{\textnormal{if}}~ \forall n \in \mathbb{N} \,:\, \pi_n(p) \,\equiv\,\,  \pi_n(q) & ({\it AIP})\\\\
			
			\ENK & & \\
			
			
			\end{array}
		}
		\]
	\end{minipage}
	
	\caption{The axiom system $\EDNK$ (including $\ENK$)}
	\label{fig:axiom-ccp}
\end{figure*}

\section{Semantic Results}\label{sec:sem-results}

\begin{ICALP}
In this section we define bisimilarity of {\DNK} policies and provide a corresponding sound and ground-complete axiomatization.

Bisimilarity of {\DNK} terms is defined in the standard fashion:

\begin{defn}[Bisimilarity ($\sim$)]\label{def::bisim}
A symmetric relation $R$ over {\DNK} policies is a \emph{bisimulation} whenever for $(p, q) \in R$ the following holds: \\
If $(p, H_0, H_1) \xrightarrow{\gamma} (p', H'_0, H'_1)$ then
exists $q'$ s.t. $(q, H_0, H_1) \xrightarrow{\gamma} (q', H'_0, H'_1)$ and $(p', q') \in R$,
with
$\gamma ::= (\sigma, \sigma')  \mid x?r \mid x!r \mid \rec{x, r}$.

We call \emph{bisimilarity} the largest bisimulation relation.
Two policies $p$ and $q$ are \emph{bisimilar}  ($p \sim q$) if and only if there is a bisimulation relation $R$ such that $(p, q) \in R$.
\end{defn}

Semantic equivalence of $\NetKATnoDup$ policies is preserved by $\DNA$ bisimilarity.

\begin{proposition}[Semantic Layering]\label{lm:sound-net-kat}
Let $p$ and $q$ be two $\NetKATnoDup$ policies. The following holds:
$
 \llbracket p \rrbracket = \llbracket q  \rrbracket \textnormal{ iff } (p ; d) \sim (q ; d)
$
for any {\DNK} policy $d$.
\end{proposition}
\begin{proof}
This follows according to the definition of bisimilarity and $\dedr{\cpolseqsucc}$ in Figure~\ref{fig::sem-queue-sem}.
\end{proof}

We further provide some additional ingredients needed to introduce the {\DNK} axiomatization in Figure~\ref{fig:axiom-ccp}. First, note that our notion of bisimilarity identifies synchronization steps as in $\dedr{\reconfig}$ in Figure~\ref{fig::sem-queue-sem}. At the axiomatization level, this requires introducing corresponding constants $\recp{x,z}$ defined as:

\[\sosrule{}{(\recp{x,z} ; p, H_0, H_1) \xgtrans{}{\rec{x,z}} (p, H_0, H_1)}.\]

In accordance with standard approaches to process algebra (see, e.g., ~\cite{DBLP:journals/iandc/AcetoBV94,DBLP:books/daglib/0069083}) we consider the restriction operator $\delta_{\mathcal{L}}(-)$ with ${\mathcal{L}}$ a set of forbidden actions ranging over $x?z$ and $x!z$ as in~(\ref{def::DNA-syntax}). In practice, we use the restriction operator to force synchronous communication. We also define a projection operator $\pi_n(-)$ that, intuitively, captures the first $n$ steps of a {\DNK} policy. $\pi_n(-)$ is crucial for defining the so-called ``Approximation Induction Principle'' that enables reasoning about equivalence of recursive {\DNK} specifications. Last, but not least, in our axiomatization we employ the left-merge operator ($\llfloor$) and the communication-merge operator ($\mid$) utilised for axiomatizing parallel composition. Intuitively, a process of shape $p \llfloor q$ behaves like $p$ as a first step, and then continues as the parallel composition between the remaining behaviour of $p$ and $q$. A process of shape $p \mid q$ forces the synchronous communication between $p$ and $q$ in a first step, and then continues as the parallel composition between the remaining behaviours of $p$ and $q$. The full description of these auxiliary operators is provided in~\cite{arxiv}.

In Figure~\ref{fig:axiom-ccp}, we introduce $\EDNK$ -- the axiom system of {\DNK}, including the NetKAT axiomatization $\ENK$.
Most of the axioms in Figure~\ref{fig:axiom-ccp} comply to the standard axioms of parallel and communicating processes~\cite{DBLP:books/daglib/0069083}, where, intuitively, $\oplus$ plays the role of non-deterministic choice, $;$ resembles sequential composition and $\bot$ is a process that deadlocks.
An interesting axiom is $(A7):\, p \Par \bot \equiv p$ which, intuitively, states that if one network component fails, then the whole system continues with the behaviour of the remaining components. This is a departure from the approach in~\cite{Kappe20}, where recovery is not possible in case of a component's failure; i.e., $e \Par 0 \equiv 0$. Additionally, $(A12)$ ``pin-points'' a communication step via the newly introduced constants of form $\recp{x,z}$.
Axiom $(A0)$ states that if the current packet is dropped as a result of the unsuccessful evaluation of a NetKAT policy, then the continuation is deadlocked. $(A1)$ enables mapping the non-deterministic choice at the level of NetKAT to the setting of {\DNK}.

The axioms encoding the restriction operator $\delta_{\mathcal{L}}(-)$ and the projection operator $\pi_n(-)$ are defined in the standard fashion, on top of {\DNK} normal forms.
The latter are defined inductively, as sums of complete tests and complete assignments $\alpha \cdot \pi$, or communication steps $x?q,\, x!q$ and $\recp{x,q}$, followed by arbitrary {\DNK} policies:
\[
\Sigma^{\oplus}_{i \in I} (\alpha_i \cdot \pi_i) ; d_i \,\oplus\, \Sigma^{\oplus}_{j \in J} c_j ; d_j \,(\oplus \bot)
\]
where $d_i, d_j$ range over {\DNK} policies and $c_j ::= x?q \mid x!q \mid \recp{x,q}$ with $q$ denoting terms in $\NetKATnoDup$.
Complete tests (typically denoted by $\alpha$) and complete assignments (typically denoted by $\pi$) were originally introduced in~\cite{netkat}. In short: let $F = \{f_1, \ldots, f_n\}$ be a set of fields names with values in $V_i$, for $i \in \{1, \ldots, n\}$.
We call \emph{complete test} (resp., \emph{complete assignment}) an expression $f_1 = v_1 \cdot \ldots \cdot f_n = v_n$ (resp., $f_1 \leftarrow v_1 \cdot \ldots \cdot f_n \leftarrow v_n$), with $v_i \in V_i$, for $i \in \{1, \ldots, n\}$. 

\begin{theorem}[Soundness \& Completeness]\label{thm:sound-complete}
$\EDNK$ is sound and ground-complete with respect to {\DNK} bisimilarity.
\end{theorem}
\begin{proof}
The proof of soundness follows in two steps. First, we show that axioms $(A0)$-$(\Pi_{\oplus})$ together with the NetKAT axioms in $\ENK$ are sound w.r.t. {\DNK} bisimilarity. This is shown in a standard fashion, by case analysis on transitions of shape
$(p, H_0, H'_0) \xrightarrow{\gamma} (q, H_1, H'_1)$
with $\gamma ::= (\sigma, \sigma') \mid x?n \mid x!n \mid \rec{x,n}$. Then, we show that $(AIP)$ holds; the proof is close to the one of Theorem~$2.5.8$ in~\cite{DBLP:books/daglib/0069083} and uses the branching finiteness property of guarded {\DNK} policies.

Given two bisimilar {\DNK} policies $p$ and $q$ in normal form, 
ground-completeness reduces to showing that every summand of $p$ can be proven syntactically equal to a summand of $q$. The reasoning is by structural induction, and case analysis on the type of transitions that can be triggered as a first step. For showing the existence of {\DNK} normal forms, we exploit the results of Lemma~$4$ in~\cite{netkat} stating that the standard semantics of every NetKAT expression is equal to the union of its minimal nonzero terms. In the context of $\NetKATnoDup$ and packet values drawn
from finite domains (as is the case in~\cite{netkat}), this union can be equivalently expressed as a sum of complete tests and complete assignments. I.e., $\vdash r \equiv \Sigma_{i \in I} \alpha_i \cdot \pi_i$ for every $\NetKATnoDup$ expression $r$.
We refer to~\cite{arxiv} for the complete proof and additional details.
\end{proof}

\end{ICALP}

\begin{ARXIV}
In this section we define bisimilarity of {\DNK} policies, introduce some necessary definitions and terminology, and provide a corresponding sound and complete axiomatization.

\subsection{An Axiom System for {\bf{DyNetKAT}} Bisimilarity}\label{sec:axioms}

Bisimilarity of {\DNK} terms is defined in the standard fashion:

\begin{defn}[Bisimilarity ($\sim$)]\label{def::bisim}
A symmetric relation $R$ over {\DNK} policies is a \emph{bisimulation} whenever for $(p, q) \in R$ the following holds: \\
If $(p, H_0, H_1) \xrightarrow{\gamma} (p', H'_0, H'_1)$ then
exists $q'$ s.t. $(q, H_0, H_1) \xrightarrow{\gamma} (q', H'_0, H'_1)$ and $(p', q') \in R$,
with
$\gamma ::= (\sigma, \sigma')  \mid x?r \mid x!r \mid \rec{x, r}$.

We call \emph{bisimilarity} the largest bisimulation relation.

Two policies $p$ and $q$ are \emph{bisimilar}  ($p \sim q$) if and only if there is a bisimulation relation $R$ such that $(p, q) \in R$.
\end{defn}

Semantic equivalence of $\NetKATnoDup$ policies is preserved by $\DNA$ bisimilarity.

\begin{proposition}[Semantic Layering]\label{lm:sound-net-kat}
Let $p$ and $q$ be two $\NetKATnoDup$ policies. The following holds:
\[
 \llbracket p \rrbracket = \llbracket q  \rrbracket \textnormal{ iff } (p ; d) \sim (q ; d)
\]
for any {\DNK} policy $d$.
\end{proposition}
\begin{proof}
The result follows directly according to the definition of bisimilarity and $\dedr{\cpolseqsucc}$ in Figure~\ref{fig::sem-queue-sem}.
\end{proof}

Next, we introduce the restriction operator $\delta_{\mathcal{L}}(-)$~\cite{DBLP:journals/iandc/AcetoBV94,DBLP:books/daglib/0069083}, with ${\mathcal{L}}$ a set of forbidden actions ranging over $x?z$ and $x!z$ as in~(\ref{def::DNA-syntax}). The semantics of $\delta_{\mathcal{L}}(-)$ is:
\begin{equation}\label{eq:delta}
\dedr{{\delta^{}}} \sosrule{(p, H_0, H_1) \xgtrans{{{}}}{\gamma} (p', H'_0, H'_1)}{(\delta_{\mathcal{L}}(p), H_0, H_1) \xgtrans{{}}{\gamma} (\delta_{\mathcal{L}}(p'), H'_0, H'_1)} \gamma \not \in \mathcal{L}
\end{equation}
In practice, we use the restriction operator to force synchronous communication. For an example, consider the synchronising controllers in Figure~\ref{fig:synchronise2}. Let $\mathcal{L}$ be the set of restricted actions ranging over elements of shape $upS_i?X$, $upS_i!X$, $syn?1$ and $syn!1$. The restricted system $\delta_{\mathcal{L}}(\mathit{SDN}_{S1, \ldots, S6} \Par C_1 \Par C_2)$
ensures that: (1) traffic through $S2$ and $S1$ is first disabled via reconfigurations $\rec{upS2, 0}$ and $\rec{upS1,0}$ and (2) the controllers acknowledge this deactivation via a synchronization step $\rec{syn, 1}$ before installing further flow tables for $S4$ and $S6$.

In the style of~\cite{DBLP:books/daglib/0069083}, we define a projection operator $\pi_n(-)$ that, intuitively, captures the first $n$ steps of a {\DNK} policy. Its formal semantics is:
\begin{equation}\label{eq:pi}
    \dedr{{\pi^{}}} \sosrule{(p, H_0, H_1) \xgtrans{{{}}}{\gamma} (p', H'_0, H'_1)}{(\pi_{n+1}(p), H_0, H_1) \xgtrans{{}}{\gamma} (\pi_{n}(p'), H'_0, H'_1)} 
\end{equation}
As we shall later see, $\pi_n(-)$ is crucial for defining the so-called ``Approximation Induction Principle'' that enables reasoning about equivalence of recursive {\DNK} specifications.

We further provide some additional ingredients needed to introduce the {\DNK} axiomatization in Figure~\ref{fig:axiom-ccp}. First, note that our notion of bisimilarity identifies synchronization steps as in $\dedr{\reconfigsr}$ and $\dedr{\reconfigrs}$. At the axiomatization level, this requires introducing corresponding constants $\recp{x,z}$ defined as:
\begin{equation}\label{eq:op-recp-axiom}
    \dedr{\recp{x,z}} \sosrule{}{(\recp{x,z} ; p, H_0, H_1) \xgtrans{}{\rec{x,z}} (p, H_0, H_1)} 
\end{equation}

Last, but not least, we introduce the left-merge operator ($\llfloor$) and the communication-merge operator ($\mid$) utilised for axiomatizing parallel composition. Intuitively, a process of shape $p \llfloor q$ behaves like $p$ as a first step, and then continues as the parallel composition between the remaining behaviour of $p$ and $q$. A process of shape $p \mid q$ forces the synchronous communication between $p$ and $q$ in a first step, and then continues as the parallel composition between the remaining behaviours of $p$ and $q$. The corresponding semantic rules are:

\begin{equation}\label{eq:aux-op-sem}
\begin{array}{l}
\dedr{\llfloor} \sosrule{(p, {H_0, H_1}) \xgtrans{}{\gamma} (p', H'_0, H'_1 ) }{(p \llfloor q, H_0, H_1)) \xgtrans{}{\gamma} (p' \Par q,  H'_0, H'_1)}~~~~\gamma ::= (\sigma, \sigma') \mid  x!p \mid x?p \mid \rec{x,p}\\[3ex]
\dedr{\mid^{?!}} \sosrule{(p, H, H') \xgtrans{{}}{x?r} (p', H, H') \quad\quad  (q, H, H') \xgtrans{{}}{x!r} (q', H, H')}{(p \mid q, H, H') \xgtrans{}{\rec{x,p}} (p' \Par q', H, H')}\\[3ex]
\dedr{\mid^{!?}} \sosrule{(p, H, H') \xgtrans{{}}{x!r} (p', H, H') \quad\quad  (q, H, H') \xgtrans{{}}{x?r} (q', H, H')}{(p \mid q, H, H') \xgtrans{}{\rec{x,p}} (p' \Par q', H, H')}
\end{array}
\end{equation}

From this point onward, we denote by {\DNK} the extension with the operators in~(\ref{eq:delta}), (\ref{eq:pi}) and (\ref{eq:op-recp-axiom}):
\begin{equation}\label{def::DNA-syntax-e}
\begin{array}{ccl}
     \mathit{N} & ::= & {\NetKATnoDup}\\[1ex]
     \mathit{D}_e & ::= &
     { \zeroq } \mid
     \mathit{N} \Seq \mathit{D}
     \mid  x?\mathit{N} \Seq D_e \mid 
     x!\mathit{N} \Seq D_e \mid \recp{x,N} \Seq D_e \mid\\
     && \mathit{D}_e \Par \mathit{D}_e \mid
     D_e \oplus D_e \mid \delta_{\mathcal{L}}(D_e)\mid \pi_n(D_e) \mid D_e \llfloor D_e \mid D_e \!\!\mid\!\! D_e \mid X  \\
     && X \triangleq D_e,\, n \in \mathbb{N},\,
     \mathcal{L} = \{ c \mid c ::= x?N \mid x!N\}
     
\end{array}
\end{equation}
Bisimilarity is defined for {\DNK} terms as in~(\ref{def::DNA-syntax-e}) in the natural fashion, according to the operational semantics of the new operators in~(\ref{eq:delta}), (\ref{eq:pi}) and (\ref{eq:op-recp-axiom}).

\begin{lem}\label{lm:congruence}
{\DNK} bisimilarity is a \emph{congruence}.
\end{lem}
\begin{proof}
The result follows from the fact that
the semantic rules defined in this paper comply to the congruence formats proposed in~\cite{MOUSAVI2005107}.
\end{proof}

\begin{defn}[Complete Tests \& Assignments~\cite{netkat}]\label{def:CTA}
Let $F = \{f_1, \ldots, f_n\}$ be a set of fields names with values in $V_i$, for $i \in \{1, \ldots, n\}$. We call \emph{complete test} (typically denoted by $\alpha$) an expression $f_1 = v_1 \cdot \ldots \cdot f_n = v_n$, with $v_i \in V_i$, for $i \in \{1, \ldots, n\}$.
We call \emph{complete assignment} (typically denoted by $\pi$) an expression $f_1 \leftarrow v_1 \cdot \ldots \cdot f_n \leftarrow v_n$, with $v_i \in V_i$, for $i \in \{1, \ldots, n\}$.
We sometimes write $\alpha_{\pi}$ in order to denote the complete test derived from the complete assignment $\pi$ by replacing all $f_i \leftarrow v_i$ in $\pi$ with $f_i = v_i$; symmetrically for $\pi_{\alpha}$. Additionally, we sometimes write $\sigma_{\alpha}$ to denote the network packet whose fields are assigned the corresponding values in $\alpha$; symmetrically for $\sigma_{\pi}$.
\end{defn}

In Figure~\ref{fig:axiom-ccp}, we introduce $\EDNK$ -- the axiom system of {\DNK}, including the NetKAT axiomatization $\ENK$.
Most of the axioms in Figure~\ref{fig:axiom-ccp} comply to the standard axioms of parallel and communicating processes~\cite{DBLP:books/daglib/0069083}, where, intuitively, $\oplus$ plays the role of non-deterministic choice, $;$ resembles sequential composition and $\bot$ is a process that deadlocks.

For instance, axioms $(A2)-(A5)$ encode the ACI properties of $\oplus$ together with the fact that $\bot$ is the neutral element.

Axioms $(A8)-(A15)$ define parallel composition ($\Par$) in terms of left-merge ($\llfloor$) and communication-merge ($\,\mid\,$) in the standard fashion.
Additionally, $(A12)$ ``pin-points'' a communication step via the newly introduced constants of form $\recp{x,z}$.
An interesting axiom is $(A7):\, p \Par \bot \equiv p$ which, intuitively, states that if one network component fails, then the whole system continues with the behaviour of the remaining components. This is a departure from the approach in~\cite{Kappe20}, where recovery is not possible in case of a component's failure; i.e., $e \Par 0 \equiv 0$.

Axiom $(A0)$ states that if the current packet is dropped as a result of the unsuccessful evaluation of a NetKAT policy, then the continuation is deadlocked. $(A1)$ enables mapping the non-deterministic choice at the level of NetKAT to the setting of {\DNK}.

The axioms encoding the restriction operator $\delta_{\mathcal{L}}(-)$ and the projection operator $\pi_n(-)$ are defined in the standard fashion, on top of {\DNK} normal forms later defined in this section.
Intuitively, normal forms are defined inductively, as sums of complete tests and complete assignments $\alpha \cdot \pi$, or communication steps $x?q,\, x!q$ and $\recp{x,q}$, followed by arbitrary {\DNK} policies.

Last, but not least, $(AIP)$ corresponds to the so-called ``Approximation Induction Principle'', and it provides a mechanism for reasoning on the equivalence of recursive behaviours, up to a certain limit denoted by $n$.

\subsubsection{Soundness and Completeness}\label{sec:sound-complete}
In what follows, we show that the axiom system $\EDNK$ is sound and ground-complete with respect to {\DNK} bisimilarity.

We proceed by first defining a notion of normal forms of {\DNK} terms, together with a notion of guardedness and a statement about the branching finiteness of guarded {\DNK} processes.

\begin{lem}[{$\NetKATnoDup$} Normal Forms]\label{lm:netkat-nf}
We call a $\NetKATnoDup$ policy $q$ in \emph{normal form} (n.f.) whenever $q$ is of shape
$
\Sigma_{\alpha \cdot \pi \in \mathcal{A}} \alpha \cdot \pi
$
with $\mathcal{A} = \{\alpha_i \cdot \pi_i \mid i \in I\}$.
For every $\NetKATnoDup$ policy $p$ there exists a $\NetKATnoDup$ policy $q$ in n.f. such that
$\ENK \vdash p \equiv q$.
\end{lem}
\begin{proof}
The result follows by Lemma~$4$ in~\cite{netkat}, stating that:
\begin{equation}\label{eq:Gp}
\llbracket p \rrbracket = \bigcup_{x \in G(p)} \llbracket x \rrbracket
\end{equation}
where $G(p)$ defines the language model of NetKAT terms. Let $A$ be the set of all complete tests, and $\Pi$ be the set of all complete assignments. Similarly to~\cite{netkat}, we consider network packets with values in finite domains. Consequently, $A$ and $\Pi$ are finite.
In~\cite{netkat}, $G(p)$ is defined as a set with elements in $A \cdot (\Pi \cdot \textbf{dup})^* \cdot \Pi$. Recall that, in our setting, we work with the $\textbf{dup}$-free fragment of NetKAT. Hence, $G(p)$ is a finite set of shape $G = \{\alpha_i \cdot \pi_i \mid i \in I, \alpha_i \in A, \pi_i \in \Pi\}$. Based on the definition of $\llbracket-\rrbracket$ and~(\ref{eq:Gp}) it follows that:
\begin{equation}\label{eq:Gp2}
\llbracket p \rrbracket = \llbracket \Sigma_{\alpha \cdot \pi \in G} \alpha \cdot \pi \rrbracket
\end{equation}
Therefore, by the completeness of NetKAT, it holds that:
$\ENK \vdash p \equiv \Sigma_{\alpha \cdot \pi \in G} \alpha \cdot \pi$. In other words, $p$ can be reduced to a term in n.f.
\end{proof}

\begin{defn}[$\DNK$ Normal Forms]\label{def:dnk-nf}
We call a $\DNK$ policy in \emph{normal form} (n.f.) if it is of shape
\[
\Sigma^{\oplus}_{i \in I} (\alpha_i \cdot \pi_i) ; d_i \oplus \Sigma^{\oplus}_{j \in J} c_j ; d_j \,(\oplus \bot)
\]
where $d_i, d_j$ range over {\DNK} policies and $c_j ::= x?q \mid x!q \mid \recp{x,q}$ with $q$ denoting terms in $\NetKATnoDup$.
\end{defn}

\begin{defn}[Guardedness]\label{def::guard}
A {\DNA} policy $p$ is \emph{guarded} if and only if all occurrences of all variables $X$ in $p$ are \emph{guarded}. An occurrence of a variable $X$ in a policy $p$ is \emph{guarded} if and only if (i) $p$ has a subterm of shape $p' ; t$ such that
either $p'$ is variable-free, or all the occurrences of variables $Y$ in $p'$ are guarded,
and $X$ occurs in $t$, or (ii) if $p$ is of shape $y?X;t$, $y!X;t$ or $\recp{X,t}$.
\end{defn}

\begin{lem}[Branching Finiteness]\label{lm:fin-br}
All guarded {\DNK} policies are finitely branching.
\end{lem}

\begin{lem}[$\DNK$ Normalization]\label{lm:dnk-norm}
$\EDNK$ is normalising for {\DNK}. In other words,
for every guarded {\DNK} policy $p$ there exists a {\DNK} policy $q$ in n.f. such that $\EDNK \vdash p \equiv q$.
\end{lem}
\begin{proof}
The proof follows from Lemma~\ref{lm:netkat-nf} and $(A1):\,(z+y);p \equiv z ; p \oplus y ; p$ in a standard fashion, by structural induction.\\
\emph{Base cases.}
\begin{itemize}
    \item $p \triangleq \bot$ trivially holds
    \item $p \triangleq q ; d$ with $q$ a $\NetKATnoDup$ term holds by Lemma~\ref{lm:netkat-nf} and $(A1)$
    \item $p \triangleq c ; d$ with $c ::= x?q \mid x!q \mid \recp{x,q}$ trivially holds
\end{itemize}
\emph{Induction step.}
\begin{itemize}
    \item $p \triangleq p_1 \oplus p_2$ ~~~~~~~~$p \triangleq X$ - case discarded, as $p$ is not guarded
    \item $p \triangleq p_1 \llfloor p_2$ ~~~~~~~~~~$p \triangleq \pi_n(')$
    \item $p \triangleq p_1 \,\mid\, p_2$ ~~~~~~~~$p \triangleq \delta_{\mathcal{L}}(p')$
    \item $p \triangleq p_1 \Par p_2$
\end{itemize}
All items above follow by the axiom system $\EDNK$ and the induction hypothesis, under the assumption that $p_1, p_2$ and $p'$ are guarded.
\end{proof}

For simplicity, in what follows, we assume that {\DNK} policies are guarded.

\begin{lem}[Soundness of $\EDNKnoAIP$]\label{lm:sound-noAIP}
Let $\EDNKnoAIP$ stand for the axiom system $\EDNK$ in Figure~\ref{fig:axiom-ccp}, without the axiom $(AIP)$. $\EDNKnoAIP$ is sound for {\DNK} bisimilarity.
\end{lem}
\begin{proof}
The proof reduces to showing that for all $p$, $q$ {\DNK} policies, the following holds:
If $\EDNKnoAIP \vdash p \equiv q$ then $p \sim q$. 
This is proven in a standard fashion, by case analysis on transitions of shape
\[
(p, H_0, H'_0) \xrightarrow{\gamma} (q, H_1, H'_1)
\]
with $\gamma ::= (\sigma, \sigma') \mid x?n \mid x!n \mid \rec{x,n} $, according to the semantic rules in Figure~\ref{fig::sem-queue-sem}, (\ref{eq:delta}), (\ref{eq:pi}), (\ref{eq:op-recp-axiom}) and (\ref{eq:aux-op-sem}).

For an example, consider $(A1)$ and $(A12)$ in Figure~\ref{fig:axiom-ccp}; the proof of soundness for these axioms are given in the following. The soundness proofs for the rest of the axioms are provided in Appendix~\ref{sec::soundness}.
\begin{itemize}
	
	\item{
	Axiom under consideration:
	\begin{equation}
	(z + y) \Seq p \equiv z \Seq p \oplus y \Seq p \quad (A1)
	\end{equation}
	for $z, y \in {\NetKATnoDup}$ and $p \in {\DNK}$. Consider an arbitrary but fixed network packet $\sigma$,
	let $S_z \triangleq \llbracket z \rrbracket(\sigma \!\!::\!\! \langle \rangle)$, $S_y \triangleq \llbracket y \rrbracket(\sigma \!\!::\!\! \langle \rangle)$ and $S_{zy} \triangleq \llbracket z + y \rrbracket(\sigma \!\!::\!\! \langle \rangle)$. According to the semantic rules of {\DNK}, the derivations of the term $(z + y) \Seq p$ are as follows:
	\begin{enumerate}[leftmargin=.5in,label=(\alph*)]
		\item{\label{enum::a1-a}
		\[
			\begin{array}{r@{}l@{}c@{}r@{}l@{}c@{}r@{}l@{}}
			\textnormal{For all } \sigma' \in S_{zy}: \quad \dedr{\cpolseqsucc{}} &\sosrule{}{ ( (z + y) ; p, \sigma :: H, H')  \trans{(\sigma, \sigma')} ( p, H , \sigma' :: H')}\\
			
			\end{array}
		\]
		}
	\end{enumerate}

	Accordingly, the derivations of the term $z \Seq p \oplus y \Seq p$ are as follows:
	\begin{enumerate}[leftmargin=.5in,label=(\alph*)]
		\setcounter{enumi}{1}
		\item{\label{enum::a1-b}
		\[
			\begin{array}{r@{}r@{}l@{}c@{}r@{}l@{}c@{}r@{}l@{}}
			\textnormal{For all } \sigma' \in S_z:\quad\quad & \dedr{\cpolseqsucc{}} &\sosrule{}{ ( z ; p, \sigma :: H, H')  \trans{(\sigma, \sigma')} ( p, H , \sigma' :: H')}\\
			
			&\dedr{\cpoloplusl}&\sosrule{}{ (z ; p \oplus y ; p, \sigma :: H, H')  \trans{(\sigma, \sigma')} ( p, H, \sigma' :: H')}
			\end{array}
		\]
		}
		\item{\label{enum::a1-c}
			\[
				\begin{array}{r@{}r@{}l@{}c@{}r@{}l@{}c@{}r@{}l@{}}
				\textnormal{For all } \sigma' \in S_y:\quad\quad & \dedr{\cpolseqsucc{}} &\sosrule{}{ ( y ; p, \sigma :: H, H')  \trans{(\sigma, \sigma')} ( p, H , \sigma' :: H')}\\
				
				&\dedr{\cpoloplusr}&\sosrule{}{ (z ; p \oplus y ; p, \sigma :: H, H')  \trans{(\sigma, \sigma')} ( p, H, \sigma' :: H')}
				\end{array}
			\]
		}
	\end{enumerate}

As demonstrated in \ref{enum::a1-a} and \ref{enum::a1-b}, \ref{enum::a1-c}, both of the terms $(z + y) \Seq p$ and $z \Seq p \oplus y \Seq p$ initially only afford a transition of shape $(\sigma, \sigma')$ and they converge into the same expression after taking that transition:
	\begin{alignat}{2}
	((z + y) ; p, \sigma :: H, H') \trans{(\sigma, \sigma')}& (p, H, \sigma' :: H') \\
	(z \Seq p \oplus y \Seq p, \sigma :: H, H') \trans{(\sigma, \sigma')}& (p, H, \sigma' :: H')
	\end{alignat}
	In the case of the term $(z + y) \Seq p$, the possible values for the $\sigma'$ ranges over $S_{zy}$. Whereas for the term $(z + y) \Seq p$, the possible values for the $\sigma'$ ranges over $S_{z} \cup S_{y}$.
	However, observe that $S_{zy}$ is equal to $S_{z} \cup S_{y}$:
	\begin{alignat}{2}
	S_{zy} &= \llbracket z + y \rrbracket(\sigma \!\!::\!\! \langle \rangle) & (\textnormal{Definition of } S_{zy}) \\ 
	&= \llbracket z \rrbracket(\sigma \!\!::\!\! \langle \rangle) \cup \llbracket y \rrbracket(\sigma \!\!::\!\! \langle \rangle) \quad & (\textnormal{Definition of } +) \\ 
	&= S_z \cup S_y &  (\textnormal{Definition of } S_{z} \textnormal{ and } S_{y})
	\end{alignat}
	Hence, it is straightforward to conclude that the following holds:
	\begin{equation}
	((z + y) \Seq p) \sim (z \Seq p \oplus y \Seq p)
	\end{equation}
}

\item{
	Axiom under consideration:
	\begin{equation}
	(x?\mathit{z} \Seq p) \mid (x!\mathit{z} \Seq q) \equiv \recp{x,z} \Seq (p \Par q) \quad (A12)
	\end{equation}

	for $p, q \in {\DNK}$. The derivations of the term $(x?\mathit{z} \Seq p) \mid (x!\mathit{z} \Seq q)$ are as follows:
	\begin{enumerate}[leftmargin=.5in,label=(\alph*)]
		\item{\label{enum::a12-a}
		\[
			\begin{array}{r@{}l@{}c@{}r@{}l@{}c@{}r@{}l@{}}
			\dedr{\cpolmsgrec} & \sosrule{}{ (x?z ; p, H, H')  \trans{x?z} (p, H, H')} &
			\quad \quad &
			\dedr{\cpolmsgsend}&\sosrule{}{ (x!z ; q, H, H')  \trans{x!z} ( q, H, H')}\\
			
			\dedr{\mid^{?!}} & \multicolumn{6}{@{}l@{}}{ \sosrule{\myquad[30]}{((x?\mathit{z} \Seq p) \mid (x!\mathit{z} \Seq q), H, H') \trans{\rec{x,z}} (p \Par q, H, H')}}
			\end{array} 
		\]
		}
	\end{enumerate}

The derivations of the term $\recp{x,z} \Seq (p \Par q)$ are as follows:
	\begin{enumerate}[leftmargin=.5in,label=(\alph*)]
		\setcounter{enumi}{1}
		\item{\label{enum::a12-b}
			\[
				\begin{array}{r@{}l@{}c@{}r@{}l@{}c@{}r@{}l@{}}
				 \dedr{\recp{x,z}} & \sosrule{}{(\recp{x,z} \Seq (p \Par q), H, H') \trans{\rec{x,z}} (p \Par q, H, H')} 
				\end{array} 
			\]
		}
	\end{enumerate}
As demonstrated in \ref{enum::a12-a} and \ref{enum::a12-b}, both of the terms $(x?\mathit{z} \Seq p) \mid (x!\mathit{z} \Seq q)$ and $\recp{x,z} \Seq (p \Par q)$ initially only afford the transition $\rec{x, z}$ and they converge into the same expression after taking that transition:
\begin{alignat}{2}
((x?\mathit{z} \Seq p) \mid (x!\mathit{z} \Seq q), H, H') \trans{\rec{x,z}}& (p \Par q, H, H')\\ 
(\recp{x,z} \Seq (p \Par q), H, H') \trans{\rec{x,z}}& (p \Par q, H, H')
\end{alignat}
Hence, it is straightforward to conclude that the following holds:
\begin{equation}\label{eq::sound-cm9-conc}
((x?\mathit{z} \Seq p) \mid (x!\mathit{z} \Seq q)) \sim (\recp{x,z} \Seq (p \Par q)).
\end{equation}
}
\end{itemize}


\end{proof}

\begin{lem}[Soundness of $AIP$]\label{lem:sound-aip}
The Approximation Induction Principle $(AIP)$ is sound for {\DNK} bisimilarity.
\end{lem}
\begin{proof}
The proof is close to the one of Theorem~$2.5.8$ in~\cite{DBLP:books/daglib/0069083} and uses the branching finiteness property of {\DNK} policies in Lemma~\ref{lm:fin-br}. Assume two {\DNK} policies $p, p'$ such that
\begin{equation}\label{eq:aip-1}
    \forall n \in \mathbb{N} : \pi_n(p) \equiv \pi_n(p')
\end{equation}
By Lemma~\ref{lm:sound-noAIP} it follows that 
\begin{equation}\label{eq:aip-2}
    \forall n \in \mathbb{N} : \pi_n(p) \sim \pi_n(p')
\end{equation}
We want to prove that $p \sim p'$. The idea is to build a bisimulation relation $R$ such that $(p, p') \in R$. We define $R$ as follows:
\begin{equation}\label{eq:aip-3}
    R = \{(t, t') \mid \forall n \in \mathbb{N} : \pi_n(t) \sim \pi_n(t')\}
\end{equation}
Without loss of generality, assume that $p$ and $p'$ are in n.f. Assume $(p, p') \in R$ and
\begin{equation}\label{eq:aip-4}
    (p, H_0, H'_0) \xrightarrow{\gamma} (p_1, H_1, H'_1)
\end{equation}
Next, for all $n > 0$, define
\begin{equation}\label{eq:aip-5}
    S_n = \{p'_1 \mid (p', H_0, H'_0) \xrightarrow{\gamma} (p'_1, H_1, H'_1) \textnormal{ and }  \pi_n(p_1) \sim \pi_n(p'_1)\}
\end{equation}
The following hold:
\begin{enumerate}
    \item $S_1 \supseteq S_2 \ldots$ as if $\pi_{n+1}(p) \sim \pi_{n+1}(p')$ then $\pi_{n}(p_1) \sim \pi_{n}(p'_1)$. The latter is a straightforward result derived according to the definition of $\sim$ and the semantics of $\pi_n(-)$, under the assumption that $p, p'$ are in n.f.
    \item $S_n \not = \emptyset$ for all $n \geq 1$ since $\pi_{n+1}(p) \sim \pi_{n+1}(p')$ by~(\ref{eq:aip-2}) and $(p, H_0, H'_0) \xrightarrow{\gamma} (p_1, H_1, H'_1)$ according to~(\ref{eq:aip-4})
    \item $S_n$ is finite, for all $n \in \mathbb{N}$, as $p'$ is finitely branching according to Lemma~\ref{lm:fin-br} 
\end{enumerate}
Hence, the sequence $S_1, S_2, \ldots$ remains constant from some $n$ onward and $\cap_{n \geq 0} S_n \not = \emptyset$. Let $p'_1 \in \cap_{n \geq 0} S_n$. It holds that:
\begin{itemize}
    \item $ (p', H_0, H'_0) \xrightarrow{\gamma} (p'_1, H_1, H'_1)$
    \item $(p_1, p'_1) \in R$ by the definition of $R$ and $S_n$
\end{itemize}
Symmetrically to~(\ref{eq:aip-4}), assume $(p, p') \in R$ and $(p', H_0, H'_0) \xrightarrow{\gamma} (p'_1, H_1, H'_1)$. By following a similar reasoning, we can show that:
\begin{itemize}
    \item $ (p, H_0, H'_0) \xrightarrow{\gamma} (p_1, H_1, H'_1)$
    \item $(p_1, p'_1) \in R$ by the definition of $R$ and $S_n$
\end{itemize}
Hence, $R$ is a bisimulation relation and $p \sim p'$.
\end{proof}

\begin{theorem}[Soundness \& Completeness]\label{thm:sound-complete}
$\EDNK$ is sound and ground-complete for {\DNK} bisimilarity.
\end{theorem}
\begin{proof}
Soundness:
if $\EDNK \vdash p \equiv q$ then $p \sim q$,
follows from Lemma~\ref{lm:sound-noAIP} and Lemma~\ref{lem:sound-aip}.

Completeness: if $p \sim q$ then $\EDNK \vdash p \equiv q$, is shown as follows. Without loss of generality, assume $p$ and $q$ are in n.f., according to Lemma~\ref{lm:dnk-norm}. We want to show that:
\begin{equation}\label{eq:comp-1}
\begin{array}{rcl}
    p & \equiv & q \oplus p\\
    q & \equiv & p \oplus q
\end{array}
\end{equation}
which, by ACI of $\oplus$ implies $p \equiv q$. This reduces to showing that every summand of $p$ is a summand of $q$ and vice-versa. We first argue that every summand of $p$ is a summand of $q$. The reasoning is by structural induction.\\
\emph{Base case.}
\begin{itemize}
    \item $p \triangleq \bot$. It holds by the hypothesis $p \sim q$ that $q \triangleq \bot$.
\end{itemize}
\emph{Induction step.}
\begin{itemize}
    \item $p \triangleq ((\alpha \cdot \pi);p') \oplus p''$. Then, $(p, \sigma_{\alpha} :: H, H' ) \xrightarrow{(\sigma_{\alpha}, \sigma_{\pi})} (p',  H, \sigma_{\pi}::H' )$ implies by the hypothesis $p \sim q$ that $(q, \sigma_{\alpha} :: H, H' ) \xrightarrow{(\sigma_{\alpha}, \sigma_{\pi})} (q',  H, \sigma_{\pi}::H' )$ and $p' \sim q'$. Recall that $q$ is in n.f.; hence, by the shape of the semantic rules in Figure~\ref{fig::sem-queue-sem} it holds that $q \triangleq ((\alpha \cdot \pi);q') \oplus q''$. By the induction hypothesis, it holds that $p' \equiv q'$ hence, $(\alpha \cdot \pi);p'$ is a summand of $q$ as well.
    \item Cases $p \triangleq (c;p') \oplus p''$ with $c ::= x?n \mid x!n \mid \recp{x,n}$ follow in a similar fashion.
\end{itemize}
Hence, $p \equiv q \oplus p$ holds. The symmetric case $q \equiv p \oplus q$ follows the same reasoning.
\end{proof}
\end{ARXIV}
\section{A Framework for Safety}\label{sec:safety}

\begin{ICALP}
In this section we provide a language for specifying safety properties of {\DNK} networks, together with a procedure for reasoning about safety in an equational fashion. Intuitively, safety properties enable specifying the absence of undesired network behaviours.

\begin{defn}[Safety Properties - Syntax]\label{def:safety-syntax-ext}
Let $\mathcal{A}$ be an alphabet over letters of shape $\alpha \cdot \pi$ and $\recp{x,p}$, with $\alpha$ and $\pi$ ranging over complete tests and assignments, and $\recp{x,p}$ ranging over reconfiguration actions.
{Safety properties} are defined in the following fashion:
\[
\begin{array}{rcl}
     \mathit{act}  & ::= & \alpha \cdot \pi \mid \recp{x, p} \mid \neg \mathit{act} ~(\alpha \cdot \pi,\, \recp{x, p} \in \mathcal{A}) \\
     \mathit{regexp} & ::= & \textit{true} \mid \mathit{act} \mid \mathit{regexp} + \mathit{regexp} \mid \mathit{regexp} \cdot \mathit{regexp}
     \mid \\ && (\mathit{regexp})^n~~(\textnormal{with } n \geq 1)
     \\
     \mathit{prop} & ::= & [\mathit{regexp}] \mathit{false}
\end{array}
\]
\end{defn}

 A safety property specification $\mathit{prop}$ is satisfied whenever the behaviour encoded by $\textit{regexp}$ should not be observed within the network. Regular expressions $\textit{regexp}$ are defined with respect to actions $\mathit{act}$: a flow of shape $\alpha \cdot \pi$ is the observable behaviour of a (${\NetKATnoDup}$) policy transforming a packet encoded by $\alpha$ into $\alpha_{\pi}$, whereas $\recp{x, p}$ corresponds to a reconfiguration step in a network. Recursively, a sum of regular expressions $\mathit{regexp}_1 + \mathit{regexp}_2$ encodes the union of the two behaviours, a concatenation of regular expressions $\mathit{regexp}_1 \cdot \mathit{regexp}_2$ encodes the behaviour of $\mathit{regexp}_1$ followed by the behaviour of $\mathit{regexp}_2$. A property of shape $[\neg a] \textit{false}$, with $a \in \mathcal{A}$, states that the system cannot do anything apart from $a$ as a first step.
The property $[\textit{true}]\textit{false}$ states that no action can be observed in the network, whereas
$[\mathit{r}^n]\textit{false}$ encodes the repeated application of $r$ for $n$ times.

The full semantics of safety property is provided in~\cite{arxiv}. In short, the semantic map
$
\llbracket - \rrbracket : \mathit{Prop} \rightarrow \mathit{\DNK}
$
associates to each property $[\textit{regexp}]\textit{false}$ as in Definition~\ref{def:safety-syntax-ext}, a {\DNK} process that can only execute sequences of actions outside the scope of $\textit{regexp}$.
Typically, safety analysis is reduced to reachability analysis. In our context, a safety property is violated whenever the network system under analysis displays a (finite) execution that is not in the behaviour of the property. Thus, the aforementioned semantic map is based on traces (or words in $\mathcal{A}^*$) and is not sensitive to branching.
This paves the way to reasoning about the satisfiability of safety properties in an equational fashion.

\begin{defn}[Safe Network Systems]\label{def:safety}
Let $\EDNKtr$ stand for the equational axioms in Figure~\ref{fig:axiom-ccp}, including the additional axiom that enables switching from the context of bisimilarity to trace equivalence of {\DNK} policies, namely:
$
    p ; (q \oplus r) \equiv p ; q \oplus p ; r.
$
Assume a specification given as the safety formula $s$ and a network system implemented as the {\DNK} policy $i$. We say that the network is \emph{safe} whenever the following holds:
$
    \EDNKtr \vdash \llbracket s \rrbracket \oplus i \equiv \llbracket s \rrbracket.
$
In words: checking whether $i$ satisfies $s$ reduces to checking whether the trace behaviour of $i$ is included into that of $s$. 
\end{defn}

For an example, consider the firewall in Figure~\ref{fig:stateful1} and the corresponding encoding in Figure~\ref{fig:stateful2}. Recall that reaching $\textit{int}$ from $\textit{ext}$ without observing a secure connection request is a faulty behaviour. This entails the safety formula $s_n$ defined as $[(\neg \recp{\textit{secConReq},1})^n \cdot (\alpha \cdot \pi)]\textit{false}$, for $n \in \mathbb{N}$, $\alpha \triangleq (\mathit{port}=\textit{ext})$ and $\pi \triangleq (\mathit{port} \leftarrow \textit{int})$.
Therefore, checking whether the network is safe reduces to checking, for all $n \in \mathbb{N}$:
$
    \EDNKtr \vdash \llbracket s_n \rrbracket \oplus  \mathit{Init} \equiv \llbracket s_n \rrbracket
$.
Note that, for a fixed $n$, the verification procedure resembles bounded model checking~\cite{DBLP:books/daglib/0020348}.

\end{ICALP}

\begin{ARXIV}
In this section we provide a language for specifying safety properties of {\DNK} networks, together with a procedure for reasoning about safety in an equational fashion. Intuitively, safety properties enable specifying undesired network behaviours.

\begin{defn}[Safety Properties - Syntax]\label{def:safety-syntax}
Let $\mathcal{A}$ be an alphabet over letters of shape $\alpha \cdot \pi$ and $\rec{x,p}$, with $\alpha$ and $\pi$ ranging over complete tests and assignments as in Definition~\ref{def:CTA}, and $\rec{x,p}$ ranging over reconfiguration actions.
A \emph{safety property} $\mathit{prop}$ is defined as:
\[
\begin{array}{rcl}
     \mathit{act} & ::= & \alpha \cdot \pi \mid \recp{x, p} ~~~~~~~(\textnormal{where } \alpha \cdot \pi,\, \recp{x, p} \in \mathcal{A}) \\
     \mathit{regexp} & ::= & \mathit{act} \mid \mathit{regexp} + \mathit{regexp} \mid \mathit{regexp} \cdot \mathit{regexp}
     \\
     \mathit{prop} & ::= & [\mathit{regexp}] \mathit{false}
\end{array}
\]
\end{defn}

The intuition behind Definition~\ref{def:safety-syntax} is as follows. A safety property specification $\mathit{prop}$ is satisfied whenever the behaviour encoded by $\textit{regexp}$ cannot be observed within the network. Regular expressions $\textit{regexp}$ are defined with respect to actions $\mathit{act}$: a flow of shape $\alpha \cdot \pi$ is the observable behaviour of a (${\NetKATnoDup}$) policy transforming a packet encoded by $\alpha$ into $\alpha_{\pi}$, whereas $\recp{x, p}$ corresponds to a reconfiguration step in a network. Recursively, a sum of regular expressions $\mathit{regexp}_1 + \mathit{regexp}_2$ encodes the union of the two behaviours, a concatenation of regular expressions $\mathit{regexp}_1 \cdot \mathit{regexp}_2$ encodes the behaviour of $\mathit{regexp}_1$ followed by the behaviour of $\mathit{regexp}_2$.

\begin{defn}
[Head Normal Forms for Safety]\label{def:hnf-safe}
Let $\mathcal{A}$ be an alphabet over letters of shape $\alpha \cdot \pi$ and $\rec{x,p}$, with $\alpha$ and $\pi$ ranging over complete tests and assignments as in Definition~\ref{def:CTA}, and $\rec{x,p}$ ranging over reconfiguration actions. We write $w, w'$ for (non-empty) words with letters in $\mathcal{A}$ (i.e., $w, w' \in \mathcal{A}^*$) and $\mid w \mid$ for the length of $w$.
We write $w' \preceq w$ whenever $w'$ is a prefix of $w$ (including $w$).

Let $r$ be a regular expression ($\textit{regexp}$) as in Definition~\ref{def:safety-syntax}.
We call \emph{head normal form} of $r$, denoted by $\hnf{r}$, the sum of words obtained by distributing $\cdot$ over $+$ in $r$, in the standard fashion:
    \[
    \begin{array}{rcl}
      \hnf{a} & \triangleq & a ~~~(a \in \mathcal{A})\\
      \hnf{w} & \triangleq & w ~~~(w \in \mathcal{A^*})\\
      \hnf{r_1 + r_2} & \triangleq & \hnf{r_1} + \hnf{r_2}\\
      \hnf{r_1 \cdot (r_2 + r_3)} & \triangleq & \hnf{r_1 \cdot r_2} +  \hnf{r_1 \cdot r_3}\\
      \hnf{(r_1 + r_2) \cdot r_3} & \triangleq & \hnf{r_1 \cdot r_3} + \hnf{r_1 \cdot r_3}\\
      \hnf{r' \cdot (r_1 + r_2) \cdot r''} & \triangleq & \hnf{r' \cdot r_1 \cdot r''} + \hnf{r' \cdot r_2 \cdot r''}
    \end{array}
    \]
\end{defn}

Next, we give the formal semantics of safety properties.

\begin{defn}[Safety Properties - Semantics]\label{def:safety-semantics}
Let $\mathit{Prop}$ stand for the set of all properties as in Definition~\ref{def:safety-syntax}. The semantic map
$
\llbracket - \rrbracket : \mathit{Prop} \rightarrow \mathit{\DNK}
$
associates to each safety property in $\mathit{Prop}$ a $\mathit{\DNK}$ expression as follows.

Let $\Theta$ be the {\DNK} policy (in normal form) encoding all possible behaviours over $\mathcal{A}$: 
\[
\Theta \triangleq \Sigma^{\oplus}_{\alpha \cdot \pi \in \mathcal{A}}(\alpha \cdot \pi ; \bot \oplus \alpha \cdot \pi ; \Theta) \,\oplus\, \Sigma^{\oplus}_{\recp{x,p} \in \mathcal{A}}(\recp{x,p} ; \bot \oplus \recp{x,p} ; \Theta)
\]
Then:
\[
\begin{array}{lrcl}
     (c_{1}) & \llbracket\, [\Sigma_{\footnotesize \begin{array}{l} i \in I \\ w_i \in A^* \end{array}} \!\!\!\!w_i]\mathit{false}\,\rrbracket & \triangleq &
     \Sigma^{\oplus}_{\footnotesize \begin{array}{l} w \in \mathcal{A}^* \\ \mid w \mid < M \\ \forall i \in I : w_i \not \preceq w  \end{array}} \hspace{-30pt}\overline{w} ; \bot
     ~~~\oplus~~~
     \Sigma^{\oplus}_{\footnotesize \begin{array}{l} w \in \mathcal{A}^* \\ \mid w \mid = M \\ \forall i \in I : w_i \not \preceq w \end{array}} \hspace{-30pt}(\overline{w} ; \bot \,\oplus\, \overline{w};\Theta)\\\\[1ex]
     
     (c_2) & \llbracket [r] \textit{false} \rrbracket  & \triangleq & \llbracket [\hnf{r}] \textit{false} \rrbracket ~~~~\textnormal{[otherwise]}

\end{array}
\]
such that $M$ is the length of the longest word $w_i$, with $i \in I$, and $\overline{w}$ is a {\DNK}-compatible term obtained from $w$ where all letters have been separated by $;$ and inductively defined in the obvious way:
    \[
    \begin{array}{rcl}
         \overline{a} & \triangleq & a ~~~~(a \in \mathcal{A})  \\
         \overline{a \cdot w} & \triangleq & a ; \overline{w} ~~~~(a \in \mathcal{A}, w \in \mathcal{A^*}) 
    \end{array}
    \]
\end{defn}


The semantic map $
\llbracket - \rrbracket : \mathit{Prop} \rightarrow \mathit{\DNK}
$ is defined in accordance with the intuition provided in the beginning of this section.
For instance, as shown in $(c_1)$, if none of the sequences of steps $w_i$ can be observed in the system, then the associated {\DNK} term
prevents the immediate execution of all $w_i$.
%
Typically, safety analysis is reduced to reachability analysis. Intuitively, in our context, a safety property is violated whenever the network system under analysis displays a (finite) execution that is not in the behaviour of the property. Thus, the semantic map in Definition~\ref{def:safety-semantics} is based on traces (or words in $\mathcal{A}^*$) and is not sensitive to branching; see the use of head normal forms in $(c_2)$.

With these ingredients at hand, we can reason about the satisfiability of safety properties in an equational fashion.

\begin{defn}[$\EDNKtr$]\label{def:EDNKtr}
Let $\EDNKtr$ stand for the equational axioms in Figure~\ref{fig:axiom-ccp}, including the additional axiom that enables switching from the context of bisimilarity to trace equivalence of {\DNK} policies, namely:
\begin{equation}\label{eq:BisimToTr}
    p ; (q \oplus r) \equiv p ; q \oplus p ; r~~~~(A_{16})
\end{equation}
\end{defn}

\begin{defn}[Safe Network Systems]\label{def:safety}
Assume a specification given as the safety formula $s$ and a network system implemented as the {\DNK} policy $i$. We say that the network is \emph{safe} whenever the following holds:
\begin{equation}\label{eq:check-safety}
    \EDNKtr \vdash \llbracket s \rrbracket \oplus i \equiv \llbracket s \rrbracket
\end{equation}
In words: checking whether $i$ satisfies $s$ reduces to checking whether the trace behaviour of $i$ is included into that of $s$. 
\end{defn}

\subsection{Sugars for Safety}\label{sec:sugar-safe}
In this section we introduce a version of safety properties extended with negated actions ($\neg (\alpha \cdot \pi)$ and, respectively, $\neg \recp{x,p}$), the $\textit{true}$ construct and repetitions ($r^n$), equally expressive but enabling more concise property specifications.

\begin{defn}[Safety Properties - Extended Syntax]\label{def:safety-syntax-ext}
Let $\mathcal{A}$ be an alphabet over letters of shape $\alpha \cdot \pi$ and $\recp{x,p}$, with $\alpha$ and $\pi$ ranging over complete tests and assignments as in Definition~\ref{def:CTA}, and $\recp{x,p}$ ranging over reconfiguration actions.
{Safety properties} are extended in the following fashion:
\[
\begin{array}{rcl}
     \mathit{act}_e & ::= & \alpha \cdot \pi \mid \recp{x, p} \mid \neg \mathit{act}_e ~~~~~~~(\textnormal{with } \alpha \cdot \pi,\, \recp{x, p} \in \mathcal{A}) \\
     \mathit{regexp}_e & ::= & \textit{true} \mid \mathit{act}_e \mid \mathit{regexp}_e + \mathit{regexp}_e \mid \mathit{regexp}_e \cdot \mathit{regexp}_e
     \mid (\mathit{regexp}_e)^n~~(\textnormal{with } n \geq 1)
     \\
     \mathit{prop}_e & ::= & [\mathit{regexp}_e] \mathit{false}
\end{array}
\]
\end{defn}

Intuitively, 
a property of shape $[\neg a] \textit{false}$, with $a \in \mathcal{A}$, states that the system cannot do anything apart from $a$ as a first step.
The property $[\textit{true}]\textit{false}$ states that no action can be observed in the network, whereas
$[\mathit{r}^n]\textit{false}$ encodes the repeated application of $r$ for $n$ times.

Let $\mathit{Reg}$ and, respectively, $\mathit{Reg}_e$ denote the set of regular expressions $\mathit{regexp}$ in Definition~\ref{def:safety-syntax} and, respectively, the set of regular expressions $\mathit{regexp}_e$ in Definition~\ref{def:safety-syntax-ext}. The ``desugaring'' function defining the regular equivalent of the extended safety properties is defined as follows:
\[
\begin{array}{rcl}
     \multicolumn{3}{c}{ds : \mathit{Reg}_e \rightarrow \mathit{Reg}}\\
     ds(\textit{true}) & \triangleq & \Sigma_{a \in \mathcal{A}} a\\
     ds(\neg (\alpha \cdot \pi)) & \triangleq &  \Sigma_{\footnotesize \begin{array}{c} \alpha_i \cdot \pi_i \in \mathcal{A} \\ \alpha_i \not= \alpha \\ or \\ \pi_i \not = \pi \end{array}}\alpha_i \cdot \pi_i\\
     ds(\neg \recp{x,p}) & \triangleq & \Sigma_{\footnotesize \begin{array}{c} \recp{y,q} \in \mathcal{A} \\ \recp{y,q} \not= \recp{x,p} \end{array}}\!\!\!\!\recp{y,q}\\
     ds(r^n) & \triangleq & ds(\underbrace{r \cdot r \cdot \ldots \cdot r}_\text{$n$ times}) \\
     ds(r_1 \cdot r_2) & \triangleq & ds(r_1) \cdot ds(r_2)~~~\textnormal{if } r_1 \cdot r_2 \not \in \textit{Reg}\\
     ds(r_1 + r_2) & \triangleq & ds(r_1) + ds(r_2) ~~\textnormal{if } r_1 + r_2 \not \in \textit{Reg}\\
     ds(r) & \triangleq & r ~~~\textnormal{[otherwise]}
\end{array}
\]
The (overloaded) semantic map $
\llbracket - \rrbracket : \mathit{Prop}_e \rightarrow \mathit{\DNK}
$ is defined as expected:
\[
\llbracket\, [r_e]\textit{false} \,\rrbracket \triangleq \llbracket\, [ds(r_e)]\textit{false} \,\rrbracket
\]

For an example, consider the distributed controllers in Figure~\ref{fig:distributedController1} and the corresponding encoding in Figure~\ref{fig:independent2}. Recall that reaching $H4$ from $S2$ is considered a breach in the system. This entails the safety formulae $s_n$ defined as $[(true)^n \cdot (\alpha \cdot \pi)]\textit{false}$, for $n \in \mathbb{N}$, $\alpha \triangleq (\mathit{port}=2)$ and $\pi \triangleq (\mathit{port} \leftarrow 15)$.
In words: no matter what sequence of events (of length $n$) is executed, $\alpha \cdot \pi$ cannot happen as the next step.
Therefore, checking whether the network is safe reduces to checking, for all $n \in \mathbb{N}$:
\begin{equation}\label{eq:eg-breach}
    \EDNKtr \vdash \llbracket s_n \rrbracket \oplus  \mathit{SDN} \equiv \llbracket s_n \rrbracket
\end{equation}

Note that, for a fixed $n$, the verification procedure resembles bounded model checking~\cite{DBLP:books/daglib/0020348}.

\end{ARXIV}
\section{Implementation}\label{sec:implementation}

In Section~\ref{sec:safety} we introduced a notion of safety for {\DNK} and provided a mechanism for reasoning about safety in an equational fashion, by exploiting {\DNK} trace semantics. 
To this end, we search for traces that violate the safety property, i.e.,  
we turn the equational reasoning about safety into checking for reachability  properties of shape $s \triangleq \langle \textit{regexp}\rangle \textit{true}$; for an implementation $i$, this is achieved by checking the following equation using our axiomatization: $\EDNKtr \vdash i \oplus \llbracket s \rrbracket \equiv i$.


We developed a prototype tool, called {\tool}\footnote{\url{https://github.com/hcantunc/DyNetiKAT}}, 
based on Maude~\cite{DBLP:conf/maude/ClavelDELMMT07n}, NetKAT decision procedure~\cite{netkat-decision} and Python~\cite{DBLP:conf/usenix/Rossum07}, for checking the aforementioned equation. We build upon the reachability checking method in NetKAT~\cite{netkat}. For a reminder: $\textit{out}$ is reachable from $\textit{in}$, in the context of a switch policy $p$ and topology $t$, whenever the following property is satisfied: $in \cdot (p \cdot t)^* \cdot out \not \equiv \drop$ (and vice-versa). Furthermore, we also consider waypointing properties. An intermediate point $w$ between $in$ and $out$ is considered a waypoint from $in$ to $out$ if all the packets from $in$ to $out$ go through $w$. Such a property is satisfied if the following equivalence holds~\cite{netkat}:
\begin{alignat}{2}
in \cdot (p \cdot t)^* \cdot out + in \cdot (\neg out \cdot p \cdot t)^* \cdot w \cdot (\neg in \cdot p \cdot t)^* \cdot out \notag \\\equiv in \cdot (\neg out \cdot p \cdot t)^* \cdot w \cdot (\neg in \cdot p \cdot t)^* \cdot out\notag
\end{alignat}
The inputs to our tool are a {\DNK} program $p$, an input predicate $in$, an output predicate $out$, a waypoint predicate $w$ for waypointing properties, and the equivalences that describe the desired properties. The checking of properties is achieved by analyzing the step by step behaviour of {\DNK} terms in normal form via a set of operators $head(D)$, and $tail(D, R)$, where $D$ is a {\DNK} term and $R$ is a set of terms of shape $\recp{X, N}$. Intuitively, the operator $head(D)$ returns a NetKAT policy which represents the current configuration in the input $D$ and the operator $tail(D, R)$ returns a {\DNK} policy which is the sum of {\DNK} policies inside $D$ that appear after the reconfiguration events in $R$.
\begin{ARXIV}
The operators $head$ and $tail$ are defined as follows:
\begin{alignat}{4}
    head(\bot) &= \drop{} & tail(\bot, R) &= \bot \\
    head(N ; D) &= N + head(D) & tail(N ; D, R) &= tail(D, R) \\
    head(D \oplus Q) &= head(D) + head(Q) & tail(D \oplus Q, R) &= tail(D, R) \oplus tail(Q, R) \\
	head(\recp{X, N} ; D) &= \drop{} & tail(\recp{X, N} ; D, R) &= D \oplus tail(D, R)~\textnormal{if}~ \recp{X, Z}\in R \\
	& &\quad\quad
	tail(\recp{X, N} ; D, R) &= \bot ~\textnormal{if}~ \recp{X, N} \not\in R
\end{alignat}
Note that we assume the input $D$ to the operators $head$ and $tail$ do not contain terms of shape $x?q$ and $x!q$. This can be ensured by always applying the restriction operator $\delta$ on $D$.
\end{ARXIV} 
After extracting the desired configurations by using the $head$ and $tail$ operators, we then utilize the NetKAT decision procedure for checking properties. For an example, consider the stateful firewall in Figure~\ref{fig:stateful1} and the corresponding encoding in Figure~\ref{fig:stateful2}. Assume that we would like to check if the packets at port $ext$  can arrive at port $int$ only after a proper reconfiguration on the channel $secConReq$. For this example, the analysis reduces to defining
the input predicate $in \triangleq port = ext$, the output predicate $out \triangleq port=int$, $R=\{\recp{secConReq, \packcopy{}}\}$, and the following properties:
\begin{alignat}{2}
in \cdot head(Init) \cdot out \equiv \drop{}\label{eq::sf-3}\\
in \cdot head(tail(Init, R)) \cdot out \not \equiv \drop{}\label{eq::sf-4}
\end{alignat}
Intuitively, the equivalence in~(\ref{eq::sf-3}) expresses that packets at port $ext$ are not able to reach to port $int$ in the initial configuration and~(\ref{eq::sf-4}) expresses that the configuration after the synchronization on the channel $secConReq$ allows this flow.

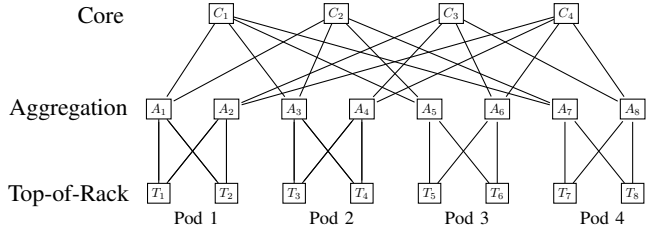
\begin{figure}[t]
    \centering
    \begin{tikzpicture}[>=stealth,
	shorten >=1pt,auto,node distance=1.8cm,switch node/.style={circle,draw}, host node/.style={scale=0.5,rectangle,draw}]
	\node[host node] (a1) {$A_1$};
	\node[host node] (a2) [right of=a1] {$A_2$};
	\node[host node] (a3) [right of=a2] {$A_3$};
	\node[host node] (a4) [right of=a3] {$A_4$};
	\node[host node] (a5) [right of=a4] {$A_5$};
	\node[host node] (a6) [right of=a5] {$A_6$};
	\node[host node] (a7) [right of=a6] {$A_7$};
	\node[host node] (a8) [right of=a7] {$A_8$};
	
	\node[host node] (c1) [above right=1cm and 0.5cm of a1] {$C_1$};
	\node[host node] (c2) [right=1.2cm of c1] {$C_2$};
	\node[host node] (c3) [right=1.2cm of c2] {$C_3$};
	\node[host node] (c4) [right=1.2cm of c3] {$C_4$};
	
	\node[host node] (s1) [below =0.85cm of a1] {$T_1$};
	\node[host node] (s2) [right of=s1] {$T_2$};
	\node[host node] (s3) [right of=s2] {$T_3$};
	\node[host node] (s4) [right of=s3] {$T_4$};
	
	\node[host node] (s5) [below =0.85cm of a5] {$T_5$};
	\node[host node] (s6) [right of=s5] {$T_6$};
	\node[host node] (s7) [right of=s6] {$T_7$};
	\node[host node] (s8) [right of=s7] {$T_8$};
	
	\node[] at (-1.2,0) {\small{Aggregation}};
	\node[] at (-0.77,1.25) {\small{Core}};
	\node[] at (-1.2,-1.15) {\small{Top-of-Rack}};
	\node[] at (0.5,-1.45) {\scriptsize{Pod 1}};
	\node[] at (2.3,-1.45) {\scriptsize{Pod 2}};
	\node[] at (4.1,-1.45) {\scriptsize{Pod 3}};
	\node[] at (5.9,-1.45) {\scriptsize{Pod 4}};
		
	\path[-]
	(s1) edge [pos=0.1] node [left] {} (a1)
	(s1) edge [pos=0.01] node [xshift=0.05cm, right] {} (a2)
	(s2) edge [pos=0.01] node [xshift=-0.05cm, left] {} (a1)
	(s2) edge [pos=0.1] node [right] {} (a2)
	
	(s3) edge [pos=0.1] node [left] {} (a3)
	(s3) edge [pos=0.01] node [xshift=0.05cm, right] {} (a4)
	(s4) edge [pos=0.01] node [xshift=-0.1cm, left] {} (a3)
	(s4) edge [pos=0.1] node [right] {} (a4)
	
	(s5) edge [pos=0.1] node [left] {} (a5)
	(s5) edge [pos=0.01] node [xshift=0.05cm, right] {} (a6)
	(s6) edge [pos=0.01] node [xshift=-0.05cm, left] {} (a5)
	(s6) edge [pos=0.1] node [right] {} (a6)
	
	(s7) edge [pos=0.1] node [left] {} (a7)
	(s7) edge [pos=0.01] node [xshift=0.05cm, right] {} (a8)
	(s8) edge [pos=0.01] node [xshift=-0.05cm, left] {} (a7)
	(s8) edge [pos=0.1] node [right] {} (a8)
	
	(a1) edge [pos=0.1] node [left] {} (s1)
	(a1) edge [pos=0.01] node [xshift=0.05cm, right] {} (s2)
	(a2) edge [pos=0.01] node [xshift=-0.05cm, left] {} (s1)
	(a2) edge [pos=0.1] node [right] {} (s2)
	
	(a3) edge [pos=0.1] node [left] {} (s3)
	(a3) edge [pos=0.01] node [xshift=0.05cm, right] {} (s4)
	(a4) edge [pos=0.01] node [xshift=-0.05cm, left] {} (s3)
	(a4) edge [pos=0.1] node [right] {} (s4)
		
	(c1) edge [pos=0.01] node [xshift=-0.05cm, left] {} (a1)	
	(c1) edge [pos=0.0] node [xshift=0.05cm, right] {} (a3)
	(c1) edge [pos=0.01] node [xshift=-0.05cm, left] {} (a5)
	(c1) edge [pos=0.01] node [xshift=-0.05cm, left] {} (a7)
	
	(c2) edge [pos=0.0] node [xshift=-0.05cm, left] {} (a1)
	(c2) edge [pos=0.0] node [xshift=0.05cm, right] {} (a3)
	(c2) edge [pos=0.0] node [xshift=-0.05cm, left] {} (a5)
	(c2) edge [pos=0.0] node [xshift=0.05cm, right] {} (a7)
	
	(c3) edge [pos=0.0] node [xshift=-0.05cm, left] {} (a2)
	(c3) edge [pos=0.0] node [xshift=0.05cm, right] {} (a4)
	(c3) edge [pos=0.0] node [xshift=-0.05cm, left] {} (a6)
	(c3) edge [pos=0.0] node [xshift=0.05cm, right] {} (a8)
	
	(c4) edge [pos=0.0] node [xshift=-0.05cm, left] {} (a2)
	(c4) edge [pos=0.0] node [xshift=0.05cm, right] {} (a4)
	(c4) edge [pos=0.0] node [xshift=-0.05cm, left] {} (a6)
	(c4) edge [pos=0.0] node [xshift=0.05cm, right] {} (a8);
	\end{tikzpicture}

  	\caption{A FatTree Topology}
	\label{fig:fattree}
\end{figure}

We performed experiments on the FatTree~\cite{fattree} topologies to evaluate the performance of our implementation. FatTrees are hierarchical topologies which are very commonly used in data centers. Figure~\ref{fig:fattree} illustrates a FatTree with $3$ levels: core, aggregation and top-of-rack (ToR). The switches at each level contain a number of redundant links to the switches at the next upper level. The groups of ToR switches and their corresponding aggregation switches are called pods. For our experiments, we generated $6$ FatTrees that grow in size and achieve a maximum size of $1344$ switches. For these networks we computed a shortest path forwarding policy between all pairs of ToR switches. The number of switches in the ToR layer is set to $k^3/4$ where $k$ is the number of pods in the network.
We check certain dynamic properties on these networks and assess the time performance of our tool.  
We consider a scenario which involves two ToR switches $T_a$ and $T_b$ that reside in different pods. Initially, all the packets from $T_a$ to $T_b$ traverse through a firewall  $A_x$ in the aggregation layer which filters SSH packets. The controller then decides to shift the firewall from $A_x$ to another switch $A_{x'}$ in the aggregation layer. For this purpose, the controller updates the corresponding aggregation and core layer switches which results in a total of $4$ updates. The properties that we check in this scenario are as follows: (i) At any point while the controller is performing the updates, non-SSH packets from $T_a$ can always reach to $T_b$. (ii) At any point while the controller is performing the updates, SSH packets from $T_a$ can never reach to $T_b$.
(iii) After all the updates are performed, $A_{x'}$ is a waypoint between $T_a$ and $T_b$. We conducted the  experiments on a computer running Ubuntu 20.04 LTS with 16 core 2.4GHz Intel i9-9980HK processor and 64 GB RAM. The results of these experiments are displayed in Figure~\ref{fig:fattree_results}. We report the preprocessing time, the time taken for checking properties (i), (ii), and (iii) individually (referred to as Reachability-I, Reachability-II, and Waypointing, respectively), and also time taken to check all the properties in parallel (referred to as All Properties). The reported times are the average of 10 runs. The results indicate that preprocessing step is a 
non-negligible factor that contributes to overall time. However, preprocessing is independent of the property that is being checked and this procedure only needs to be done once for a given network. After the preprocessing step, the individual properties can be checked in less than $2$ seconds for networks with less than $100$ switches. For larger networks with sizes up to $931$ and $1344$ switches, the individual properties can be checked in a maximum of $5$ minutes and $11$ minutes, respectively. It can be observed that checking for the property (iii) can take more than twice as much time as checking for the properties (i) and (ii). In the experiments where we check all the properties in parallel, we allocated one thread for each property. In this setting, checking all the properties in parallel introduced $24\%$ overhead on average. An important observation in the results of these experiments is that after the preprocessing is performed, on average $87\%$ of the running times are spent in the NetKAT decision procedure and this step becomes the bottleneck in analyzing larger networks.

In order to be able to compare our technique with another verification method, we also aimed to perform an analysis based on explicit state model checking. For this purpose, we devised an operational semantics for NetKAT and implemented it in Maude along with the operational semantics of \DNK{}. However, this method immediately failed at scaling even for small networks, hence, we did not perform further analysis by using this method.

\begin{figure}[h]
    \centering
    \includegraphics[scale=0.088]{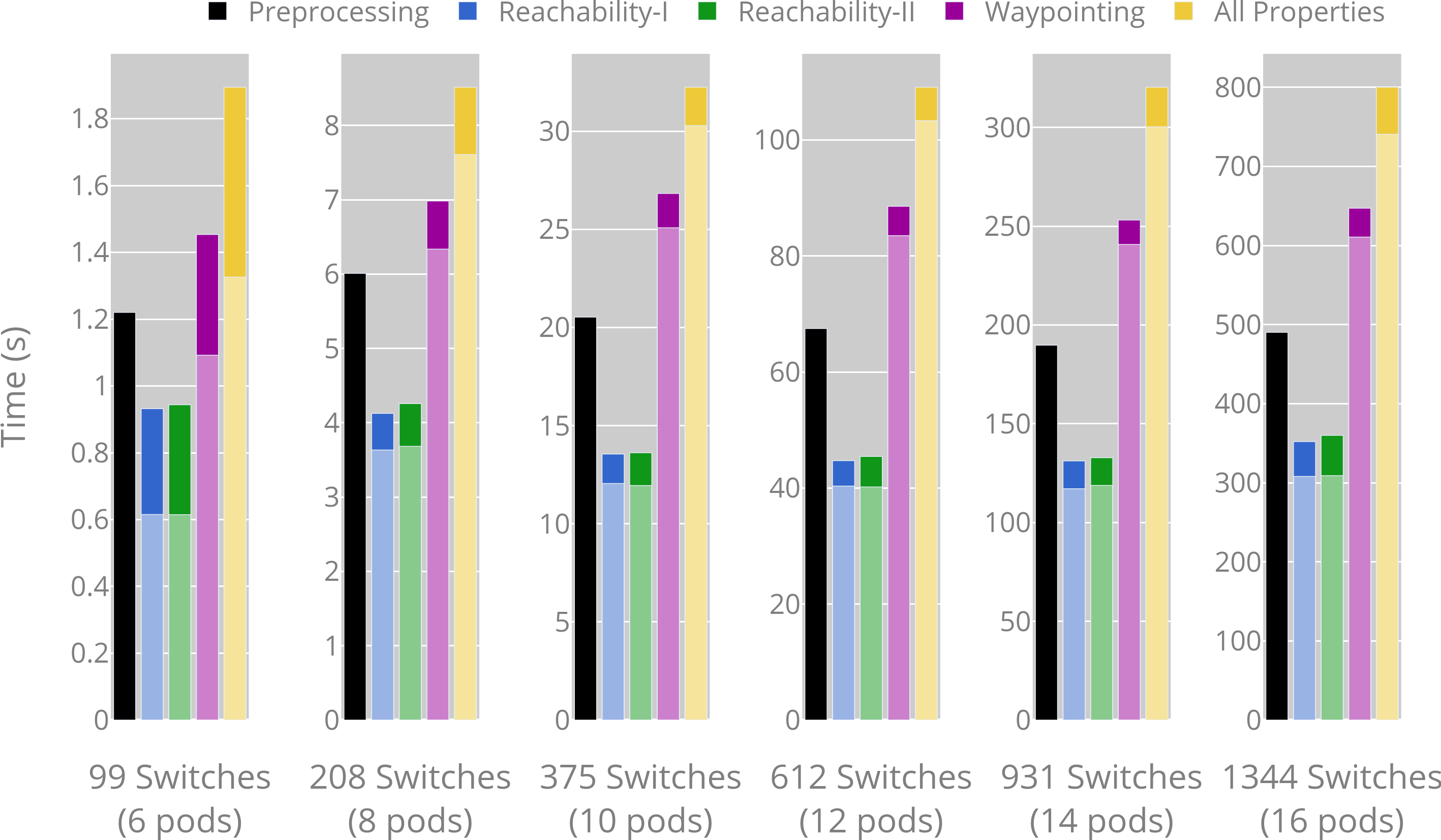}
    \caption{Results of FatTree experiments. Light colored areas indicate the time spent in the NetKAT decision procedure and solid colored areas indicate the time spent in our equational reasoning framework.}
    \label{fig:fattree_results}
\end{figure}

\section{Conclusions}\label{sec:conclusions}

We developed a language, called {\DNK}  for modelling and reasoning about dynamic reconfigurations in Software Defined Networks. Our language builds upon the concepts, syntax, and semantics of NetKAT and hence, provides a modular extension and makes it possible to reuse the theory and tools of NetKAT. 
We define a formal semantics for our language and provide a sound and ground-complete axiomatization. We exploit our axiomatization to analyse reachability properties of dynamic networks and show that our approach is indeed scalable to  networks with hundreds of switches.

Our language builds upon the assumption that control plane updates interleave with data plane packet processing in such a way that each data plane packet sees one set of flow tables throughout their flight in the network. This assumption is inspired by the  framework put forward by Reitblatt et al.\  \cite{DBLP:conf/sigcomm/ReitblattFRSW12} and is motivated by the requirement to design a modular extension on top of NetKAT. However, we have experimented with a much smaller-stepped semantics in which the control plane updates can have a finer interleaving with in-flight packet moves. This alternative language breaks the hierarchy with NetKAT and a naive treatment of this alternative semantics will lead to much larger state-spaces. We would like to investigate this small-step semantics and efficient analysis techniques for it further.

\section*{Acknowledgments.}
The authors would like to thank Alexandra Silva and Tobias Kapp\'e for their useful insight into the NetKAT framework.

\begin{ICALP}
\newpage
\end{ICALP}

\begin{ARXIV}
\newpage
\appendix
\begin{ARXIV}

\subsection{Soundness Proofs}\label{sec::soundness}

\begin{itemize}
	\item{
	Axiom under consideration:
	\begin{equation}
	\drop{} \Seq p \equiv \bot \quad (A0)
	\end{equation}
	for $p \in {\DNK}$. According to the semantic rules of {\DNK}, the derivations of the term $\drop{} \Seq p$ are as follows:
	\begin{enumerate}[leftmargin=.5in,label=(\alph*)]
		\item{\label{enum::a0-a}
			\[
			\begin{array}{r@{}l@{}c@{}r@{}l@{}c@{}r@{}l@{}}
			\textnormal{For all } \sigma' \in \llbracket \drop{} \rrbracket(\sigma \!\!::\!\! \langle \rangle): ~ \dedr{\cpolseqsucc{}} &\sosrule{}{ ( \drop{} \Seq p, \sigma :: H, H')  \trans{(\sigma, \sigma')} ( p, H , \sigma' :: H')}\\
			
			\end{array}
			\]
		}
	\end{enumerate}
	
	However, observe that $\llbracket \drop{} \rrbracket(\sigma \!\!::\!\! \langle \rangle)$ is equal to empty set:
	\begin{alignat}{2}
	 \llbracket \drop{} \rrbracket(\sigma \!\!::\!\! \langle \rangle) = \{ \}  \quad & (\textnormal{Definition of } \drop{})
	\end{alignat}
	Hence, the term $\drop{} \Seq p$ does not afford any transition. Similarly, observe that according to the semantic rules of {\DNK}, the term $\bot$ does not afford any transition. Hence, the following trivially holds:
	\begin{equation}
	(\drop{} \Seq p) \sim \bot
	\end{equation}
	}

	\item{
	Axiom under consideration:
	\begin{equation}
	p \oplus q \equiv q \oplus p \quad (A2)
	\end{equation}
	
	for $p, q \in {\DNK}$. According to the semantic rules of {\DNK}, the following are the possible transitions that can initially occur in the terms $p \oplus q$ and $q \oplus p$:
	\begin{gather*}
	\begin{cases}
	(1)\;(p, H_0, H'_0)  \trans{\gamma} (p', H_1, H'_1)\\
	(2)\;(q, H_0, H'_0)  \trans{\gamma} (q', H_1, H'_1)
	\end{cases}
	\end{gather*}
	$\gamma ::= (\sigma, \sigma') \mid x!z \mid x?z \mid {\rec{x,z}}$
	
	\begin{caseof}
		\case{1}{$(p, H_0, H'_0)  \trans{\gamma} (p', H_1, H'_1)$}{
			
			The derivations of $p \oplus q$ are as follows:
			\begin{enumerate}[leftmargin=.5in,label=(\alph*)]
				\item{\label{enum::a2-a}
					\[
					\begin{array}{r@{}l@{}c@{}r@{}l@{}c@{}r@{}l@{}}
					\dedr{\cpoloplusl}\sosrule{(p, H_0, H'_0)  \trans{\gamma} (p', H_1, H'_1)  }{ (p \oplus q, H_0, H'_0)  \trans{\gamma} ( p', H_1, H'_1)}
					\end{array} 
					\] 
				}
			\end{enumerate}
		The derivations of $q \oplus p$ are as follows:
		\begin{enumerate}[leftmargin=.5in,label=(\alph*)]
			\setcounter{enumi}{1}
			\item{\label{enum::a2-b}
				\[
				\begin{array}{r@{}l@{}c@{}r@{}l@{}c@{}r@{}l@{}}
				\dedr{\cpoloplusr}
				\sosrule{(p, H_0, H'_0)  \trans{\gamma} (p', H_1, H'_1)  }{ (q \oplus p, H_0, H'_0)  \trans{\gamma} (p', H_1, H'_1)}
				\end{array} 
				\] 
			}
		\end{enumerate}
	
		As demonstrated in \ref{enum::a2-a} and \ref{enum::a2-b}, if $(p, H_0, H'_0)  \trans{\gamma} (p', H_1, H'_1)$ holds then both of the terms $p \oplus q$ and $q \oplus p$ converge to the same expression with the $\gamma$ transition:
		\begin{equation}
		\label{eq::a2-r1}
		\begin{alignedat}{2}
		(p \oplus q, H_0, H'_0) \trans{\gamma}& (p', H_1, H'_1)\\
		(q \oplus p, H_0, H'_0) \trans{\gamma}& (p', H_1, H'_1)
		\end{alignedat}
		\end{equation}
		}
		
		\case{2}{$(q, H_0, H'_0)  \trans{\gamma} (q', H_1, H'_1)$}{
			
			The derivations of $p \oplus q$ are as follows:
			\begin{enumerate}[leftmargin=.5in,label=(\alph*)]
				\setcounter{enumi}{2}
				\item{\label{enum::a2-c}
					\[
					\begin{array}{r@{}l@{}c@{}r@{}l@{}c@{}r@{}l@{}}
					\dedr{\cpoloplusr}\sosrule{(q, H_0, H'_0)  \trans{\gamma} (q', H_1, H'_1)  }{ (p \oplus q, H_0, H'_0)  \trans{\gamma} (q', H_1, H'_1)}
					\end{array} 
					\] 
				}
			\end{enumerate}
			The derivations of $q \oplus p$ are as follows:
			\begin{enumerate}[leftmargin=.5in,label=(\alph*)]
				\setcounter{enumi}{3}
				\item{\label{enum::a2-d}
					\[
					\begin{array}{r@{}l@{}c@{}r@{}l@{}c@{}r@{}l@{}}
					\dedr{\cpoloplusl}
					\sosrule{(q, H_0, H'_0)  \trans{\gamma} (q', H_1, H'_1)  }{ (q \oplus p, H_0, H'_0)  \trans{\gamma} (q', H_1, H'_1)}
					\end{array} 
					\] 
				}
			\end{enumerate}
		As demonstrated in \ref{enum::a2-c} and \ref{enum::a2-d}, if $(q, H_0, H'_0)  \trans{\gamma} (q', H_1, H'_1)$ holds then both of the terms $p \oplus q$ and $q \oplus p$ converge to the same expression with the $\gamma$ transition:
		\begin{equation}
		\label{eq::a2-r2}
		\begin{alignedat}{2}
		(p \oplus q, H_0, H'_0) \trans{\gamma}& (q', H_1, H'_1)\\
		(q \oplus p, H_0, H'_0) \trans{\gamma}& (q', H_1, H'_1)
		\end{alignedat}
		\end{equation}
		}
	\end{caseof}
	
	Therefore, by (\ref{eq::a2-r1}) and (\ref{eq::a2-r2}) it is straightforward to conclude that the following holds:
	\begin{equation}
	(p \oplus q) \sim (q \oplus p) 
	\end{equation}
	}

	\item{
	Axiom under consideration:
	\begin{equation}
	(p \oplus q) \oplus r \equiv p \oplus (q \oplus r) \quad  (A3)
	\end{equation}
	for $p, q, r \in {\DNK}$. According to the semantic rules of {\DNK}, the following are the possible transitions that can initially occur in the terms $(p \oplus q) \oplus r$ and $p \oplus (q \oplus r)$:
	\begin{gather*}
	\begin{cases}
	(1)\;(p, H_0, H'_0)  \trans{\gamma} (p', H_1, H'_1)\\
	(2)\;(q, H_0, H'_0)  \trans{\gamma} (q', H_1, H'_1)\\
	(3)\;(r, H_0, H'_0)  \trans{\gamma} (r', H_1, H'_1)
	\end{cases}
	\end{gather*}
	$\gamma ::= (\sigma, \sigma') \mid x!z \mid x?z \mid {\rec{x,z}}$
		\begin{caseof}
		\case{1}{$(p, H_0, H'_0)  \trans{\gamma} (p', H_1, H'_1)$}{
			
		The derivations of $(p \oplus q) \oplus r$ are as follows:
			\begin{enumerate}[leftmargin=.5in,label=(\alph*)]
				\item{\label{enum::a3-a}
					\[
					\begin{array}{r@{}l@{}c@{}r@{}l@{}c@{}r@{}l@{}}
					\dedr{\cpoloplusl}&\sosrule{(p, H_0, H'_0)  \trans{\gamma} (p', H_1, H'_1)  }{ (p \oplus q, H_0, H'_0)  \trans{\gamma} ( p', H_1, H'_1)}\\
					\dedr{\cpoloplusl}&\sosrule{}{ ((p \oplus q) \oplus r, H_0, H'_0)  \trans{\gamma} (p', H_1, H'_1)}
					\end{array} 
					\] 
				}
			\end{enumerate}
			The derivations of $p \oplus (q \oplus r)$ are as follows:
			\begin{enumerate}[leftmargin=.5in,label=(\alph*)]
				\setcounter{enumi}{1}
				\item{\label{enum::a3-b}
					\[
					\begin{array}{r@{}l@{}c@{}r@{}l@{}c@{}r@{}l@{}}
					\dedr{\cpoloplusl}
					\sosrule{(p, H_0, H'_0)  \trans{\gamma} (p', H_1, H'_1)  }{ (p \oplus (q \oplus r), H_0, H'_0)  \trans{\gamma} (p', H_1, H'_1)}
					\end{array} 
					\] 
				}
			\end{enumerate}
			
			As demonstrated in \ref{enum::a3-a} and \ref{enum::a3-b}, if $(p, H_0, H'_0)  \trans{\gamma} (p', H_1, H'_1)$ holds then both of the terms $(p \oplus q) \oplus r$ and $p \oplus (q \oplus r)$ converge to the same expression with the $\gamma$ transition:
			\begin{equation}
			\label{eq::a3-r1}
			\begin{alignedat}{2}
			((p \oplus q) \oplus r, H_0, H'_0) \trans{\gamma}& (p', H_1, H'_1)\\
			(p \oplus (q \oplus r), H_0, H'_0) \trans{\gamma}& (p', H_1, H'_1)
			\end{alignedat}
			\end{equation}
		}
	
		\case{2}{$(q, H_0, H'_0)  \trans{\gamma} (q', H_1, H'_1)$}{
		
		The derivations of $(p \oplus q) \oplus r$ are as follows:
		\begin{enumerate}[leftmargin=.5in,label=(\alph*)]
			\setcounter{enumi}{2}
			\item{\label{enum::a3-c}
				\[
				\begin{array}{r@{}l@{}c@{}r@{}l@{}c@{}r@{}l@{}}
				\dedr{\cpoloplusr}&\sosrule{(q, H_0, H'_0)  \trans{\gamma} (q', H_1, H'_1)  }{ (p \oplus q, H_0, H'_0)  \trans{\gamma} ( q', H_1, H'_1)}\\
				\dedr{\cpoloplusl}&\sosrule{}{ ((p \oplus q) \oplus r, H_0, H'_0)  \trans{\gamma} (q', H_1, H'_1)}
				\end{array} 
				\] 
			}
		\end{enumerate}
		The derivations of $p \oplus (q \oplus r)$ are as follows:
		\begin{enumerate}[leftmargin=.5in,label=(\alph*)]
			\setcounter{enumi}{3}
			\item{\label{enum::a3-d}
				\[
				\begin{array}{r@{}l@{}c@{}r@{}l@{}c@{}r@{}l@{}}
				\dedr{\cpoloplusl}&
				\sosrule{(q, H_0, H'_0)  \trans{\gamma} (q', H_1, H'_1)  }{ (q \oplus r, H_0, H'_0)  \trans{\gamma} (q', H_1, H'_1)}\\
				\dedr{\cpoloplusr}&
				\sosrule{}{ (p \oplus (q \oplus r), H_0, H'_0)  \trans{\gamma} (q', H_1, H'_1)}
				\end{array} 
				\] 
			}
		\end{enumerate}
		
		As demonstrated in \ref{enum::a3-c} and \ref{enum::a3-d}, if $(q, H_0, H'_0)  \trans{\gamma} (q', H_1, H'_1)$ holds then both of the terms $(p \oplus q) \oplus r$ and $p \oplus (q \oplus r)$ converge to the same expression with the $\gamma$ transition:
		\begin{equation}
		\label{eq::a3-r2}
		\begin{alignedat}{2}
		((p \oplus q) \oplus r, H_0, H'_0) \trans{\gamma}& (q', H_1, H'_1)\\
		(p \oplus (q \oplus r), H_0, H'_0) \trans{\gamma}& (q', H_1, H'_1)
		\end{alignedat}
		\end{equation}
	}

		\case{3}{$(r, H_0, H'_0)  \trans{\gamma} (r', H_1, H'_1)$}{
	
	The derivations of $(p \oplus q) \oplus r$ are as follows:
	\begin{enumerate}[leftmargin=.5in,label=(\alph*)]
		\setcounter{enumi}{4}
		\item{\label{enum::a3-e}
			\[
			\begin{array}{r@{}l@{}c@{}r@{}l@{}c@{}r@{}l@{}}
			\dedr{\cpoloplusr}&\sosrule{(r, H_0, H'_0)  \trans{\gamma} (r', H_1, H'_1)  }{ ((p \oplus q) \oplus r, H_0, H'_0)  \trans{\gamma} ( r', H_1, H'_1)}
			\end{array} 
			\] 
		}
	\end{enumerate}
	The derivations of $p \oplus (q \oplus r)$ are as follows:
	\begin{enumerate}[leftmargin=.5in,label=(\alph*)]
		\setcounter{enumi}{5}
		\item{\label{enum::a3-f}
			\[
			\begin{array}{r@{}l@{}c@{}r@{}l@{}c@{}r@{}l@{}}
			\dedr{\cpoloplusr}&
			\sosrule{(r, H_0, H'_0)  \trans{\gamma} (r', H_1, H'_1)  }{ (q \oplus r, H_0, H'_0)  \trans{\gamma} (r', H_1, H'_1)}\\
			\dedr{\cpoloplusr}&
			\sosrule{}{ (p \oplus (q \oplus r), H_0, H'_0)  \trans{\gamma} (r', H_1, H'_1)}
			\end{array} 
			\] 
		}
	\end{enumerate}
	
	As demonstrated in \ref{enum::a3-e} and \ref{enum::a3-f}, if $(r, H_0, H'_0)  \trans{\gamma} (r', H_1, H'_1)$ holds then both of the terms $(p \oplus q) \oplus r$ and $p \oplus (q \oplus r)$ converge to the same expression with the $\gamma$ transition:
	\begin{equation}
	\label{eq::a3-r3}
	\begin{alignedat}{2}
	((p \oplus q) \oplus r, H_0, H'_0) \trans{\gamma}& (r', H_1, H'_1)\\
	(p \oplus (q \oplus r), H_0, H'_0) \trans{\gamma}& (r', H_1, H'_1)
	\end{alignedat}
	\end{equation}
}
	Therefore, by (\ref{eq::a3-r1}), (\ref{eq::a3-r2}) and (\ref{eq::a3-r3}) it is straightforward to conclude that the following holds:
	\begin{equation}
	((p \oplus q) \oplus r) \sim (p \oplus (q \oplus r)) 
	\end{equation}
	\end{caseof}
	}

	\item{
	Axiom under consideration:
	\begin{equation}
	p \oplus p \equiv p \quad (A4) 
	\end{equation}
	for $p \in {\DNK}$. According to the semantic rules of {\DNK}, the following are the possible transitions that can initially occur in the terms $p \oplus p$ and $p$:
	\begin{gather*}
	\begin{cases}
	(1)\;(p, H_0, H'_0)  \trans{\gamma} (p', H_1, H'_1)
	\end{cases}
	\end{gather*}
	$\gamma ::= (\sigma, \sigma') \mid x!z \mid x?z \mid {\rec{x,z}}$
	
	\begin{caseof}
		\case{1}{$(p, H_0, H'_0)  \trans{\gamma} (p', H_1, H'_1)$}{
			
			The derivations of $p \oplus p$ are as follows:
			\begin{enumerate}[leftmargin=.5in,label=(\alph*)]
				\item{\label{enum::a4-a}
					\[
					\begin{array}{r@{}l@{}c@{}r@{}l@{}c@{}r@{}l@{}}
					\dedr{\cpoloplusl}\sosrule{(p, H_0, H'_0)  \trans{\gamma} (p', H_1, H'_1)  }{ (p \oplus p, H_0, H'_0)  \trans{\gamma} ( p', H_1, H'_1)}
					\end{array} 
					\] 
				}
			\item{\label{enum::a4-b}
				\[
				\begin{array}{r@{}l@{}c@{}r@{}l@{}c@{}r@{}l@{}}
				\dedr{\cpoloplusr}\sosrule{(p, H_0, H'_0)  \trans{\gamma} (p', H_1, H'_1)  }{ (p \oplus p, H_0, H'_0)  \trans{\gamma} ( p', H_1, H'_1)}
				\end{array} 
				\] 
			}
			\end{enumerate}
			
			As demonstrated in \ref{enum::a4-a} and \ref{enum::a4-b}, if $(p, H_0, H'_0)  \trans{\gamma} (p', H_1, H'_1)$ holds then it is also the case that the term $p \oplus p$ evolves into the same expression with the $\gamma$ transition:
			\begin{alignat}{2}
			(p \oplus p, H_0, H'_0) \trans{\gamma}& (p', H_1, H'_1)
			\end{alignat}
			Hence, it is straightforward to conclude that the following holds:
			\begin{equation}
			(p \oplus p) \sim p
			\end{equation}
		}
	\end{caseof}
	}

		\item{
		Axiom under consideration:
		\begin{equation}
		p \oplus \bot \equiv p \quad (A5)
		\end{equation}
		for $p \in {\DNK}$. According to the semantic rules of {\DNK}, the following are the possible transitions that can initially occur in the terms $p \oplus \bot$ and $p$:
		\begin{gather*}
		\begin{cases}
		(1)\;(p, H_0, H'_0)  \trans{\gamma} (p', H_1, H'_1)
		\end{cases}
		\end{gather*}
		$\gamma ::= (\sigma, \sigma') \mid x!z \mid x?z \mid {\rec{x,z}}$
		
		\begin{caseof}
			\case{1}{$(p, H_0, H'_0)  \trans{\gamma} (p', H_1, H'_1)$}{
				
				The derivations of $p \oplus \bot$ are as follows:
				\begin{enumerate}[leftmargin=.5in,label=(\alph*)]
					\item{\label{enum::a5-a}
						\[
						\begin{array}{r@{}l@{}c@{}r@{}l@{}c@{}r@{}l@{}}
						\dedr{\cpoloplusl}\sosrule{(p, H_0, H'_0)  \trans{\gamma} (p', H_1, H'_1)  }{ (p \oplus \bot, H_0, H'_0)  \trans{\gamma} ( p', H_1, H'_1)}
						\end{array} 
						\] 
					}
				\end{enumerate}
				
				As demonstrated in \ref{enum::a5-a}, if $(p, H_0, H'_0)  \trans{\gamma} (p', H_1, H'_1)$ holds then it is also the case that the term $p \oplus \bot$ evolves into the same expression with the $\gamma$ transition:
				\begin{alignat}{2}
				(p \oplus \bot, H_0, H'_0) \trans{\gamma}& (p', H_1, H'_1)
				\end{alignat}
				Hence, it is straightforward to conclude that the following holds:
				\begin{equation}
				(p \oplus \bot) \sim p
				\end{equation}
		}
	\end{caseof}
	}

    \item{
    Axiom under consideration:
		\begin{equation}
		p \Par q \equiv q \Par p \quad (A6)
		\end{equation}
		for $p, q \in {\DNK}$. The soundness proof of the axiom $(A6)$ follows by induction on the size of the syntactic tree associated to $p \Par q$. Without loss of generality, assume $p$ and $q$ are in n.f. The size of $p\Par q$ is then defined as follows:	\begin{alignat}{2}
		size(\bot) &= 1\\
		size(\alpha \cdot \pi \Seq t) &= 2 + size(t)\\
		size(x?z \Seq t) &= 2 + size(t)\\
		size(x!z \Seq t) &= 2 + size(t)\\
		size(\recp{x, z} \Seq t) &= 2 + size(t)\\
		size(p \oplus q) &= 1 + size(p) + size(q)\\
		size(p \Par q) &= 1 + size(p) + size(q)
		\end{alignat}
		\emph{Base case.}
		\begin{itemize}
		    \item $size(p \Par q) = 3$. It follows that $p \triangleq \bot$ and $q \triangleq \bot$. Therefore, the soundness of $(A6)$ trivially holds.
		\end{itemize}
			\emph{Induction step.}
		Assume $(A6)$ is sound for all $p$, $q$ such that $size(p \Par q) \leq M$, with $M \in \mathbb{N}$. We want to show that $(A6)$ is sound for all $p$, $q$ such that $size(p \Par q) > M$.
		\begin{enumerate}[label=(\roman*)]
		\setlength\itemsep{1em}
		    \item{$p \triangleq \bot$. Then, it is straightforward to observe that both $\bot \Par q$ and $q \Par \bot$ evolve according to the semantic rules corresponding to $q$. Hence, we can safely conclude that $(\bot \Par q) \sim (q \Par \bot)$ holds.}
		    
		    \item{\label{enum::a6-ii}$p \triangleq \alpha \cdot \pi ; p'$. Consider an arbitrary but fixed network packet $\sigma$, let $S_{\alpha\pi} \triangleq \llbracket \alpha \cdot \pi  \rrbracket(\sigma \!\!::\!\! \langle \rangle)$. The first step derivations entailed by $p$ in a context $p \Par t$ are as follows:
				\begin{enumerate}[leftmargin=.5in,label=(\alph*)]
					\item{\label{enum::a6-c1-a}
						\[
						\begin{array}{l@{}r@{}l@{}c@{}r@{}l@{}c@{}r@{}l@{}}
						\textnormal{For all } \sigma' \in S_{\alpha\pi}:\quad\quad&\dedr{\cpolseqsucc}&\sosrule{}{ (\alpha \cdot \pi \Seq p', \sigma :: H, H')  \trans{(\sigma, \sigma')} ( p', H, \sigma' :: H')}\\ &\dedr{\intl}&\sosrule{}{ ((\alpha \cdot \pi \Seq p') \Par t, \sigma :: H, H')  \trans{(\sigma, \sigma')} ( p' \Par t, H, \sigma' :: H')} 
						\end{array} 
						\] 
					}
				\end{enumerate}
			    The first step derivations entailed by $p$ in a context $t \Par p$ are as follows:
				\begin{enumerate}[leftmargin=.5in,label=(\alph*)]
					\setcounter{enumii}{1}
					\item{\label{enum::a6-c1-b}
						\[
						\begin{array}{l@{}r@{}l@{}c@{}r@{}l@{}c@{}r@{}l@{}}
						\textnormal{For all } \sigma' \in S_{\alpha\pi}:\quad\quad&\dedr{\cpolseqsucc}&\sosrule{}{ (\alpha \cdot \pi \Seq p', \sigma :: H, H')  \trans{(\sigma, \sigma')} ( p', H, \sigma' :: H')}\\ &\dedr{\intr}&\sosrule{}{ (t \Par (\alpha \cdot \pi \Seq p'), \sigma :: H, H')  \trans{(\sigma, \sigma')} ( t \Par p', H, \sigma' :: H')}
						\end{array} 
						\] 
					}
				\end{enumerate}
			 Hence, given that $q$ in n.f. always evolves into a term $t$ with simpler structure (according to the {\DNK} semantic rules), and based on the induction hypothesis, it is safe to conclude that $(p \Par q) \sim (q \Par p)$.
		    }

    \item{
		$p \triangleq \recp{x, z} \Seq p'$. The reasoning is similar to \ref{enum::a6-ii} above.
		}

		\item{\label{enum::a6-iv}$p \triangleq x?z ; p'$. The first step of asynchronous derivations entailed by $p$ in a context $p \Par t$ are as follows:
		\begin{enumerate}[leftmargin=.5in,label=(\alph*)]
			\item{\label{enum::a6-c2-a}
				\[
				\begin{array}{r@{}l@{}c@{}r@{}l@{}c@{}r@{}l@{}} \dedr{\cpolmsgrec}&\sosrule{}{ (x?z \Seq p', H, H')  \trans{x?z} ( p', H, H')}\\
				\dedr{\intl} & \sosrule{}{ ((x?z \Seq p') \Par t, H, H')  \trans{x?z} ( p' \Par t, H, H')}
				\end{array} 
				\] 
			}
		\end{enumerate}
		The first step of asynchronous derivations entailed by $p$ in a context $t \Par p$ are as follows:
		\begin{enumerate}[leftmargin=.5in,label=(\alph*)]
			\setcounter{enumii}{1}
			\item{\label{enum::a6-c2-b}
				\[
				\begin{array}{r@{}l@{}c@{}r@{}l@{}c@{}r@{}l@{}} \dedr{\cpolmsgrec}&\sosrule{}{ (x?z \Seq p', H, H')  \trans{x?z} ( p', H, H')}\\
				\dedr{\intr} & \sosrule{}{ (t \Par (x?z \Seq p'), H, H')  \trans{x?z} ( t \Par p', H, H')}
				\end{array} 
				\] 
			}
		\end{enumerate}
		Furthermore, if $q$ has a summand of shape $x!z \Seq q'$, then:\\
		The first step synchronous derivations of $p \Par q$ are as follows:
		\begin{enumerate}[leftmargin=.5in,label=(\alph*)]
	    \setcounter{enumii}{2}
    	\item{\label{enum::a6-c3-a}
		\[
			\begin{array}{r@{}l@{}c@{}r@{}l@{}c@{}r@{}l@{}}
			\dedr{\cpolmsgrec} & \sosrule{}{ (x?z ; p', H, H')  \trans{x?z} (p', H, H')} &
			\quad \quad &
			\dedr{\cpolmsgsend}&\sosrule{}{ (x!z \Seq q', H, H')  \trans{x!z} ( q', H, H')}\\
			
			\dedr{\reconfigrs} & \multicolumn{6}{@{}l@{}}{ \sosrule{\myquad[31]~}{((x?\mathit{z} \Seq p') \Par (x!z \Seq q'), H, H') \trans{\rec{x,z}} (p' \Par q', H, H')}}
			\end{array} 
		\]
		}
		
		\end{enumerate}

		The first step  synchronous derivations of $q \Par p$ are as follows:	\begin{enumerate}[leftmargin=.5in,label=(\alph*)]
	    \setcounter{enumii}{3}
    	\item{\label{enum::a6-c3-b}
		\[
		\begin{array}{r@{}l@{}c@{}r@{}l@{}c@{}r@{}l@{}}
			\dedr{\cpolmsgsend} & \sosrule{}{ (x!z \Seq q', H, H')  \trans{x!z} (q', H, H')} &
			\quad \quad &
			\dedr{\cpolmsgrec}&\sosrule{}{ (x?z ; p', H, H')  \trans{x?z} ( p', H, H')}\\
			
			\dedr{\reconfigsr} & \multicolumn{6}{@{}l@{}}{ \sosrule{\myquad[31]~}{((x!z \Seq q') \Par (x?\mathit{z} \Seq p'), H, H') \trans{\rec{x,z}} (q' \Par p', H, H')}}
			\end{array} 
		\]
		}
		\end{enumerate}	
		In connection with (iv)(a) and (iv)(b) above, note that $q$ in n.f. always evolves into a term $t$ with simpler structure (according to the {\DNK} semantic rules). This, together with the observations in (iv)(c) and (iv)(d), and based on the induction hypothesis, enable us to safely to conclude that $(p \Par q) \sim (q \Par p)$.
		}

		\item{$p \triangleq x!z \Seq p'$. The reasoning is similar to \ref{enum::a6-iv} above.}

		\item{
		$p \triangleq p_1 \oplus p_2$ where $p_1$ and $p_2$ are in n.f. Without loss of generality, assume $p_1 ::= \alpha \cdot \pi \Seq p_1' \mid  {\recp{x,z}} \Seq p_1' \mid x?z \Seq p_1'  \mid x!z \Seq p_1'$ and assume $(p_1, H_0, H'_0)  \trans{\gamma} (p_1', H_1, H'_1)$. The derivation entailed by $p_1$ in p is as follows:
				\[	\begin{array}{r@{}l@{}c@{}r@{}l@{}c@{}r@{}l@{}}
					\dedr{\cpoloplusl}&\sosrule{(p_1, H_0, H'_0)  \trans{\gamma} (p_1', H_1, H'_1)  }{ (p_1 \oplus p_2, H_0, H'_0)  \trans{\gamma} ( p_1', H_1, H'_1)}	\end{array} 
					\] 
			From this point onward, showing that the first step derivations entailed by $p_1$ in a context $p \Par t$ correspond to the first step derivations entailed by $p_1$ in a context $t \Par p$
			follows the reasoning in (ii)$-$(v), with $p_1$ ranging over terms of shape $(\alpha \cdot \pi \Seq p_1'),  ({\recp{x,z}} \Seq p_1'), (x!z \Seq p_1')$ and $(x?z \Seq p_1')$, respectively.  Hence, given that $q$ in n.f. always evolves into a term $t$ with simpler structure (according to the {\DNK} semantic rules), and based on the induction hypothesis, it is safe to conclude that $(p \Par q) \sim (q \Par p)$.    

		}

		\end{enumerate}
    }

	\item{
	Axiom under consideration:
	\begin{equation}
	p \Par \bot \equiv p \quad (A7)
	\end{equation}
	for $p \in {\DNK}$. According to the semantic rules of {\DNK}, observe that both $p \Par \bot$ and $p$ evolve according to the semantic rules corresponding to $p$. Hence, it is straightforward to conclude that the following holds:
	\begin{equation}
	(p \Par \bot) \sim p
	\end{equation}
	}
    
	\item{
	Axiom under consideration:
	\begin{equation}
	p \Par q \equiv p \llfloor q \oplus q \llfloor p \oplus p \mid q \quad (A8)
	\end{equation}
	for $p, q \in {\DNK}$. According to the semantic rules of {\DNK}, the following are the possible transitions that can initially occur in the terms $p \Par q$ and $p \llfloor q \oplus q \llfloor p \oplus p \mid q$:

		\begin{gather*}
	\begin{cases}
	(1)\;(p, H_0, H'_0)  \trans{\gamma} (p', H_1, H'_1)\\
	(2)\;(q, H_0, H'_0)  \trans{\gamma} (q', H_1, H'_1)\\
	(3)\;(p, H_0, H'_0)  \trans{x!z} (p', H_1, H'_1) \quad (q, H_0, H'_0)  \trans{x?z} (q', H_1, H'_1)\\
	(4)\;(p, H_0, H'_0)  \trans{x?z} (p', H_1, H'_1) \quad (q, H_0, H'_0)  \trans{x!z} (q', H_1, H'_1)
	\end{cases}
	\end{gather*}
	$\gamma ::= (\sigma, \sigma') \mid x!z \mid x?z \mid {\rec{x,z}}$
	
	\begin{caseof}
		\case{1}{$(p, H_0, H'_0)  \trans{\gamma} (p', H_1, H'_1)$}{
			
			The derivations of $p \Par q$ are as follows:
			\begin{enumerate}[leftmargin=.5in,label=(\alph*)]
				\item{\label{enum::a8-a}
					\[
					\begin{array}{r@{}l@{}c@{}r@{}l@{}c@{}r@{}l@{}}
					\dedr{\intl}\sosrule{(p, H_0, H'_0)  \trans{\gamma} (p', H_1, H'_1)  }{ (p \Par q, H_0, H'_0)  \trans{\gamma} ( p' \Par q, H_1, H'_1)}
					\end{array} 
					\] 
				}
			\end{enumerate}
			The derivations of $p \llfloor q \oplus q \llfloor p \oplus p \mid q$ are as follows:
			\begin{enumerate}[leftmargin=.5in,label=(\alph*)]
				\setcounter{enumi}{1}
				\item{\label{enum::a8-b}
					\[
					\begin{array}{r@{}l@{}c@{}r@{}l@{}c@{}r@{}l@{}}
					\dedr{\llfloor}&
					\sosrule{(p, H_0, H'_0)  \trans{\gamma} (p', H_1, H'_1)  }{ (p \llfloor q, H_0, H'_0)  \trans{\gamma} (p' \Par q, H_1, H'_1)}\\
					\dedr{\cpoloplusl}&
					\sosrule{}{ (p \llfloor q \oplus q \llfloor p \oplus p \mid q, H_0, H'_0)  \trans{\gamma} (p' \Par q, H_1, H'_1)}
					\end{array} 
					\] 
				}
			\end{enumerate}
			
			As demonstrated in \ref{enum::a8-a} and \ref{enum::a8-b}, if $(p, H_0, H'_0)  \trans{\gamma} (p', H_1, H'_1)$ holds then both of the terms $p \Par q$ and $p \llfloor q \oplus q \llfloor p \oplus p \mid q$ converge to the same expression with the $\gamma$ transition:
			\begin{equation}
			\label{eq::a8-r1}
			\begin{alignedat}{2}
			(p \Par q, H_0, H'_0) \trans{\gamma}& (p' \Par q, H_1, H'_1)\\
			(p \llfloor q \oplus q \llfloor p \oplus p \mid q, H_0, H'_0) \trans{\gamma}& (p' \Par q, H_1, H'_1)
			\end{alignedat}
			\end{equation}
		}

		\case{2}{$(q, H_0, H'_0)  \trans{\gamma} (p', H_1, H'_1)$}{
		
		The derivations of $p \Par q$ are as follows:
		\begin{enumerate}[leftmargin=.5in,label=(\alph*)]
			\setcounter{enumi}{2}
			\item{\label{enum::a8-c}
				\[
				\begin{array}{r@{}l@{}c@{}r@{}l@{}c@{}r@{}l@{}}
				\dedr{\intr}\sosrule{(q, H_0, H'_0)  \trans{\gamma} (q', H_1, H'_1)  }{ (p \Par q, H_0, H'_0)  \trans{\gamma} ( p \Par q', H_1, H'_1)}
				\end{array} 
				\] 
			}
		\end{enumerate}
		The derivations of $p \llfloor q \oplus q \llfloor p \oplus p \mid q$ are as follows:
		\begin{enumerate}[leftmargin=.5in,label=(\alph*)]
			\setcounter{enumi}{3}
			\item{\label{enum::a8-d}
				\[
				\begin{array}{r@{}l@{}c@{}r@{}l@{}c@{}r@{}l@{}}
				\dedr{\llfloor}&
				\sosrule{(q, H_0, H'_0)  \trans{\gamma} (q', H_1, H'_1)  }{ (q \llfloor p, H_0, H'_0)  \trans{\gamma} (q' \Par p, H_1, H'_1)}\\
				\dedr{\cpoloplusr}&
				\sosrule{}{ (p \llfloor q \oplus q \llfloor p, H_0, H'_0)  \trans{\gamma} (q' \Par p, H_1, H'_1)}\\
				\dedr{\cpoloplusl}&
				\sosrule{}{ (p \llfloor q \oplus q \llfloor p \oplus p \mid q, H_0, H'_0)  \trans{\gamma} (q' \Par p, H_1, H'_1)}
				\end{array} 
				\] 
			}
		\end{enumerate}
		
		As demonstrated in \ref{enum::a8-c} and \ref{enum::a8-d}, if $(p, H_0, H'_0)  \trans{\gamma} (p', H_1, H'_1)$ holds then both of the terms $p \Par q$ and $p \llfloor q \oplus q \llfloor p \oplus p \mid q$ are able to perform the $\gamma$ transition:
		\begin{equation}
		\label{eq::a8-r2}
		\begin{alignedat}{2}
		(p \Par q, H_0, H'_0) \trans{\gamma}& (p \Par q', H_1, H'_1)\\
		(p \llfloor q \oplus q \llfloor p \oplus p \mid q, H_0, H'_0) \trans{\gamma}& (q' \Par p, H_1, H'_1)
		\end{alignedat}
		\end{equation}
		Observe that the terms evolve into different expressions, however, according to the axiom $A6$ the ``$\Par$'' operator is commutative. Hence, the following holds:
		\begin{equation}\label{eq::a8-r2-2}
		(p \Par q') \sim (q' \Par p)
		\end{equation}
	}

		\case{3}{$(p, H_0, H'_0)  \trans{x!z} (p', H_1, H'_1) \quad (q, H_0, H'_0)  \trans{x?z} (q', H_1, H'_1)$}{
			
			The derivations of $p \Par q$ are as follows:
			\begin{enumerate}[leftmargin=.5in,label=(\alph*)]
				\setcounter{enumi}{4}
				\item{\label{enum::a8-e}
					\[
					\begin{array}{r@{}l@{}c@{}r@{}l@{}c@{}r@{}l@{}}
					\dedr{\reconfigsr}\sosrule{(p, H_0, H'_0)  \trans{x!z} (p', H_1, H'_1) \quad (q, H_0, H'_0)  \trans{x?z} (q', H_1, H'_1)  }{ (p \Par q, H_0, H'_0)  \trans{\rec{x,z}} ( p' \Par q', H_1, H'_1)}
					\end{array} 
					\] 
				}
			\end{enumerate}
			The derivations of $p \llfloor q \oplus q \llfloor p \oplus p \mid q$ are as follows:
			\begin{enumerate}[leftmargin=.5in,label=(\alph*)]
				\setcounter{enumi}{5}
				\item{\label{enum::a8-f}
					\[
					\begin{array}{r@{}l@{}c@{}r@{}l@{}c@{}r@{}l@{}}
					\dedr{\reconfigsr} &
					\sosrule{(p, H_0, H'_0)  \trans{x!z} (p', H_1, H'_1) \quad (q, H_0, H'_0)  \trans{x?z} (q', H_1, H'_1)  }{ (p \mid q, H_0, H'_0)  \trans{\rec{x,z}} (p' \Par q', H_1, H'_1)} \\
					\dedr{\cpoloplusr}&
					\sosrule{\myquad[25]~~}{ (p \llfloor q \oplus q \llfloor p \oplus p \mid q, H_0, H'_0)  \trans{\rec{x,z}} (p' \Par q', H_1, H'_1)}
					\end{array} 
					\] 
				}
			\end{enumerate}
			
			As demonstrated in \ref{enum::a8-e} and \ref{enum::a8-f}, if $(p, H_0, H'_0)  \trans{x!z} (p', H_1, H'_1)$ and $(q, H_0, H'_0)  \trans{x?z} (q',\allowbreak H_1, H'_1)$ hold then both of the terms $p \Par q$ and $p \llfloor q \oplus q \llfloor \oplus p \mid q$ converge to the same expression with the $\rec{x,z}$ transition:
			\begin{equation}
			\label{eq::a8-r3}
			\begin{alignedat}{2}
			(p \Par q, H_0, H'_0) \trans{\rec{x,z}}& (p' \Par q', H_1, H'_1)\\
			(p \llfloor q \oplus q \llfloor p \oplus p \mid q, H_0, H'_0) \trans{\rec{x,z}}& (p' \Par q', H_1, H'_1)
			\end{alignedat}
			\end{equation}
			
			}
		
		\case{4}{$(p, H_0, H'_0)  \trans{x?z} (p', H_1, H'_1) \quad (q, H_0, H'_0)  \trans{x!z} (q', H_1, H'_1)$}{
			
			The derivations of $p \Par q$ are as follows:
			\begin{enumerate}[leftmargin=.5in,label=(\alph*)]
				\setcounter{enumi}{6}
				\item{\label{enum::a8-g}
					\[
					\begin{array}{r@{}l@{}c@{}r@{}l@{}c@{}r@{}l@{}}
					\dedr{\reconfigrs}\sosrule{(p, H_0, H'_0)  \trans{x?z} (p', H_1, H'_1) \quad (q, H_0, H'_0)  \trans{x!z} (q', H_1, H'_1)  }{ (p \Par q, H_0, H'_0)  \trans{\rec{x,z}} ( p' \Par q', H_1, H'_1)}
					\end{array} 
					\] 
				}
			\end{enumerate}
			The derivations of $p \llfloor q \oplus q \llfloor p \oplus p \mid q$ are as follows:
			\begin{enumerate}[leftmargin=.5in,label=(\alph*)]
				\setcounter{enumi}{7}
				\item{\label{enum::a8-h}
					\[
					\begin{array}{r@{}l@{}c@{}r@{}l@{}c@{}r@{}l@{}}
					\dedr{\reconfigrs} &
					\sosrule{(p, H_0, H'_0)  \trans{x?z} (p', H_1, H'_1) \quad (q, H_0, H'_0)  \trans{x!z} (q', H_1, H'_1)  }{ (p \mid q, H_0, H'_0)  \trans{\rec{x,z}} (p' \Par q', H_1, H'_1)} \\
					\dedr{\cpoloplusr}&
					\sosrule{\myquad[25]~~}{ (p \llfloor q \oplus q \llfloor p \oplus p \mid q, H_0, H'_0)  \trans{\rec{x,z}} (p' \Par q', H_1, H'_1)}
					\end{array} 
					\] 
				}
			\end{enumerate}
			
			As demonstrated in \ref{enum::a8-g} and \ref{enum::a8-h}, if $(p, H_0, H'_0)  \trans{x?z} (p', H_1, H'_1)$ and $(q, H_0, H'_0) \trans{x!z} (q',\allowbreak H_1, H'_1)$ hold then both of the terms $p \Par q$ and $p \llfloor q \oplus q \llfloor \oplus p \mid q$ converge to the same expression with the $\rec{x,z}$ transition:
			\begin{equation}
			\label{eq::a8-r4}
			\begin{alignedat}{2}
			(p \Par q, H_0, H'_0) \trans{\rec{x,z}}& (p' \Par q', H_1, H'_1)\\
			(p \llfloor q \oplus q \llfloor p \oplus p \mid q, H_0, H'_0) \trans{\rec{x,z}}& (p' \Par q', H_1, H'_1)
			\end{alignedat}
			\end{equation}
			
		}
	\end{caseof}
	Therefore, by (\ref{eq::a8-r1}), (\ref{eq::a8-r2}), (\ref{eq::a8-r2-2}), (\ref{eq::a8-r3}) and (\ref{eq::a8-r4}) it is straightforward to conclude that the following holds:
	\begin{equation}
	(p \Par q) \sim (p \llfloor q \oplus q \llfloor p \oplus p \mid q)
	\end{equation}
}

\item{
Axiom under consideration:
\begin{equation}
\bot \llfloor p \equiv \bot \quad (A9)
\end{equation}
for $p \in {\DNK}$. Observe that according to the semantic rules of {\DNK}, the terms $\bot \llfloor p$ and $\bot$ do not afford any transition. Hence, the following trivially holds:
\begin{equation}
(\bot \llfloor p) \sim \bot
\end{equation}
}

\item{
	Axiom under consideration:
	\begin{equation}
	(a \Seq p) \llfloor q \equiv a \Seq (p \Par q) \quad (A10)
	\end{equation}
	for $a \in \{z, x?z, x!z, \recp{x,z} \}$, $z \in {\NetKATnoDup}$ and $p, q \in {\DNK}$. In the following, we make a case analysis on the shape of $a$ and show that the terms $(a \Seq p) \llfloor q$ and $a \Seq (p \Par q)$ are bisimilar.
	\begin{caseof}
		\case{1}{$a \triangleq z$}{
			
			Consider an arbitrary but fixed network packet $\sigma$, let $S_z \triangleq \llbracket z \rrbracket(\sigma \!\!::\!\! \langle \rangle)$. The derivations of $(z \Seq p) \llfloor q$ are as follows:
			\begin{enumerate}[leftmargin=.5in,label=(\alph*)]
				\item{\label{enum::a10-a}
					\[
					\begin{array}{r@{}r@{}l@{}c@{}r@{}l@{}c@{}r@{}l@{}}
					\textnormal{For all } \sigma' \in S_{z}:\quad\quad & \dedr{\cpolseqsucc}&\sosrule{}{ (z \Seq p, \sigma :: H, H')  \trans{(\sigma, \sigma')} ( p, H, \sigma' :: H')}\\
					& \dedr{\llfloor} & \sosrule{}{ ((z \Seq p) \llfloor q, \sigma :: H, H')  \trans{(\sigma, \sigma')} ( p \Par q, H, \sigma' :: H)}
					\end{array} 
					\] 
				}
			\end{enumerate}
			The derivations of $z \Seq (p \Par q)$ are as follows:
			\begin{enumerate}[leftmargin=.5in,label=(\alph*)]
				\setcounter{enumi}{1}
				\item{\label{enum::a10-b}
					\[
					\begin{array}{r@{}l@{}c@{}r@{}l@{}c@{}r@{}l@{}}
					\textnormal{For all } \sigma' \in S_{z}: \quad \dedr{\cpolseqsucc}
					\sosrule{}{ (z \Seq (p \Par q), \sigma :: H, H')  \trans{(\sigma, \sigma')} (p \Par q, H, \sigma' :: H')}
					\end{array} 
					\] 
				}
			\end{enumerate}
						
			As demonstrated in \ref{enum::a10-a} and \ref{enum::a10-b}, both of the terms $(z ; p) \llfloor q$ and $z \Seq (p \Par q)$ initially afford the same set of transitions of shape $(\sigma, \sigma')$ and they converge to the same expression after taking these transitions:
			\begin{equation}
			\label{eq::a10-r1}
			\begin{alignedat}{2}
			((z ; p) \llfloor q, \sigma :: H, H') \trans{(\sigma, \sigma')}& (p \Par q, H, \sigma' :: H') \\
			(z \Seq (p \Par q), \sigma :: H, H') \trans{(\sigma, \sigma')}& (p \Par q, H, \sigma' :: H')
			\end{alignedat}
			\end{equation}
		}
	
	\case{2}{$a \triangleq x?z$}{
		
		The derivations of $(x?z \Seq p) \llfloor q$ are as follows:
		\begin{enumerate}[leftmargin=.5in,label=(\alph*)]
			\setcounter{enumi}{2}
			\item{\label{enum::a10-c}
				\[
				\begin{array}{r@{}l@{}c@{}r@{}l@{}c@{}r@{}l@{}} \dedr{\cpolmsgrec}&\sosrule{}{ (x?z \Seq p, H, H')  \trans{x?z} ( p, H, H')}\\
				\dedr{\llfloor} & \sosrule{}{ ((x?z \Seq p) \llfloor q, H, H')  \trans{x?z} ( p \Par q, H, H')}
				\end{array} 
				\] 
			}
		\end{enumerate}
		The derivations of $x ? z \Seq (p \Par q)$ are as follows:
		\begin{enumerate}[leftmargin=.5in,label=(\alph*)]
			\setcounter{enumi}{3}
			\item{\label{enum::a10-d}
				\[
				\begin{array}{r@{}l@{}c@{}r@{}l@{}c@{}r@{}l@{}} 
				\dedr{\cpolmsgrec}
				\sosrule{}{ (x?z \Seq (p \Par q),  H, H')  \trans{x?z} (p \Par q, H, H')}
				\end{array} 
				\] 
			}
		\end{enumerate}
		
		As demonstrated in \ref{enum::a10-c} and \ref{enum::a10-d}, both of the terms $(x ? z ; p) \llfloor q$ and $x ? z \Seq (p \Par q)$ initially only afford the $x?z$ transition and they converge to the same expression after taking this transition:
		\begin{equation}
		\label{eq::a10-r2}
		\begin{alignedat}{2}
		((x ? z ; p) \llfloor q, H, H') \trans{x?z}& (p \Par q, H, H') \\
		(x ? z \Seq (p \Par q), H, H') \trans{x?z}& (p \Par q, H, H')
		\end{alignedat}
		\end{equation}
	}

	\case{3}{$a \triangleq x!z$}{
	
	The derivations of $(x!z \Seq p) \llfloor q$ are as follows:
	\begin{enumerate}[leftmargin=.5in,label=(\alph*)]
		\setcounter{enumi}{4}
		\item{\label{enum::a10-e}
			\[
			\begin{array}{r@{}l@{}c@{}r@{}l@{}c@{}r@{}l@{}} \dedr{\cpolmsgsend}&\sosrule{}{ (x!z \Seq p, H, H')  \trans{x!z} ( p, H, H')}\\
			\dedr{\llfloor} & \sosrule{}{ ((x!z \Seq p) \llfloor q, H, H')  \trans{x!z} ( p \Par q, H, H')}
			\end{array} 
			\] 
		}
	\end{enumerate}
	The derivations of $x ! z \Seq (p \Par q)$ are as follows:
	\begin{enumerate}[leftmargin=.5in,label=(\alph*)]
		\setcounter{enumi}{5}
		\item{\label{enum::a10-f}
			\[
			\begin{array}{r@{}l@{}c@{}r@{}l@{}c@{}r@{}l@{}} 
			\dedr{\cpolmsgsend}
			\sosrule{}{ (x!z \Seq (p \Par q),  H, H')  \trans{x!z} (p \Par q, H, H')}
			\end{array} 
			\] 
		}
	\end{enumerate}
	
	As demonstrated in \ref{enum::a10-e} and \ref{enum::a10-f}, both of the terms $(x ! z ; p) \llfloor q$ and $x ! z \Seq (p \Par q)$ initially only afford the $x!z$ transition and they converge to the same expression after taking this transition:
	\begin{equation}
	\label{eq::a10-r3}
	\begin{alignedat}{2}
	((x ! z ; p) \llfloor q, H, H') \trans{x!z}& (p \Par q, H, H') \\
	(x ! z \Seq (p \Par q), H, H') \trans{x!z}& (p \Par q, H, H')
	\end{alignedat}
	\end{equation}
}

	\case{4}{$a \triangleq \recp{x,z}$}{
	
	The derivations of $(\recp{x,z} \Seq p) \llfloor q$ are as follows:
	\begin{enumerate}[leftmargin=.5in,label=(\alph*)]
		\setcounter{enumi}{6}
		\item{\label{enum::a10-g}
			\[
			\begin{array}{r@{}l@{}c@{}r@{}l@{}c@{}r@{}l@{}} \dedr{\recp{x,z}}&\sosrule{}{ (\recp{x,z} \Seq p, H, H')  \trans{\rec{x,z}} ( p, H, H')}\\
			\dedr{\llfloor} & \sosrule{}{ ((\recp{x,z} \Seq p) \llfloor q, H, H')  \trans{\rec{x,z}} ( p \Par q, H, H')}
			\end{array} 
			\] 
		}
	\end{enumerate}
	The derivations of $\recp{x, z} \Seq (p \Par q)$ are as follows:
	\begin{enumerate}[leftmargin=.5in,label=(\alph*)]
		\setcounter{enumi}{7}
		\item{\label{enum::a10-h}
			\[
			\begin{array}{r@{}l@{}c@{}r@{}l@{}c@{}r@{}l@{}} 
			\dedr{\recp{x,z}}
			\sosrule{}{ (\recp{x,z} \Seq (p \Par q),  H, H')  \trans{\rec{x,z}} (p \Par q, H, H')}
			\end{array} 
			\] 
		}
	\end{enumerate}
	
	As demonstrated in \ref{enum::a10-g} and \ref{enum::a10-h}, both of the terms $(\recp{x,z} ; p) \llfloor q$ and $\recp{x,z} \Seq (p \Par q)$ initially only afford the $\rec{x,z}$ transition and they converge to the same expression after taking this transition:
	\begin{equation}
	\label{eq::a10-r4}
	\begin{alignedat}{2}
	((\recp{x,z} ; p) \llfloor q, H, H') \trans{\rec{x,z}}& (p \Par q, H, H') \\
	(\recp{x,z} \Seq (p \Par q), H, H') \trans{\rec{x,z}}& (p \Par q, H, H')
	\end{alignedat}
	\end{equation}
}

	\end{caseof}
	Therefore, by (\ref{eq::a10-r1}), (\ref{eq::a10-r2}), (\ref{eq::a10-r3}) and (\ref{eq::a10-r4}) it is straightforward to conclude that the following holds:
	\begin{equation}
	((a ; p) \llfloor q) \sim (a \Seq (p \Par q))
	\end{equation}
}

	\item{
	Axiom under consideration:
	\begin{equation}
	(p \oplus q) \llfloor r \equiv (p \llfloor r) \oplus (q \llfloor r) \quad (A11)
	\end{equation}
	for $p, q, r \in {\DNK}$. According to the semantic rules of {\DNK}, the following are the possible transitions that can initially occur in the terms $(p \oplus q) \llfloor r$ and $(p \llfloor r) \oplus (q \llfloor r)$:
	\begin{gather*}
	\begin{cases}
	(1)\;(p, H_0, H'_0)  \trans{\gamma} (p', H_1, H'_1)\\
	(2)\;(q, H_0, H'_0)  \trans{\gamma} (q', H_1, H'_1)
	\end{cases}
	\end{gather*}
	$\gamma ::= (\sigma, \sigma') \mid x!z \mid x?z \mid {\rec{x,z}}$
	\begin{caseof}
		\case{1}{$(p, H_0, H'_0)  \trans{\gamma} (p', H_1, H'_1)$}{
			
			The derivations of $(p \oplus q) \llfloor r$ are as follows:
			\begin{enumerate}[leftmargin=.5in,label=(\alph*)]
				\item{\label{enum::a11-a}
					\[
					\begin{array}{r@{}l@{}c@{}r@{}l@{}c@{}r@{}l@{}}
					\dedr{\cpoloplusl}&\sosrule{(p, H_0, H'_0)  \trans{\gamma} (p', H_1, H'_1)  }{ (p \oplus q, H_0, H'_0)  \trans{\gamma} ( p', H_1, H'_1)}\\
					\dedr{\llfloor}&\sosrule{}{ ((p \oplus q) \llfloor r, H_0, H'_0)  \trans{\gamma} ( p' \Par r, H_1, H'_1)}
					\end{array} 
					\] 
				}
			\end{enumerate}
			The derivations of $(p \llfloor r) \oplus (q \llfloor r)$ are as follows:
			\begin{enumerate}[leftmargin=.5in,label=(\alph*)]
				\setcounter{enumi}{1}
				\item{\label{enum::a11-b}
					\[
					\begin{array}{r@{}l@{}c@{}r@{}l@{}c@{}r@{}l@{}}
					\dedr{\llfloor}&
					\sosrule{(p, H_0, H'_0)  \trans{\gamma} (p', H_1, H'_1)  }{ (p \llfloor r, H_0, H'_0)  \trans{\gamma} (p' \Par r, H_1, H'_1)}\\
					\dedr{\cpoloplusl}&
					\sosrule{}{ ((p \llfloor r) \oplus (q \llfloor r), H_0, H'_0)  \trans{\gamma} (p' \Par r, H_1, H'_1)}
					\end{array} 
					\] 
				}
			\end{enumerate}
			
			As demonstrated in \ref{enum::a11-a} and \ref{enum::a11-b}, if $(p, H_0, H'_0) \trans{\gamma} (p', H_1, H'_1)$ holds then both of the terms $(p \oplus q) \llfloor r $ and $(p \llfloor r) \oplus (q \llfloor r)$ converge to the same expression with the $\gamma$ transition:
			\begin{equation}
			\label{eq::a11-r1}
			\begin{alignedat}{2}
			((p \oplus q) \llfloor r, H_0, H'_0) \trans{\gamma}& (p' \Par r, H_1, H'_1)\\
			((p \llfloor r) \oplus (q \llfloor r), H_0, H'_0) \trans{\gamma}& (p' \Par r, H_1, H'_1)
			\end{alignedat}
			\end{equation}
		}
	
		\case{2}{$(q, H_0, H'_0)  \trans{\gamma} (q', H_1, H'_1)$}{
		
		The derivations of $(p \oplus q) \llfloor r$ are as follows:
		\begin{enumerate}[leftmargin=.5in,label=(\alph*)]
			\setcounter{enumi}{2}
			\item{\label{enum::a11-c}
				\[
				\begin{array}{r@{}l@{}c@{}r@{}l@{}c@{}r@{}l@{}}
				\dedr{\cpoloplusr}&\sosrule{(q, H_0, H'_0)  \trans{\gamma} (q', H_1, H'_1)  }{ (p \oplus q, H_0, H'_0)  \trans{\gamma} ( q', H_1, H'_1)}\\
				\dedr{\llfloor}&\sosrule{}{ ((p \oplus q) \llfloor r, H_0, H'_0)  \trans{\gamma} (q' \Par r, H_1, H'_1)}
				\end{array} 
				\] 
			}
		\end{enumerate}
		The derivations of $(p \llfloor r) \oplus (q \llfloor r)$ are as follows:
		\begin{enumerate}[leftmargin=.5in,label=(\alph*)]
			\setcounter{enumi}{3}
			\item{\label{enum::a11-d}
				\[
				\begin{array}{r@{}l@{}c@{}r@{}l@{}c@{}r@{}l@{}}
				\dedr{\llfloor}&
				\sosrule{(q, H_0, H'_0)  \trans{\gamma} (q', H_1, H'_1)  }{ (q \llfloor r, H_0, H'_0)  \trans{\gamma} (q' \Par r, H_1, H'_1)}\\
				\dedr{\cpoloplusr}&
				\sosrule{}{ ((p \llfloor r) \oplus (q \llfloor r), H_0, H'_0)  \trans{\gamma} (q' \Par r, H_1, H'_1)}
				\end{array} 
				\] 
			}
		\end{enumerate}
		
		As demonstrated in \ref{enum::a11-c} and \ref{enum::a11-d}, if $(q, H_0, H'_0) \trans{\gamma} (q', H_1, H'_1)$ holds then both of the terms $(p \oplus q) \llfloor r $ and $(p \llfloor r) \oplus (q \llfloor r)$ converge to the same expression with the $\gamma$ transition:
		\begin{equation}
		\label{eq::a11-r2}
		\begin{alignedat}{2}
		((p \oplus q) \llfloor r, H_0, H'_0) \trans{\gamma}& (q' \Par r, H_1, H'_1)\\
		((p \llfloor r) \oplus (q \llfloor r), H_0, H'_0) \trans{\gamma}& (q' \Par r, H_1, H'_1)
		\end{alignedat}
		\end{equation}
	}
	\end{caseof}
	Therefore, by (\ref{eq::a11-r1}) and (\ref{eq::a11-r2}) it is straightforward to conclude that the following holds:
	\begin{equation}
	((p \oplus q) \llfloor r) \sim ((p \llfloor r) \oplus (q \llfloor r))
	\end{equation}
}

	\item{
	Axiom under consideration:
	\begin{equation}
	(p \oplus q) \mid r \equiv (p \mid r) \oplus (q \mid r) \quad (A13)
	\end{equation}
	for $p, q, r \in {\DNK}$. According to the semantic rules of {\DNK}, the following are the possible transitions that can initially occur in the terms $(p \oplus q) \mid r$ and $(p \mid r) \oplus (q \mid r)$:
	\begin{gather*}
	\begin{cases}
	(1)\;(p, H, H')  \trans{x!z} (p', H, H')\quad &(r, H, H')  \trans{x?z} (r', H, H')\\
	(2)\;(p, H, H')  \trans{x?z} (p', H, H')\quad &(r, H, H')  \trans{x!z} (r', H, H')\\
	(3)\;(q, H, H')  \trans{x!z} (q', H, H')\quad &(r, H, H')  \trans{x?z} (r', H, H')\\
	(4)\;(q, H, H')  \trans{x?z} (q', H, H')\quad &(r, H, H')  \trans{x!z} (r', H, H')\\
	\end{cases}
	\end{gather*}
	\begin{caseof}
		\case{1}{$(p, H, H')  \trans{x!z} (p', H, H')\quad (r, H, H')  \trans{x?z} (r', H, H')$}{
			
			The derivations of $(p \oplus q) \mid r$ are as follows:
			\begin{enumerate}[leftmargin=.5in,label=(\alph*)]
				\item{\label{enum::a13-a}
					\[
					\begin{array}{r@{}l@{}c@{}r@{}l@{}c@{}r@{}l@{}}
					\dedr{\cpoloplusl}&\sosrule{(p, H, H')  \trans{x!z} (p', H, H')  }{ (p \oplus q, H, H')  \trans{x!z} ( p', H, H')} \quad \quad & & \sosrule{}{(r, H, H')  \trans{x?z} (r', H, H')} \\
					\dedr{\mid^{!?}}&\multicolumn{3}{@{}l@{}}{\sosrule{\myquad[28]}{ (p \oplus q) \mid r, H, H')  \trans{\rec{x, z}} ( p' \Par r', H, H')}}
					\end{array} 
					\] 
				}
			\end{enumerate}
			The derivations of $(p \mid r) \oplus (q \mid r)$ are as follows:
			\begin{enumerate}[leftmargin=.5in,label=(\alph*)]
				\setcounter{enumi}{1}
				\item{\label{enum::a13-b}
					\[
					\begin{array}{r@{}l@{}c@{}r@{}l@{}c@{}r@{}l@{}}
					\dedr{\mid^{!?}}&
					\sosrule{(p, H, H')  \trans{x!z} (p', H, H') \quad\quad (r, H, H')  \trans{x?z} (r', H, H')}{ (p \mid r, H, H')  \trans{\rec{x, z}} (p' \Par r', H, H')}\\
					\dedr{\cpoloplusl}&\sosrule{\myquad[24]~}{ ((p \mid r) \oplus (q \mid r), H, H')  \trans{\rec{x, z}} ( p' \Par r', H, H')}
					\end{array} 
					\] 
				}
			\end{enumerate}
			
			As demonstrated in \ref{enum::a13-a} and \ref{enum::a13-b}, if $(p, H, H')  \trans{x!z} (p', H, H')$ and $(r, H, H') \trans{x!z} (r',\allowbreak H, H')$ hold then both of the terms $(p \oplus q) \mid r$ and $(p \mid r) \oplus (q \mid r)$ converge to the same expression with the $\rec{x, z}$ transition:
			\begin{equation}
			\label{eq::a13-r1}
			\begin{alignedat}{2}
			((p \oplus q) \mid r, H, H') \trans{\rec{x, z}}& (p' \Par r', H, H')\\
			((p \mid r) \oplus (q \mid r), H, H') \trans{\rec{x, z}}& (p' \Par r', H, H')
			\end{alignedat}
			\end{equation}
		}
	
		\case{2}{$(p, H, H')  \trans{x?z} (p', H, H')\quad (r, H, H')  \trans{x!z} (r', H, H')$}{
		
		The derivations of $(p \oplus q) \mid r$ are as follows:
		\begin{enumerate}[leftmargin=.5in,label=(\alph*)]
			\setcounter{enumi}{2}
			\item{\label{enum::a13-c}
				\[
				\begin{array}{r@{}l@{}c@{}r@{}l@{}c@{}r@{}l@{}}
				\dedr{\cpoloplusl}&\sosrule{(p, H, H')  \trans{x?z} (p', H, H')  }{ (p \oplus q, H, H')  \trans{x?z} ( p', H, H')} \quad \quad & & \sosrule{}{(r, H, H')  \trans{x!z} (r', H, H')} \\
				\dedr{\mid^{?!}}&\multicolumn{3}{@{}l@{}}{\sosrule{\myquad[28]}{ (p \oplus q) \mid r, H, H')  \trans{\rec{x, z}} ( p' \Par r', H, H')}}
				\end{array} 
				\] 
			}
		\end{enumerate}
		The derivations of $(p \mid r) \oplus (q \mid r)$ are as follows:
		\begin{enumerate}[leftmargin=.5in,label=(\alph*)]
			\setcounter{enumi}{3}
			\item{\label{enum::a13-d}
				\[
				\begin{array}{r@{}l@{}c@{}r@{}l@{}c@{}r@{}l@{}}
				\dedr{\mid^{?!}}&
				\sosrule{(p, H, H')  \trans{x?z} (p', H, H') \quad\quad (r, H, H')  \trans{x!z} (r', H, H')}{ (p \mid r, H, H')  \trans{\rec{x, z}} (p' \Par r', H, H')}\\
				\dedr{\cpoloplusl}&\sosrule{\myquad[24]~}{ ((p \mid r) \oplus (q \mid r), H, H')  \trans{\rec{x, z}} ( p' \Par r', H, H')}
				\end{array} 
				\] 
			}
		\end{enumerate}
		
		As demonstrated in \ref{enum::a13-c} and \ref{enum::a13-d}, if $(p, H, H')  \trans{x!z} (p', H, H')$ and $(r, H, H') \trans{x?z} (r',\allowbreak H, H')$ hold then both of the terms $(p \oplus q) \mid r$ and $(p \mid r) \oplus (q \mid r)$ converge to the same expression with the $\rec{x, z}$ transition:
		\begin{equation}
		\label{eq::a13-r2}
		\begin{alignedat}{2}
		((p \oplus q) \mid r, H, H') \trans{\rec{x, z}}& (p' \Par r', H, H')\\
		((p \mid r) \oplus (q \mid r), H, H') \trans{\rec{x, z}}& (p' \Par r', H, H')
		\end{alignedat}
		\end{equation}
	}

		\case{3}{$(q, H, H')  \trans{x!z} (q', H, H')\quad (r, H, H')  \trans{x?z} (r', H, H')$}{
	
	The derivations of $(p \oplus q) \mid r$ are as follows:
	\begin{enumerate}[leftmargin=.5in,label=(\alph*)]
		\setcounter{enumi}{4}
		\item{\label{enum::a13-e}
			\[
			\begin{array}{r@{}l@{}c@{}r@{}l@{}c@{}r@{}l@{}}
			\dedr{\cpoloplusr}&\sosrule{(q, H, H')  \trans{x!z} (q', H, H')  }{ (p \oplus q, H, H')  \trans{x!z} ( q', H, H')} \quad \quad & & \sosrule{}{(r, H, H')  \trans{x?z} (r', H, H')} \\
			\dedr{\mid^{!?}}&\multicolumn{3}{@{}l@{}}{\sosrule{\myquad[28]}{ (p \oplus q) \mid r, H, H')  \trans{\rec{x, z}} ( q' \Par r', H, H')}}
			\end{array} 
			\] 
		}
	\end{enumerate}
	The derivations of $(p \mid r) \oplus (q \mid r)$ are as follows:
	\begin{enumerate}[leftmargin=.5in,label=(\alph*)]
		\setcounter{enumi}{5}
		\item{\label{enum::a13-f}
			\[
			\begin{array}{r@{}l@{}c@{}r@{}l@{}c@{}r@{}l@{}}
			\dedr{\mid^{!?}}&
			\sosrule{(q, H, H')  \trans{x!z} (q', H, H') \quad\quad (r, H, H')  \trans{x?z} (r', H, H')}{ (q \mid r, H, H')  \trans{\rec{x, z}} (q' \Par r', H, H')}\\
			\dedr{\cpoloplusr}&\sosrule{\myquad[24]~}{ ((p \mid r) \oplus (q \mid r), H, H')  \trans{\rec{x, z}} ( q' \Par r', H, H')}
			\end{array} 
			\] 
		}
	\end{enumerate}
	
	As demonstrated in \ref{enum::a13-e} and \ref{enum::a13-f}, if $(q, H, H')  \trans{x!z} (q', H, H')$ and $(r, H, H') \trans{x!z} (r',\allowbreak H, H')$ hold then both of the terms $(p \oplus q) \mid r$ and $(p \mid r) \oplus (q \mid r)$ converge to the same expression with the $\rec{x, z}$ transition:
	\begin{equation}
	\label{eq::a13-r3}
	\begin{alignedat}{2}
	((p \oplus q) \mid r, H, H') \trans{\rec{x, z}}& (q' \Par r', H, H')\\
	((p \mid r) \oplus (q \mid r), H, H') \trans{\rec{x, z}}& (q' \Par r', H, H')
	\end{alignedat}
	\end{equation}
}

\case{4}{$(q, H, H')  \trans{x?z} (q', H, H')\quad (r, H, H')  \trans{x!z} (r', H, H')$}{
	
	The derivations of $(p \oplus q) \mid r$ are as follows:
	\begin{enumerate}[leftmargin=.5in,label=(\alph*)]
		\setcounter{enumi}{6}
		\item{\label{enum::a13-g}
			\[
			\begin{array}{r@{}l@{}c@{}r@{}l@{}c@{}r@{}l@{}}
			\dedr{\cpoloplusr}&\sosrule{(q, H, H')  \trans{x?z} (q', H, H')  }{ (p \oplus q, H, H')  \trans{x?z} ( q', H, H')} \quad \quad & & \sosrule{}{(r, H, H')  \trans{x!z} (r', H, H')} \\
			\dedr{\mid^{?!}}&\multicolumn{3}{@{}l@{}}{\sosrule{\myquad[28]}{ (p \oplus q) \mid r, H, H')  \trans{\rec{x, z}} ( p' \Par r', H, H')}}
			\end{array} 
			\] 
		}
	\end{enumerate}
	The derivations of $(p \mid r) \oplus (q \mid r)$ are as follows:
	\begin{enumerate}[leftmargin=.5in,label=(\alph*)]
		\setcounter{enumi}{7}
		\item{\label{enum::a13-h}
			\[
			\begin{array}{r@{}l@{}c@{}r@{}l@{}c@{}r@{}l@{}}
			\dedr{\mid^{?!}}&
			\sosrule{(q, H, H')  \trans{x?z} (q', H, H') \quad\quad (r, H, H')  \trans{x!z} (r', H, H')}{ (q \mid r, H, H')  \trans{\rec{x, z}} (q' \Par r', H, H')}\\
			\dedr{\cpoloplusr}&\sosrule{\myquad[24]~}{ ((p \mid r) \oplus (q \mid r), H, H')  \trans{\rec{x, z}} ( q' \Par r', H, H')}
			\end{array} 
			\] 
		}
	\end{enumerate}
	
	As demonstrated in \ref{enum::a13-g} and \ref{enum::a13-h}, if $(q, H, H')  \trans{x?z} (q', H, H')$ and $(r, H, H') \trans{x!z} (r',\allowbreak H, H')$ hold then both of the terms $(p \oplus q) \mid r$ and $(p \mid r) \oplus (q \mid r)$ converge to the same expression with the $\rec{x, z}$ transition:
	\begin{equation}
	\label{eq::a13-r4}
	\begin{alignedat}{2}
	((p \oplus q) \mid r, H, H') \trans{\rec{x, z}}& (q' \Par r', H, H')\\
	((p \mid r) \oplus (q \mid r), H, H') \trans{\rec{x, z}}& (q' \Par r', H, H')
	\end{alignedat}
	\end{equation}
}
	
	\end{caseof}
	
	Therefore, by (\ref{eq::a13-r1}), (\ref{eq::a13-r2}), (\ref{eq::a13-r3}) and (\ref{eq::a13-r4}) it is straightforward to conclude that the following holds:
	\begin{equation}
	((p \oplus q) \mid r) \sim ((p \mid r) \oplus (q \mid r))
	\end{equation}
}

	\item{
	Axiom under consideration:
	\begin{equation}
	p \mid q \equiv q \mid p \quad (A14)
	\end{equation}
	for $p, q \in {\DNK}$. According to the semantic rules of {\DNK}, the following are the possible transitions that can initially occur in the terms $p \mid q$ and $q \mid p$:
	\begin{gather*}
	\begin{cases}
	(1)\;(p, H, H')  \trans{x!z} (p', H, H')\quad &(q, H, H')  \trans{x?z} (q', H, H')\\
	(2)\;(p, H, H')  \trans{x?z} (p', H, H')\quad &(q, H, H')  \trans{x!z} (q', H, H')
	\end{cases}
	\end{gather*}
	\begin{caseof}
		\case{1}{$(p, H, H')  \trans{x!z} (p', H, H')\quad (q, H, H')  \trans{x?z} (q', H, H')$}{
			
			The derivations of $p \mid q$ are as follows:
			\begin{enumerate}[leftmargin=.5in,label=(\alph*)]
				\item{\label{enum::a14-a}
					\[
					\begin{array}{r@{}l@{}c@{}r@{}l@{}c@{}r@{}l@{}}
					\dedr{\mid^{!?}}\sosrule{(p, H, H')  \trans{x!z} (p', H, H') \quad\quad (q, H, H')  \trans{x?z} (q', H, H')}{ (p \mid q, H, H')  \trans{\rec{x, z}} ( p' \Par q', H, H')}
					\end{array} 
					\] 
				}
			\end{enumerate}
			The derivations of $q \mid p$ are as follows:
			\begin{enumerate}[leftmargin=.5in,label=(\alph*)]
				\setcounter{enumi}{1}
				\item{\label{enum::a14-b}
					\[
					\begin{array}{r@{}l@{}c@{}r@{}l@{}c@{}r@{}l@{}}
					\dedr{\mid^{?!}}
					\sosrule{(q, H, H')  \trans{x?z} (q', H, H') \quad\quad (p, H, H')  \trans{x!z} (p', H, H') }{ (q \mid p, H, H')  \trans{\rec{x, z}} (q' \Par p', H, H')}
					\end{array} 
					\] 
				}
			\end{enumerate}
			
			As demonstrated in \ref{enum::a14-a} and \ref{enum::a14-b}, if $(p, H, H')  \trans{x!z} (p', H, H')$ and $(q, H, H')  \trans{x?z} (q',\allowbreak H, H')$ hold then both of the terms $p \mid q$ and $q \mid p$ are able to perform the $\rec{x, z}$ transition:
			\begin{equation}
			\label{eq::a14-r1}
			\begin{alignedat}{2}
			(p \mid q, H, H') \trans{\rec{x, z}}& (p' \Par q' , H, H')\\
			(q \mid p, H, H') \trans{\rec{x, z}}& (q' \Par p', H, H')
			\end{alignedat}
			\end{equation}
			Observe that the terms evolve into different expressions and we would now need to check if these terms are bisimilar. According to the axiom $(A6)$, the ``$\Par$'' operator is commutative. Hence, the following holds:
			\begin{equation}\label{eq::a14-r1-2}
			(p' \Par q') \sim (q' \Par p')
			\end{equation} 
		}
	
		\case{2}{$(p, H, H')  \trans{x?z} (p', H, H')\quad (q, H, H')  \trans{x!z} (q', H, H')$}{
			
			The derivations of $p \mid q$ are as follows:
			\begin{enumerate}[leftmargin=.5in,label=(\alph*)]
				\setcounter{enumi}{2}
				\item{\label{enum::a14-c}
			\[
					\begin{array}{r@{}l@{}c@{}r@{}l@{}c@{}r@{}l@{}}
					\dedr{\mid^{?!}}\sosrule{(p, H, H')  \trans{x?z} (p', H, H') \quad\quad (q, H, H')  \trans{x!z} (q', H, H')}{ (p \mid q, H, H')  \trans{\rec{x, z}} ( p' \Par q', H, H')}
					\end{array} 
					\] 
				}
			\end{enumerate}
			The derivations of $q \mid p$ are as follows:
			\begin{enumerate}[leftmargin=.5in,label=(\alph*)]
				\setcounter{enumi}{3}
				\item{\label{enum::a14-d}
					\[
					\begin{array}{r@{}l@{}c@{}r@{}l@{}c@{}r@{}l@{}}
					\dedr{\mid^{?!}}
					\sosrule{(q, H, H')  \trans{x!z} (q', H, H') \quad\quad (p, H, H')  \trans{x?z} (p', H, H') }{ (q \mid p, H, H')  \trans{\rec{x, z}} (q' \Par p', H, H')}
					\end{array} 
					\] 
				}
			\end{enumerate}
			
			As demonstrated in \ref{enum::a14-c} and \ref{enum::a14-d}, if $(p, H, H')  \trans{x!z} (p', H, H')$ and $(q, H, H')  \trans{x?z} (q',\allowbreak H, H')$ hold then both of the terms $p \mid q$ and $q \mid p$ are able to perform the $\rec{x, z}$ transition:
			\begin{equation}
			\label{eq::a14-r2}
			\begin{alignedat}{2}
			(p \mid q, H, H') \trans{\rec{x, z}}& (p' \Par q' , H, H')\\
			(q \mid p, H, H') \trans{\rec{x, z}}& (q' \Par p', H, H')
			\end{alignedat}
			\end{equation}
			Observe that the terms evolve into different expressions and we would now need to check if these terms are bisimilar. According to the axiom $(A6)$, the ``$\Par$'' operator is commutative. Hence, the following holds:
        	\begin{equation}\label{eq::a14-r2-2}
			(p' \Par q') \sim (q' \Par p')
			\end{equation} 
		}
	\end{caseof}
Therefore, by (\ref{eq::a14-r1}), (\ref{eq::a14-r1-2}), (\ref{eq::a14-r2}) and (\ref{eq::a14-r2-2}) it is straightforward to conclude that the following holds:
\begin{equation}
(p \mid q) \sim (q \mid p)
\end{equation}
}

	\item{
	Axiom under consideration:
	\begin{equation}
	p \mid q \equiv \bot~[owise] \quad (A15)
	\end{equation}
	for $p, q \in {\DNK}$. Observe that the $[owise]$ condition implies that $p$ cannot be of shape $x?z \Seq r$ when $q$ is of shape $x!z \Seq r'$, as otherwise the axiom $(A12)$ would become applicable (or vice versa due to commutativity of $\mid$). Furthermore, note that if $p$ or $q$ contains operators other than sequential composition ($\Seq$), that is the operators ``$\oplus$'', ``$\llfloor$'' and ``$\Par$'', then the axioms such as $(A8)$, $(A10)$ and $(A13)$ would become applicable and hence the $[owise]$ condition would not be met. The axiom $(A15)$ can be written explicitly as follows:
	\begin{alignat}{4}
	(z \Seq p) \mid q &\equiv& \bot\\
	(x?z \Seq p) \mid (x'?z' \Seq q) &\equiv&\, \bot \\
	(x!z \Seq p) \mid (x'!z' \Seq q) &\equiv& \bot \\
	(x?z \Seq p) \mid (x'!z' \Seq q) &\equiv& \bot &~ \textnormal{if}~ x \neq x'~\textnormal{or}~z\neq z'\\
	(\recp{x, z} \Seq p) \mid q &\equiv& \bot
	\end{alignat} 
	for $z, z' \in {\NetKATnoDup}$. Observe that the term $\bot$ does not afford any transition and none of the terms on the left hand side of the equivalences above afford any transition as well. Therefore, the following holds if the $[owise]$ condition is met:
	\begin{equation}
	(p \mid q) \sim \bot 
	\end{equation}
}

	\item{
	Axiom under consideration:
	\begin{equation}
	\delta_{{\mathcal{L}}}(\bot) \equiv \bot \quad (\delta_\bot)
	\end{equation}
	Observe that according to the semantic rules of {\DNK}, the terms $\delta_{{\mathcal{L}}}(\bot)$ and $\bot$ do not afford any transition. Hence, the following trivially holds:
	\begin{equation}
	(\delta_{{\mathcal{L}}}(\bot)) \sim \bot
	\end{equation}
}

	\item{
	Axiom under consideration:
	\begin{equation}
	\delta_{{\mathcal{L}}}(\at \Seq p) \equiv \at \Seq \delta_{{\mathcal{L}}}(p)~{\textnormal{if}}~ \at \not \in {\mathcal{L}} \quad (\delta_{\Seq})
	\end{equation}
		for $\at \in \{\alpha \cdot \pi, x?z, x!z, \recp{x,z} \}$, $z \in {\NetKATnoDup}$ and $p \in {\DNK}$. In the following, we make a case analysis on the shape of $\at$ and show that the terms $\delta_{{\mathcal{L}}}(\at \Seq p)$ and $\at \Seq \delta_{{\mathcal{L}}}(p)$ are bisimilar. In our analysis we always assume that the condition $\at \not \in {\mathcal{L}}$ is satisfied, as otherwise this axiom is not applicable.
	\begin{caseof}
		\case{1}{$\at \triangleq \alpha \cdot \pi$}{
			
			Consider an arbitrary but fixed network packet $\sigma$, let $S_{\alpha\pi} \triangleq \llbracket \alpha \cdot \pi \rrbracket(\sigma \!\!::\!\! \langle \rangle)$. The derivations of $\delta_{{\mathcal{L}}}((\alpha \cdot \pi) \Seq p)$ are as follows:
			\begin{enumerate}[leftmargin=.5in,label=(\alph*)]
				\item{\label{enum::d1-a}
					\[
					\begin{array}{r@{}r@{}l@{}c@{}r@{}l@{}c@{}r@{}l@{}}
					\textnormal{For all } \sigma' \in S_{\alpha\pi}:\quad\quad & \dedr{\cpolseqsucc}&\sosrule{}{ ((\alpha \cdot \pi) \Seq p, \sigma :: H, H')  \trans{(\sigma, \sigma')} ( p, H, \sigma' :: H')}\\
					& \dedr{{\delta^{}}} & \sosrule{}{ (\delta_{{\mathcal{L}}}((\alpha \cdot \pi) \Seq p), \sigma :: H, H')  \trans{(\sigma, \sigma')} ( \delta_{{\mathcal{L}}}(p), H, \sigma' :: H)}
					\end{array} 
					\] 
				}
			\end{enumerate}
			The derivations of $(\alpha \cdot \pi) \Seq \delta_{{\mathcal{L}}}(p)$ are as follows:
			\begin{enumerate}[leftmargin=.5in,label=(\alph*)]
				\setcounter{enumi}{1}
				\item{\label{enum::d1-b}
					\[
					\begin{array}{r@{}l@{}c@{}r@{}l@{}c@{}r@{}l@{}}
					\textnormal{For all } \sigma' \in S_{\alpha \pi}: \quad \dedr{\cpolseqsucc}
					\sosrule{}{ ((\alpha \cdot \pi) \Seq \delta_{{\mathcal{L}}}(p), \sigma :: H, H')  \trans{(\sigma, \sigma')} (\delta_{{\mathcal{L}}}(p), H, \sigma' :: H')}
					\end{array} 
					\] 
				}
			\end{enumerate}
			
			As demonstrated in \ref{enum::d1-a} and \ref{enum::d1-b}, both of the terms $\delta_{{\mathcal{L}}}((\alpha \cdot \pi) \Seq p)$ and $(\alpha \cdot \pi) \Seq \delta_{{\mathcal{L}}}(p)$ initially afford the same set of transitions of shape $(\sigma, \sigma')$ and they converge to the same expression after taking these transitions:
			\begin{equation}
			\label{eq::d1-r1}
			\begin{alignedat}{2}
			(\delta_{{\mathcal{L}}}((\alpha \cdot \pi) \Seq p), \sigma :: H, H') \trans{(\sigma, \sigma')}& (\delta_{{\mathcal{L}}}(p), H, \sigma' :: H') \\
			((\alpha \cdot \pi) \Seq \delta_{{\mathcal{L}}}(p), \sigma :: H, H') \trans{(\sigma, \sigma')}& (\delta_{{\mathcal{L}}}(p), H, \sigma' :: H')
			\end{alignedat}
			\end{equation}
		}
	
		\case{2}{$\at \triangleq x?z$}{
		
		The derivations of $\delta_{{\mathcal{L}}}(x?z \Seq p)$ are as follows:
		\begin{enumerate}[leftmargin=.5in,label=(\alph*)]
			\setcounter{enumi}{2}
			\item{\label{enum::d1-c}
				\[
				\begin{array}{r@{}l@{}c@{}r@{}l@{}c@{}r@{}l@{}} \dedr{\cpolmsgrec}&\sosrule{}{ (x?z \Seq p, H, H')  \trans{x?z} ( p, H, H')}\\
				\dedr{{\delta^{}}} & \sosrule{}{ (\delta_{{\mathcal{L}}}(x?z \Seq p), H, H')  \trans{x?z} ( \delta_{{\mathcal{L}}}(p), H,  H')}
				\end{array} 
				\] 
			}
		\end{enumerate}
		The derivations of $x?z \Seq \delta_{{\mathcal{L}}}(p)$ are as follows:
		\begin{enumerate}[leftmargin=.5in,label=(\alph*)]
			\setcounter{enumi}{3}
			\item{\label{enum::d1-d}
				\[
				\begin{array}{l@{}c@{}r@{}l@{}c@{}r@{}l@{}}
				\dedr{\cpolmsgrec}
				\sosrule{}{ (x?z \Seq \delta_{{\mathcal{L}}}(p), H, H')  \trans{x?z} (\delta_{{\mathcal{L}}}(p), H,  H')}
				\end{array} 
				\] 
			}
		\end{enumerate}
		
		As demonstrated in \ref{enum::d1-c} and \ref{enum::d1-d}, both of the terms $\delta_{{\mathcal{L}}}(x?z \Seq p)$ and $x?z \Seq \delta_{{\mathcal{L}}}(p)$ initially only afford the $x?z$ transition and they converge to the same expression after taking this transition:
		\begin{equation}
		\label{eq::d1-r2}
		\begin{alignedat}{2}
		(\delta_{{\mathcal{L}}}(x?z \Seq p), H, H') \trans{x?z}& (\delta_{{\mathcal{L}}}(p), H, H') \\
		(x?z \Seq \delta_{{\mathcal{L}}}(p), H, H') \trans{x?z}& (\delta_{{\mathcal{L}}}(p), H, H')
		\end{alignedat}
		\end{equation}
	}

	\case{3}{$\at \triangleq x!z$}{
	
	The derivations of $\delta_{{\mathcal{L}}}(x!z \Seq p)$ are as follows:
	\begin{enumerate}[leftmargin=.5in,label=(\alph*)]
		\setcounter{enumi}{4}
		\item{\label{enum::d1-e}
			\[
			\begin{array}{r@{}l@{}c@{}r@{}l@{}c@{}r@{}l@{}} \dedr{\cpolmsgsend}&\sosrule{}{ (x!z \Seq p, H, H')  \trans{x!z} ( p, H, H')}\\
			\dedr{{\delta^{}}} & \sosrule{}{ (\delta_{{\mathcal{L}}}(x!z \Seq p), H, H')  \trans{x!z} ( \delta_{{\mathcal{L}}}(p), H, H')}
			\end{array} 
			\] 
		}
	\end{enumerate}
	The derivations of $x!z \Seq \delta_{{\mathcal{L}}}(p)$ are as follows:
	\begin{enumerate}[leftmargin=.5in,label=(\alph*)]
		\setcounter{enumi}{5}
		\item{\label{enum::d1-f}
			\[
			\begin{array}{l@{}c@{}r@{}l@{}c@{}r@{}l@{}}
			\dedr{\cpolmsgsend}
			\sosrule{}{ (x!z \Seq \delta_{{\mathcal{L}}}(p), H, H')  \trans{x!z} (\delta_{{\mathcal{L}}}(p), H, H')}
			\end{array} 
			\] 
		}
	\end{enumerate}
	
	As demonstrated in \ref{enum::d1-e} and \ref{enum::d1-f}, both of the terms $\delta_{{\mathcal{L}}}(x!z \Seq p)$ and $x!z \Seq \delta_{{\mathcal{L}}}(p)$ initially only afford the $x!z$ transition and they converge to the same expression after taking this transition:
	\begin{equation}
	\label{eq::d1-r3}
	\begin{alignedat}{2}
	(\delta_{{\mathcal{L}}}(x!z \Seq p), H, H') \trans{x!z}& (\delta_{{\mathcal{L}}}(p), H, H') \\
	(x!z \Seq \delta_{{\mathcal{L}}}(p), \sigma :: H, H') \trans{x!z}& (\delta_{{\mathcal{L}}}(p), H, H')
	\end{alignedat}
	\end{equation}
}

	\case{4}{$\at \triangleq \recp{x, z}$}{
	
	The derivations of $\delta_{{\mathcal{L}}}(\recp{x, z} \Seq p)$ are as follows:
	\begin{enumerate}[leftmargin=.5in,label=(\alph*)]
		\setcounter{enumi}{6}
		\item{\label{enum::d1-g}
			\[
			\begin{array}{r@{}l@{}c@{}r@{}l@{}c@{}r@{}l@{}} 
			\dedr{\recp{x,z}}&\sosrule{}{ (\recp{x, z} \Seq p, H, H')  \trans{\rec{x, z}} ( p, H,  H')}\\
			\dedr{{\delta^{}}} & \sosrule{}{ (\delta_{{\mathcal{L}}}(\recp{x, z} \Seq p), H, H')  \trans{\rec{x, z}} ( \delta_{{\mathcal{L}}}(p), H, H')}
			\end{array} 
			\] 
		}
	\end{enumerate}
	The derivations of $\recp{x,z} \Seq \delta_{{\mathcal{L}}}(p)$ are as follows:
	\begin{enumerate}[leftmargin=.5in,label=(\alph*)]
		\setcounter{enumi}{7}
		\item{\label{enum::d1-h}
			\[
			\begin{array}{l@{}c@{}r@{}l@{}c@{}r@{}l@{}}
			\dedr{\recp{x,z}}
			\sosrule{}{ (\recp{x, z} \Seq \delta_{{\mathcal{L}}}(p), H, H')  \trans{\rec{x, z}} (\delta_{{\mathcal{L}}}(p), H, H')}
			\end{array} 
			\] 
		}
	\end{enumerate}
	
	As demonstrated in \ref{enum::d1-g} and \ref{enum::d1-h}, both of the terms $\delta_{{\mathcal{L}}}(\recp{x, z} \Seq p)$ and $\recp{x, z} \Seq \delta_{{\mathcal{L}}}(p)$ initially only afford the $\rec{x, z}$ transition and they converge to the same expression after taking this transition:
	\begin{equation}
	\label{eq::d1-r4}
	\begin{alignedat}{2}
	(\delta_{{\mathcal{L}}}(\recp{x, z} \Seq p), H, H') \trans{\rec{x, z}}& (\delta_{{\mathcal{L}}}(p), H, H') \\
	(\recp{x, z} \Seq \delta_{{\mathcal{L}}}(p), H, H') \trans{\rec{x, z}}& (\delta_{{\mathcal{L}}}(p), H, H')
	\end{alignedat}
	\end{equation}
}
	\end{caseof}

	Therefore, if $\at \not \in {\mathcal{L}}$, by (\ref{eq::d1-r1}), (\ref{eq::d1-r2}), (\ref{eq::d1-r3}) and (\ref{eq::d1-r4}) it is straightforward to conclude that the following holds:
	\begin{equation}
	(\delta_{{\mathcal{L}}}(\at \Seq p)) \sim (\at \Seq \delta_{{\mathcal{L}}}(p))
	\end{equation}
	}

	\item{
		Axiom under consideration:
		\begin{equation}
		\delta_{{\mathcal{L}}}(\at \Seq p) \equiv \bot ~{\textnormal{if}}~ \at \in {\mathcal{L}} \quad (\delta_{\Seq}^{\bot})
		\end{equation}
		Observe that according to the semantic rules of {\DNK}, the term $\bot$ do not afford any transition. Furthermore, if the condition $\at \in {\mathcal{L}}$ is satisfied, then the term $\delta_{{\mathcal{L}}}(\at \Seq p)$ also does not afford any transition. Therefore, if $\at \in {\mathcal{L}}$, the following trivially holds:
		\begin{equation}
		\delta_{{\mathcal{L}}}(\at \Seq p) \sim \bot
		\end{equation}
	}

	\item{
	Axiom under consideration:
	\begin{equation}
	\delta_{{\mathcal{L}}}(p\oplus q) \equiv \delta_{{\mathcal{L}}}(p) \oplus \delta_{{\mathcal{L}}}(q) \quad (\delta_{\oplus})
	\end{equation}
	for $p, q \in {\DNK}$. According to the semantic rules of {\DNK}, the following are the possible transitions that can initially occur in the terms $\delta_{{\mathcal{L}}}(p\oplus q)$ and $\delta_{{\mathcal{L}}}(p) \oplus \delta_{{\mathcal{L}}}(q) $:
	\begin{gather*}
	\begin{cases}
	(1)\;(p, H_0, H'_0)  \trans{\gamma} (p', H_1, H'_1)\\
	(2)\;(q, H_0, H'_0)  \trans{\gamma} (q', H_1, H'_1)
	\end{cases}
	\end{gather*}
	$\gamma ::= (\sigma, \sigma') \mid x!z \mid x?z \mid {\rec{x,z}}$
	\begin{caseof}
	\case{1}{$(p, H_0, H'_0)  \trans{\gamma} (p', H_1, H'_1)$}{
			
			The derivations of $\delta_{{\mathcal{L}}}(p\oplus q)$ are as follows:
			\begin{enumerate}[leftmargin=.5in,label=(\alph*)]
				\item{\label{enum::d2-a}
					\[
					\begin{array}{r@{}l@{}c@{}r@{}l@{}c@{}r@{}l@{}}
					\dedr{\cpoloplusl}&\sosrule{(p, H_0, H'_0)  \trans{\gamma} (p', H_1, H'_1)  }{ (p \oplus q, H_0, H'_0)  \trans{\gamma} ( p', H_1, H'_1)}\\
					\dedr{{\delta^{}}}&\sosrule{}{ (\delta_{{\mathcal{L}}}(p\oplus q), H_0, H'_0)  \trans{\gamma} ( \delta_{{\mathcal{L}}}(p'), H_1, H'_1)}
					\end{array} 
					\] 
				}
			\end{enumerate}
			The derivations of $\delta_{{\mathcal{L}}}(p) \oplus \delta_{{\mathcal{L}}}(q)$ are as follows:
			\begin{enumerate}[leftmargin=.5in,label=(\alph*)]
				\setcounter{enumi}{1}
				\item{\label{enum::d2-b}
					\[
					\begin{array}{r@{}l@{}c@{}r@{}l@{}c@{}r@{}l@{}}
					\dedr{{\delta^{}}} &
					\sosrule{(p, H_0, H'_0)  \trans{\gamma} (p', H_1, H'_1)  }{ (\delta_{{\mathcal{L}}}(p), H_0, H'_0)  \trans{\gamma} (\delta_{{\mathcal{L}}}(p'), H_1, H'_1)}\\
					\dedr{\cpoloplusl} &
					\sosrule{}{ (\delta_{{\mathcal{L}}}(p) \oplus \delta_{{\mathcal{L}}}(q), H_0, H'_0)  \trans{\gamma} (\delta_{{\mathcal{L}}}(p'), H_1, H'_1)}
					\end{array} 
					\] 
				}
			\end{enumerate}
			
			As demonstrated in \ref{enum::d2-a} and \ref{enum::d2-b}, if $(p, H_0, H'_0)  \trans{\gamma} (p', H_1, H'_1)$ holds then both of the terms $\delta_{{\mathcal{L}}}(p\oplus q)$ and $\delta_{{\mathcal{L}}}(p) \oplus \delta_{{\mathcal{L}}}(q)$ converge to the same expression with the $\gamma$ transition:
			\begin{equation}
			\label{eq::d2-r1}
			\begin{alignedat}{2}
			(\delta_{{\mathcal{L}}}(p\oplus q), H_0, H'_0) \trans{\gamma}& (\delta_{{\mathcal{L}}}(p'), H_1, H'_1)\\
			(\delta_{{\mathcal{L}}}(p) \oplus \delta_{{\mathcal{L}}}(q), H_0, H'_0) \trans{\gamma}& (\delta_{{\mathcal{L}}}(p'), H_1, H'_1)
			\end{alignedat}
			\end{equation}
		}
	
		\case{2}{$(q, H_0, H'_0)  \trans{\gamma} (q', H_1, H'_1)$}{
		
		The derivations of $\delta_{{\mathcal{L}}}(p\oplus q)$ are as follows:
		\begin{enumerate}[leftmargin=.5in,label=(\alph*)]
			\setcounter{enumi}{2}
			\item{\label{enum::d2-c}
				\[
				\begin{array}{r@{}l@{}c@{}r@{}l@{}c@{}r@{}l@{}}
				\dedr{\cpoloplusr}&\sosrule{(q, H_0, H'_0)  \trans{\gamma} (q', H_1, H'_1)  }{ (p \oplus q, H_0, H'_0)  \trans{\gamma} ( q', H_1, H'_1)}\\
				\dedr{{\delta^{}}}&\sosrule{}{ (\delta_{{\mathcal{L}}}(p\oplus q), H_0, H'_0)  \trans{\gamma} ( \delta_{{\mathcal{L}}}(q'), H_1, H'_1)}
				\end{array} 
				\] 
			}
		\end{enumerate}
		The derivations of $\delta_{{\mathcal{L}}}(p) \oplus \delta_{{\mathcal{L}}}(q)$ are as follows:
		\begin{enumerate}[leftmargin=.5in,label=(\alph*)]
			\setcounter{enumi}{3}
			\item{\label{enum::d2-d}
				\[
				\begin{array}{r@{}l@{}c@{}r@{}l@{}c@{}r@{}l@{}}
				\dedr{{\delta^{}}} &
				\sosrule{(q, H_0, H'_0)  \trans{\gamma} (q', H_1, H'_1)  }{ (\delta_{{\mathcal{L}}}(q), H_0, H'_0)  \trans{\gamma} (\delta_{{\mathcal{L}}}(q'), H_1, H'_1)}\\
				\dedr{\cpoloplusr} &
				\sosrule{}{ (\delta_{{\mathcal{L}}}(p) \oplus \delta_{{\mathcal{L}}}(q), H_0, H'_0)  \trans{\gamma} (\delta_{{\mathcal{L}}}(q'), H_1, H'_1)}
				\end{array} 
				\] 
			}
		\end{enumerate}
		
		As demonstrated in \ref{enum::d2-c} and \ref{enum::d2-d}, if $(q, H_0, H'_0)  \trans{\gamma} (q', H_1, H'_1)$ holds then both of the terms $\delta_{{\mathcal{L}}}(p\oplus q)$ and $\delta_{{\mathcal{L}}}(p) \oplus \delta_{{\mathcal{L}}}(q)$ converge to the same expression with the $\gamma$ transition:
		\begin{equation}
		\label{eq::d2-r2}
		\begin{alignedat}{2}
		(\delta_{{\mathcal{L}}}(p\oplus q), H_0, H'_0) \trans{\gamma}& (\delta_{{\mathcal{L}}}(q'), H_1, H'_1)\\
		(\delta_{{\mathcal{L}}}(p) \oplus \delta_{{\mathcal{L}}}(q), H_0, H'_0) \trans{\gamma}& (\delta_{{\mathcal{L}}}(q'), H_1, H'_1)
		\end{alignedat}
		\end{equation}
	}
	\end{caseof}
	Therefore, by (\ref{eq::d2-r1}), and (\ref{eq::d2-r2}) it is straightforward to conclude that the following holds:
	\begin{equation}
	(\delta_{{\mathcal{L}}}(p\oplus q)) \sim (\delta_{{\mathcal{L}}}(p) \oplus \delta_{{\mathcal{L}}}(q))
	\end{equation}
}

	\item{
	Axiom under consideration:
	\begin{equation}
	\pi_0(p) \equiv \bot \quad (\Pi_0)
	\end{equation}
	for $p \in {\DNK}$. Observe that according to the semantic rules of {\DNK}, the terms $\pi_0(p)$ and $\bot$ do not afford any transition. Hence, the following trivially holds:
	\begin{equation}
	\pi_0(p) \sim \bot
	\end{equation}
	}

	\item{
	Axiom under consideration:
	\begin{equation}
	\pi_{n}(\bot) \equiv \bot \quad (\Pi_\bot)
	\end{equation}
	for $n \in \mathbb{N}$.  Observe that according to the semantic rules of {\DNK}, the terms $\pi_0(\bot)$ and $\bot$ do not afford any transition. Hence, the following trivially holds:
	\begin{equation}
	\pi_{n}(\bot) \sim \bot
	\end{equation}
	}

	\item{
		Axiom under consideration:
		\begin{equation}
		\pi_{n+1}(\at \Seq p) \equiv \at \Seq \pi_{n}(p) \quad (\Pi_{\Seq})
		\end{equation}
		for $\at \in \{\alpha \cdot \pi, x?z, x!z, \recp{x,z} \}$, $z \in {\NetKATnoDup}$, $n \in \mathbb{N}$ and $p \in {\DNK}$. In the following, we make a case analysis on the shape of $\at$ and show that the terms $\pi_{n+1}(\at \Seq p)$ and $\at \Seq \pi_{n}(p)$ are bisimilar. 
		\begin{caseof}
			\case{1}{$\at \triangleq \alpha \cdot \pi$}{
				
				Consider an arbitrary but fixed network packet $\sigma$, let $S_{\alpha\pi} \triangleq \llbracket \alpha \cdot \pi \rrbracket(\sigma \!\!::\!\! \langle \rangle)$. The derivations of $\pi_{n+1}((\alpha \cdot \pi) \Seq p)$ are as follows:
				\begin{enumerate}[leftmargin=.5in,label=(\alph*)]
					\item{\label{enum::p1-a}
						\[
						\begin{array}{r@{}r@{}l@{}c@{}r@{}l@{}c@{}r@{}l@{}}
						\textnormal{For all } \sigma' \in S_{\alpha\pi}:\quad\quad & \dedr{\cpolseqsucc}&\sosrule{}{ ((\alpha \cdot \pi) \Seq p, \sigma :: H, H')  \trans{(\sigma, \sigma')} ( p, H, \sigma' :: H')}\\
						& \dedr{{\pi^{}}} & \sosrule{}{ (\pi_{n+1}((\alpha \cdot \pi) \Seq p), \sigma :: H, H')  \trans{(\sigma, \sigma')} ( \pi_{n}(p), H, \sigma' :: H)}
						\end{array} 
						\] 
					}
				\end{enumerate}
				The derivations of $(\alpha \cdot \pi) \Seq \pi_{n}(p)$ are as follows:
				\begin{enumerate}[leftmargin=.5in,label=(\alph*)]
					\setcounter{enumi}{1}
					\item{\label{enum::p1-b}
						\[
						\begin{array}{r@{}l@{}c@{}r@{}l@{}c@{}r@{}l@{}}
						\textnormal{For all } \sigma' \in S_{\alpha \pi}: \quad \dedr{\cpolseqsucc}
						\sosrule{}{ ((\alpha \cdot \pi) \Seq \pi_{n}(p), \sigma :: H, H')  \trans{(\sigma, \sigma')} (\pi_{n}(p), H, \sigma' :: H')}
						\end{array} 
						\] 
					}
				\end{enumerate}
				
				As demonstrated in \ref{enum::p1-a} and \ref{enum::p1-b}, both of the terms $\pi_{n+1}((\alpha \cdot \pi) \Seq p)$ and $(\alpha \cdot \pi) \Seq \pi_{n}(p)$ initially only afford the same set of transitions of shape $(\sigma, \sigma')$ and they converge to the same expression after taking these transitions:
				\begin{equation}
				\label{eq::p1-r1}
				\begin{alignedat}{2}
				(\pi_{n+1}((\alpha \cdot \pi) \Seq p), \sigma :: H, H') \trans{(\sigma, \sigma')}& (\pi_{n}(p), H, \sigma' :: H') \\
				((\alpha \cdot \pi) \Seq \pi_{n}(p), \sigma :: H, H') \trans{(\sigma, \sigma')}& (\pi_{n}(p), H, \sigma' :: H')
				\end{alignedat}
				\end{equation}
			}
			
			\case{2}{$\at \triangleq x?z$}{
				
				The derivations of $\pi_{n+1}(x?z \Seq p)$ are as follows:
				\begin{enumerate}[leftmargin=.5in,label=(\alph*)]
					\setcounter{enumi}{2}
					\item{\label{enum::p1-c}
						\[
						\begin{array}{r@{}l@{}c@{}r@{}l@{}c@{}r@{}l@{}} \dedr{\cpolmsgrec}&\sosrule{}{ (x?z \Seq p, \sigma :: H, H')  \trans{x?z} ( p, H, H')}\\
						\dedr{{\pi^{}}} & \sosrule{}{ (\pi_{n+1}(x?z \Seq p), H, H')  \trans{x?z} ( \pi_{n}(p), H, H')}
						\end{array} 
						\] 
					}
				\end{enumerate}
				The derivations of $x?z \Seq \delta_{{\mathcal{L}}}(p)$ are as follows:
				\begin{enumerate}[leftmargin=.5in,label=(\alph*)]
					\setcounter{enumi}{3}
					\item{\label{enum::p1-d}
						\[
						\begin{array}{l@{}c@{}r@{}l@{}c@{}r@{}l@{}}
						\dedr{\cpolmsgrec}
						\sosrule{}{ (x?z \Seq \pi_{n}(p), H, H')  \trans{x?z} (\pi_{n}(p), H, H')}
						\end{array} 
						\] 
					}
				\end{enumerate}
				
				As demonstrated in \ref{enum::p1-c} and \ref{enum::p1-d}, both of the terms $\pi_{n+1}(x?z \Seq p)$ and $x?z \Seq \pi_{n}(p)$ initially only afford the $x?z$ transition and they converge to the same expression after taking this transition:
				\begin{equation}
				\label{eq::p1-r2}
				\begin{alignedat}{2}
				(\pi_{n+1}(x?z \Seq p), H, H') \trans{x?z}& (\pi_{n}(p), H, H') \\
				(x?z \Seq \pi_{n}(p), H, H') \trans{x?z}& (\pi_{n}(p), H, H')
				\end{alignedat}
				\end{equation}
			}
			
			\case{3}{$\at \triangleq x!z$}{
				
				The derivations of $\pi_{n+1}(x!z \Seq p)$ are as follows:
				\begin{enumerate}[leftmargin=.5in,label=(\alph*)]
					\setcounter{enumi}{4}
					\item{\label{enum::p1-e}
						\[
						\begin{array}{r@{}l@{}c@{}r@{}l@{}c@{}r@{}l@{}} \dedr{\cpolmsgsend}&\sosrule{}{ (x!z \Seq p, H, H')  \trans{x!z} ( p, H, H')}\\
						\dedr{{\pi^{}}} & \sosrule{}{ (\pi_{n+1}(x!z \Seq p), H, H')  \trans{x!z} ( \pi_{n}(p), H, H')}
						\end{array} 
						\] 
					}
				\end{enumerate}
				The derivations of $x!z \Seq \pi_{n}(p)$ are as follows:
				\begin{enumerate}[leftmargin=.5in,label=(\alph*)]
					\setcounter{enumi}{5}
					\item{\label{enum::p1-f}
						\[
						\begin{array}{l@{}c@{}r@{}l@{}c@{}r@{}l@{}}
						\dedr{\cpolmsgsend}
						\sosrule{}{ (x!z \Seq \pi_{n}(p), H, H')  \trans{x!z} (\pi_{n}(p), H, H')}
						\end{array} 
						\] 
					}
				\end{enumerate}
				
				As demonstrated in \ref{enum::p1-e} and \ref{enum::p1-f}, both of the terms $\pi_{n+1}(x!z \Seq p)$ and $x!z \Seq \pi_{n}(p)$ initially only afford the $x!z$ transition and they converge to the same expression after taking this transition:
				\begin{equation}
				\label{eq::p1-r3}
				\begin{alignedat}{2}
				(\pi_{n+1}(x!z \Seq p), H, H') \trans{x!z}& (\pi_{n}(p), H, H') \\
				(x!z \Seq \pi_{n}(p), H, H') \trans{x!z}& (\pi_{n}(p), H, H')
				\end{alignedat}
				\end{equation}
			}
			
			\case{4}{$\at \triangleq \recp{x, z}$}{
				
				The derivations of $\pi_{n+1}(\recp{x, z} \Seq p)$ are as follows:
				\begin{enumerate}[leftmargin=.5in,label=(\alph*)]
					\setcounter{enumi}{6}
					\item{\label{enum::p1-g}
						\[
						\begin{array}{r@{}l@{}c@{}r@{}l@{}c@{}r@{}l@{}} 
						\dedr{\recp{x,z}}&\sosrule{}{ (\recp{x, z} \Seq p, H, H')  \trans{\rec{x, z}} ( p, H, H')}\\
						\dedr{{\delta^{}}} & \sosrule{}{ (\pi_{n+1}(\recp{x, z} \Seq p), H, H')  \trans{\rec{x, z}} ( \pi_{n}(p), H, H')}
						\end{array} 
						\] 
					}
				\end{enumerate}
				The derivations of $\recp{x,z} \Seq \pi_{n}(p)$ are as follows:
				\begin{enumerate}[leftmargin=.5in,label=(\alph*)]
					\setcounter{enumi}{7}
					\item{\label{enum::p1-h}
						\[
						\begin{array}{l@{}c@{}r@{}l@{}c@{}r@{}l@{}}
						\dedr{\recp{x,z}}
						\sosrule{}{ (\recp{x, z} \Seq \pi_{n}(p), H, H')  \trans{\rec{x, z}} (\pi_{n}(p), H, H')}
						\end{array} 
						\] 
					}
				\end{enumerate}
				
				As demonstrated in \ref{enum::p1-g} and \ref{enum::p1-h}, both of the terms $\pi_{n+1}(\recp{x, z} \Seq p)$ and $\recp{x, z} \Seq \pi_{n}(p)$ initially only afford the $\rec{x, z}$ transition and they converge to the same expression after taking this transition:
				\begin{equation}
				\label{eq::p1-r4}
				\begin{alignedat}{2}
				(\pi_{n+1}(\recp{x, z} \Seq p), H, H') \trans{\rec{x, z}}& (\pi_{n}(p), H, H') \\
				(\recp{x, z} \Seq \pi_{n}(p), H, H') \trans{\rec{x, z}}& (\pi_{n}(p), H, H')
				\end{alignedat}
				\end{equation}
			}
		\end{caseof}
		
		Therefore, if $\at \not \in {\mathcal{L}}$, by (\ref{eq::p1-r1}), (\ref{eq::p1-r2}), (\ref{eq::p1-r3}) and (\ref{eq::p1-r4}) it is straightforward to conclude that the following holds:
		\begin{equation}
		(\pi_{n+1}(\at \Seq p)) \sim (\at \Seq \pi_{n}(p))
		\end{equation}
	}
	
	\item{
	Axiom under consideration:
	\begin{equation}
	\pi_{n}(p\oplus q) \equiv \pi_{n}(p) \oplus \pi_{n}(q) \quad (\pi{\oplus})
	\end{equation}
	for $p, q \in {\DNK}$. Observe that if $n=0$, then both of the terms do not afford any transition and bisimilarity holds trivially. If $n > 0$, according to the semantic rules of {\DNK}, the following are the possible transitions that can initially occur in the terms $\pi_{n}(p\oplus q)$ and $\pi_{n}(p) \oplus \pi_{n}(q)$:
	\begin{gather*}
	\begin{cases}
	(1)\;(p, H_0, H'_0)  \trans{\gamma} (p', H_1, H'_1)\\
	(2)\;(q, H_0, H'_0)  \trans{\gamma} (q', H_1, H'_1)
	\end{cases}
	\end{gather*}
	$\gamma ::= (\sigma, \sigma') \mid x!z \mid x?z \mid {\rec{x,z}}$
	\begin{caseof}
		\case{1}{$(p, H_0, H'_0)  \trans{\gamma} (p', H_1, H'_1)$}{
			
			The derivations of $\pi_{n}(p\oplus q)$ are as follows:
			\begin{enumerate}[leftmargin=.5in,label=(\alph*)]
				\item{\label{enum::p2-a}
					\[
					\begin{array}{r@{}l@{}c@{}r@{}l@{}c@{}r@{}l@{}}
					\dedr{\cpoloplusl}&\sosrule{(p, H_0, H'_0)  \trans{\gamma} (p', H_1, H'_1)  }{ (p \oplus q, H_0, H'_0)  \trans{\gamma} ( p', H_1, H'_1)}\\
					\dedr{{\pi^{}}}&\sosrule{}{ (\pi_{n}(p\oplus q), H_0, H'_0)  \trans{\gamma} ( \pi_{n-1}(p'), H_1, H'_1)}
					\end{array} 
					\] 
				}
			\end{enumerate}
			The derivations of $\pi_{n}(p) \oplus \pi_{n}(q)$ are as follows:
			\begin{enumerate}[leftmargin=.5in,label=(\alph*)]
				\setcounter{enumi}{1}
				\item{\label{enum::p2-b}
					\[
					\begin{array}{r@{}l@{}c@{}r@{}l@{}c@{}r@{}l@{}}
					\dedr{{\pi^{}}} &
					\sosrule{(p, H_0, H'_0)  \trans{\gamma} (p', H_1, H'_1)  }{ (\pi_{n}(p), H_0, H'_0)  \trans{\gamma} (\pi_{n-1}(p'), H_1, H'_1)}\\
					\dedr{\cpoloplusl} &
					\sosrule{}{ (\pi_{n}(p) \oplus \pi_{n}(q), H_0, H'_0)  \trans{\gamma} (\pi_{n-1}(p'), H_1, H'_1)}
					\end{array} 
					\] 
				}
			\end{enumerate}
			
			As demonstrated in \ref{enum::p2-a} and \ref{enum::p2-b}, if $(p, H_0, H'_0)  \trans{\gamma} (p', H_1, H'_1)$ holds then both of the terms $\pi_{n}(p\oplus q)$ and $\pi_{n}(p) \oplus \pi_{n}(q)$ converge to the same expression with the $\gamma$ transition:
			\begin{equation}
			\label{eq::p2-r1}
			\begin{alignedat}{2}
			(\pi_{n}(p\oplus q), H_0, H'_0) \trans{\gamma}& (\pi_{n-1}(p'), H_1, H'_1)\\
			(\pi_{n}(p) \oplus \pi_{n}(q), H_0, H'_0) \trans{\gamma}& (\pi_{n-1}(p'), H_1, H'_1)
			\end{alignedat}
			\end{equation}
		}
		
		\case{2}{$(q, H_0, H'_0)  \trans{\gamma} (q', H_1, H'_1)$}{
			
			The derivations of $\pi_{n}(p\oplus q)$ are as follows:
			\begin{enumerate}[leftmargin=.5in,label=(\alph*)]
				\setcounter{enumi}{2}
				\item{\label{enum::p2-c}
					\[
					\begin{array}{r@{}l@{}c@{}r@{}l@{}c@{}r@{}l@{}}
					\dedr{\cpoloplusr}&\sosrule{(q, H_0, H'_0)  \trans{\gamma} (q', H_1, H'_1)  }{ (p \oplus q, H_0, H'_0)  \trans{\gamma} ( q', H_1, H'_1)}\\
					\dedr{{\pi^{}}}&\sosrule{}{ (\pi_{n}(p\oplus q), H_0, H'_0)  \trans{\gamma} ( \pi_{n-1}(q'), H_1, H'_1)}
					\end{array} 
					\] 
				}
			\end{enumerate}
			The derivations of $\pi_{n}(p) \oplus \pi_{n}(q)$ are as follows:
			\begin{enumerate}[leftmargin=.5in,label=(\alph*)]
				\setcounter{enumi}{3}
				\item{\label{enum::p2-d}
					\[
					\begin{array}{r@{}l@{}c@{}r@{}l@{}c@{}r@{}l@{}}
					\dedr{{\pi^{}}} &
					\sosrule{(q, H_0, H'_0)  \trans{\gamma} (q', H_1, H'_1)  }{ (\pi_{n}(q), H_0, H'_0)  \trans{\gamma} (\pi_{n-1}(q'), H_1, H'_1)}\\
					\dedr{\cpoloplusr} &
					\sosrule{}{ (\pi_{n}(p) \oplus \pi_{n}(q), H_0, H'_0)  \trans{\gamma} (\pi_{n-1}(q'), H_1, H'_1)}
					\end{array} 
					\] 
				}
			\end{enumerate}
			
			As demonstrated in \ref{enum::p2-c} and \ref{enum::p2-d}, if $(q, H_0, H'_0)  \trans{\gamma} (q', H_1, H'_1)$ holds then both of the terms $\pi_{n}(p\oplus q)$ and $\pi_{n}(p) \oplus \pi_{n}(q)$ converge to the same expression with the $\gamma$ transition:
			\begin{equation}
			\label{eq::p2-r2}
			\begin{alignedat}{2}
			(\pi_{n}(p\oplus q), H_0, H'_0) \trans{\gamma}& (\pi_{n-1}(q'), H_1, H'_1)\\
			(\pi_{n}(p) \oplus \pi_{n}(q), H_0, H'_0) \trans{\gamma}& (\pi_{n-1}(q'), H_1, H'_1)
			\end{alignedat}
			\end{equation}
		}
	\end{caseof}
	Therefore, by (\ref{eq::p2-r1}) and (\ref{eq::p2-r2}) it is straightforward to conclude that the following holds:
	\begin{equation}
	(\pi_{n}(p\oplus q)) \sim (\pi_{n}(p) \oplus \pi_{n}(q))
	\end{equation}
}
\end{itemize}

\end{ARXIV}
\end{ARXIV}
\end{document}